\documentclass[12pt,onecolumn, margin=1in]{IEEEtran}
\usepackage{cite}
\usepackage{amsmath,amssymb,amsfonts}
\usepackage{bbm}
\usepackage{algorithmic}
\usepackage{graphicx}
\usepackage{subcaption}
\usepackage{setspace}
\usepackage{textcomp}
\usepackage{placeins}
\usepackage{booktabs}
\usepackage[margin=1in]{geometry}
\usepackage{xcolor}
\usepackage[colorlinks=true,linkcolor=blue,citecolor=blue,urlcolor=black]{hyperref}
\usepackage{tikz}
\usetikzlibrary{arrows.meta}

\usepackage{amsmath, amsthm}

\DeclareMathOperator*{\argmin}{arg\,min}

\newtheorem{theorem}{Theorem}[section]
\newtheorem{corollary}{Corollary}[theorem]
\newtheorem{lemma}[theorem]{Lemma}

\newcommand{\bracket}[1]{\left[ #1 \right]}

\newcommand{\paren}[1]{\left( #1 \right)}
\newcommand{\abs}[1]{\left| #1 \right|}

\newcommand{\R}[1]{\mathbb{R}^{#1}}
\newcommand{\N}[1]{\mathbb{N}^{#1}}

\newcommand{\Esub}[2]{\mathbb{E}_{#2}\left[ #1 \right]}
\newcommand{\E}[1]{\mathbb{E}\left[ #1 \right]}

 \newcommand{\indic}[1]{\mathbbm{1}_{\left\{#1\right\}}} \newcommand{\iid}{\stackrel{iid}{\sim}}

\newcommand{\gauss}[2]{\mathcal{N}\paren{#1, #2}}

 \def\BibTeX{{\rm B\kern-.05em{\sc i\kern-.025em b}\kern-.08em
    T\kern-.1667em\lower.7ex\hbox{E}\kern-.125emX}}
\newcommand{\xv}{\mathbf{x}}
\newcommand{\xiv}{\boldsymbol{\xi}}

\newcommand{\bX}{\boldsymbol{X}}
\newcommand{\bY}{\boldsymbol{Y}}
\newcommand{\bXi}{\boldsymbol{\Xi}}

\newcommand{\Wh}{\hat{W}}
\newcommand{\wh}{\hat{w}}

\newcommand{\sigx}{\sigma_x}
\newcommand{\sigxi}{\sigma_\xi}
\newcommand{\xtx}{\xv^T\xv}
\newcommand{\xt}{\tilde{\xv}}
\newcommand{\xttxt}{\tilde{\xv}^T\tilde{\xv}}

\newcommand{\erf}{\mathrm{erf}}
\newcommand{\Ht}{\tilde{H}}
\newcommand{\Xfp}{X_{fp}}
\newcommand{\Yfp}{Y_{fp}}
\newcommand{\snr}{\mathrm{SNR}}
\newcommand{\compsimplex}{\boldsymbol{\Delta}_{\mathcal{X}_c}}
\newcommand{\mse}{\mathrm{MSE}}
\newcommand{\Yhfp}{\hat{Y}_{fp}}

\begin{document}
\title{The Thermodynamic Costs of Simple Linear Regression}
\author{

Samuel H. D'Ambrosia$^\dagger$, Sultan M. Daniels$^\dagger$, Michael R. DeWeese, and Anant Sahai
\thanks{This work was supported in part by the U.S. Army Research Laboratory and the U.S. Army Research Office under Contract No. W911NF-20-1-0151, and by the H2H8 Nonprofit Organization.

}
\thanks{
$^\dagger$SHD and SMD contributed equally to this work.

Samuel H. D'Ambrosia and Michael R. DeWeese are with the Department of Physics, Redwood Center for Theoretical Neuroscience, and Berkeley AI Research Lab, at the University of California, Berkeley, CA 94720 (email: shda@berkeley.edu, deweese@berkeley.edu)

Sultan M. Daniels and Anant Sahai are with the Department of Electrical Engineering and Computer Sciences, and Berkeley AI Research Lab, at the University of California, Berkeley, CA 94720 (email: sultan\_daniels@berkeley.edu, sahai@eecs.berkeley.edu)

}
}

\maketitle

\begin{abstract}
    % \doublespacing
    The construction of models from data is a significant contributor to the energetic costs of computation. Because of this, understanding how foundational thermodynamic bounds apply to modeling algorithms will be increasingly important. Here, we study the thermodynamic costs of a basic and fundamental modeling algorithm: simple linear regression. Following Landauer, we approximate the thermodynamic lower bound on irreversibly performing both exact linear regression and linear regression via stochastic gradient descent as implemented on floating-point numbers. From this, we derive energy-cost aware scaling laws for the optimal dataset size for training a linear regression model given a generalization error dependent demand for inference. Additionally, we discuss a method to lower bound the entropy production from the mismatch cost for algorithms with continuous input variables. 
\end{abstract}


\section{Introduction}

% \doublespacing

The energetic costs of computation are significant and growing, consuming 4.4$\%$ of U.S. energy as of 2023 \cite{Shehabi2024}, with an increasing share of this energetic cost due to algorithms for constructing and deploying data-driven models \cite{Luccioni2024, deVries2023, Verdecchia2023}. 
Simultaneously, as the physical size of computer components shrinks, thermodynamic bounds will become increasingly relevant \cite{Schuster2016, Conte2017, Kish2004, Zhang2022, Landauer1961, Bennett1982, Lloyd2000, Frank2017}. 
Because of this, understanding both how fundamental thermodynamic lower bounds \cite{Landauer1961, Bennett1982, Bennett1982b, Landauer1972} and more advanced stochastic thermodynamic \cite{WolpertStochThermo2024, Wolpert2019, Yadav2025, Manzano2024, KolchinskyWolpert2021} bounds apply to modeling algorithms is of increasing importance. 

Here, we explore the thermodynamic costs of training a single-parameter model: a line with an intercept of zero. Linear regression is a simple predictive modeling algorithm that is central to modern machine learning \cite{McCullochPitts1943, Rosenblatt1958, Widrow1960, garethisl}. 
Although \cite{thermo_of_learning} studied thermodynamic limits for a binary classification task, the question of how thermodynamic bounds apply to regression remains unexplored.

The thermodynamics of computation is predominantly studied for discrete algorithms \cite{demaine_energy-efficient_2016, Wolpert2019}, yet many machine learning algorithms are designed for real-valued inputs, outputs, and model parameters that must be quantized to run on digital hardware. Although previous work has studied analog implementations of learning algorithms \cite{thermo_dnn}, we choose to analyze training algorithms that use floating-point representations since most frontier deep learning systems today are implemented on digital chips that are optimized for and judged on their performance using floating-point representations \cite{nvidia_nvidia_2025, advanced_micro_devices_inc_amd_2025}. 

To quantify thermodynamic costs, in Section~\ref{sec:fpn}, we show a link between the differential entropy of continuous random variables and the entropy of their floating-point quantized counterparts by extending the uniform lattice framework in \cite{kostina_data_2017} to the nonuniform bin structure of the floating-point format. 
The entropy of Gaussian distributed random variables quantized to floating-point numbers is derived in this paper, providing a theoretical foundation for empirical observations that exponent bits of neural network weights have low entropy while the mantissa bits have high entropy \cite{bordawekar2022efloatentropycodedfloatingpoint, hao_neuzip_2024}. The approximations used here for the entropy of floating point numbers are discussed further in \cite{daniels2026entropyfloatingpointnumbers}.

We apply Landauer's principle to explore the thermodynamic costs of single-parameter linear regression via two approaches: exact linear regression (Section~\ref{sec:exact}) and linear regression by stochastic gradient descent (SGD, Section~\ref{sec:sgd}). We show that the number of input data points is the primary contributor to the thermodynamic cost in both cases, finding that the number of mantissa bits has a large contribution to the thermodynamic cost due to the high entropy of the floating-point numbers' mantissa, and that input data with a higher signal-to-noise ratio leads to a lower thermodynamic cost. 
Following our analysis of the minimum energy costs, in Section~\ref{sec:scaling}, we derive a scaling law that represents this energy-accuracy tradeoff as a profit-maximization problem that accounts for the relationship between the model quality and the user demand for running inference \cite{sardana2025chinchillaoptimalaccountinginferencelanguage}. We find that the irreducible error of the model's predictions means that, in certain regimes, using more data to increase model accuracy is not worth the associated energy cost.
Lastly, we return to analyzing the costs of training and discuss a method for lower bounding the mismatch cost contribution to thermodynamic costs beyond Landauer's bound (Section~\ref{sec:mmc}).  

\section{Preliminaries}\label{sec:prelim}

One-dimensional linear regression forms a model for a data set given by $n$ data points $\mathcal{D} = \{(X_i,Y_i)\}_{i=1}^n$. For a given $w \in \R{}$, let $\bX \sim \gauss{0}{\sigx^2I_n}$, $\bXi \sim \gauss{0}{\sigxi^2 I_n}$, where $\bX = \begin{bmatrix}X_1, \dots, X_n \end{bmatrix}^T$, and $\bXi = \begin{bmatrix}\xi_1, \dots, \xi_n \end{bmatrix}^T$. The ground truth data labels are given by $\bY = w\bX + \bXi$. $\boldsymbol{X}$ and $\boldsymbol{\Xi}$
are assumed to be independent, with $\sigx > 0$ and $\sigxi > 0$. 
Equivalently, each data point is independently sampled from the distribution 
\begin{equation}\label{eq:initial_dist}
    f_{XY}(x,y) = \frac{1}{2\pi\sigma_x \sigma_\xi} \, \exp \bigg[-\frac{x^2}{2\sigma_x^2} - \frac{(y-wx)^2}{2\sigma_\xi^2} \bigg].
\end{equation}

We will consider fitting this data with a single variable $\hat w$, representing the slope of a line passing through zero. The loss function for the model will be the mean squared error
\begin{equation}
    \label{eq:loss}
    L(\wh) = \frac{1}{2n} \sum^n_{i = 1}(\hat w x_i - y_i)^2.
\end{equation}

There are two methods we will consider for finding $\hat w$: using the analytic formula to find the exact best fit given the data set, and using stochastic gradient descent to minimize the error. The exact expression for the optimal $\hat w$ given a dataset with $n$ data points is:

\begin{equation}
    \label{eq:exact_formula}
    \hat{w} = \frac{\sum^n_{i=1} x_iy_i}{\sum^n_{i=1} x_i^2}.
\end{equation}
On the other hand, stochastic gradient descent (SGD) allows us to approximate the best fit using an iterative updating method. Batches of size $B$ are sampled from an infinite data stream of independent and identically distributed $(X_i, Y_i)$ pairs, and used to update the model parameters. $\hat w$ is initialized at the value $\hat w_0 \in \R{}$, and updated for each batch using
\begin{equation}
    \label{eq:SGD_update}
    \wh_{k + 1} = \wh_k - \frac{\eta}{B} \sum_{(x_i, y_i) \in \mathcal{B}} \frac{\partial \ell(\wh, (x_i, y_i))}{\partial \wh} \bigg |_{\wh = \wh_k},
\end{equation}
where $\eta$ is the learning rate, $k$ is the step number, and the per-sample loss function is $\ell(\wh, (x_i, y_i)) = \frac{1}{2}(\wh x_i - y_i)^2$. 

While these results are theoretically understood as a relation among continuous variables, here we assume that the computation that determines $\hat w$ is implemented on discrete registers, each of which serves as a representation for a single continuous number. For each continuous variable, we can define a quantization map $Q(X) = X_Q$ that takes a continuous random variable $X$ to a discrete random variable $X_Q$. Assuming the representation is stored on $R$ physical bits, the quantization is given by a string of binary variables $x_Q \in \{0,1\}^R$. We will denote the Shannon entropy of a quantized variable $X_Q$ as $H(X_Q) = -\sum_{x_Q}p(x_Q) \log[p(x_Q)]$, where $\log[\cdot]$ is base 2. Additionally, we avoid the more detailed numerical analysis question of how finite-precision arithmetic perturbs the computation by assuming the regression algorithms evolve according to their ideal real-valued form. Accordingly, if $\Wh$ denotes the model parameter produced by the ideal continuous algorithm, we assume stored output is approximated by $Q(\Wh) = \Wh_Q$.

\subsection{Physical Implementation}\label{sec:phys_implement}

Here we will adopt the \textit{standard accounting convention} following Wolpert \cite[ pp. 32-36]{Wolpert2019}, which allows us to track the thermodynamic costs of logically irreversible computations. A logically irreversible computation is a process where the output state of the computer cannot be used to uniquely determine the input state, resulting in a loss of information. Under the standard accounting convention, computations are assumed to be performed by a cyclic device that computes output states from input states, while only saving outputs\footnote{While reversible algorithms for linear regression \cite{Demaine2021} and stochastic gradient descent \cite{maclaurin2015_rev_SGD} exist, here we focus on the case in which computations are done irreversibly, with inputs being erased and non-recoverable from saved output parameters. The difficulties of reversible computation, and further justification for the standard accounting convention are discussed in \cite{Wolpert2019}.}. 

We assume the physical system that implements the computation is composed of physical bits, systems with two distinct physical states labeled $\{0,1\}$. Let $I,M,O \in \N{}$, and $\mathcal{X}_I = \{0,1\}^I$ denote the set of possible states of the input bits, $\mathcal{X}_M = \{0,1\}^M$ denote the set of possible states of any intermediary bits, and $\mathcal{X}_O = \{0,1\}^O$ denote the set of possible states of the output bits. The joint logical state of the computer will be denoted by $x_c = (x_I, x_M, x_O)\in \mathcal{X}_c = \mathcal{X}_I \times \mathcal{X}_M \times \mathcal{X}_O$, where $x_I \in \mathcal{X}_I$, $x_M \in \mathcal{X}_M$, $x_O \in \mathcal{X}_O$. Let $p(x_c)$ be a probability distribution over $\mathcal{X}_c$. 

The equivalence between thermodynamic entropy and Gibbs-Von Neumann entropy is a classic result in statistical mechanics \cite{tolman, gibbs, maroney2008physicalbasisgibbsvonneumann}. Following Maroney \cite[p. 13]{Maroney}, the Gibbs-Von Neumann entropy (referred to as the thermodynamic entropy) of a physical system with distinct logical states $x_c$ is given by
\begin{equation}
    \label{eq:phys_ent}
    S_{sys} = k_B \ln[2]H(p(x_c)) + \sum_{x_c \in \mathcal{X}_c} p(x_c) S_{x_c}
\end{equation}
where $S_{x_c}$ is the internal entropy of each logical state, and $H(p(x_c)) = -\sum_{x_c \in \mathcal{X}_c} p(x_c) \log[p(x_c)]$. Here we assume the internal entropy of each logical state is the same $\forall x_c \in \mathcal{X}_c, S_{x_c} = S_0$ (i.e. the 0 state of each bit has the same number of equally likely underlying physical microstates as the 1 state), meaning the average internal entropy over any probability distribution $p(x_c)$ over computational states will be $\sum_{x_c \in \mathcal{X}_c} p(x_c) S_{x_c} = S_0$. 

Landauer's principle holds that a logically irreversible computation must incur an energetic cost, due to logically irreversible operations reducing the thermodynamic entropy of the computational system \cite{Landauer1961}. Assuming the computational system begins and ends at the same average internal energy $U_0$, and all components of the computational system remain in thermal equilibrium with their environment at temperature $T$, performing an entropy-reducing computation expels on average at least $Q \geq Q_{min} = -T \Delta S_{sys}$ heat into the computational system's environment \cite{Landauer1961, Maroney}, and requires on average $W \geq W_{min} = -T \Delta S_{sys}$ work (i.e. work would be energy taken from the computer's battery, since by the first law of thermodynamics, if $\Delta U = W - Q = 0$, $Q = W$ \cite{Callen}). From this we associate the energetic cost $\Delta E$ with the work required and heat expelled, $\Delta E \triangleq W = Q$. Under these assumptions, the generalized Landauer's bound states the lower bound on the average energetic cost of a computation is given by $\Delta E \geq \Delta E_{min} = -T \Delta S_{sys}$ \cite{Landauer1961, Maroney}.

The standard accounting convention imposes the following constraints on the probability of computational states at a given step of the computational process. Let the computational process finish at step $F$. Let $x_c^0 =(x_I^0, x_M^0, x_O^0) \in \mathcal{X}_c$ be the initial state of the computer. At step $0$, $p_0(x_c) = \delta_{x_c,x_c^0}$. The cyclic device will return to this state after the computation finishes. At step $1$, the input is loaded while all other bits remain in their initial state, implying $p_1(x_c) = p(x_I) \delta_{x_M,x_M^0} \delta_{x_O,x_O^0}$, where $p(x_I)$ is the probability of input states in $\mathcal{X}_I$. Following \cite[p. 34]{Wolpert2019} if we assume all registers are initialized to the known state $x_c^0$, loading the input at step 1 can be done in an energetically cost free manner. At step $F-1$, the output will have been computed, while the input and intermediary bits are assumed to have been reset to their initial positions specified by $x_c^0$, implying that $p_{F-1}(x_c) = \delta_{x_I,x_I^0} \delta_{x_M,x_M^0} p(x_O)$ where $p(x_O)$ is the probability of output states in $\mathcal{X}_O$. Finally, step $F$ resets the system to $p_F(x_c) = p_0(x_c) = \delta_{x_c,x_c^0}$. If the output is saved and passed on to another computational system, the final step $F$ which resets the output bits can be undertaken in an energetically cost free manner \cite[p. 34]{Wolpert2019}, which matches our case since we would want to keep the model for use after training. 

This implies that only steps from 1 to $F-1$ incur an energetic cost, fixed by the difference in thermodynamic entropy between step 1 and step $F-1$. With $p_1(x_c) = p(x_I) \delta_{x_M,x_M^0} \delta_{x_O,x_O^0}$ and $X_I \sim p(x_I)$, by Eq.~\eqref{eq:phys_ent} we see the thermodynamic entropy of the system at step 1 is $S_{1} = k_B \ln[2] H(X_I) + S_0$. Similarly, with $p_{F-1}(x_c) = \delta_{x_I,x_I^0} \delta_{x_M,x_M^0} p(x_O)$ and $X_O \sim p(x_O)$ we have $S_{F-1} = k_B \ln[2]H(X_O) + S_0$. Thus, assuming equal internal entropy $S_0$ across each logical state $x_c \in \mathcal{X}_c$, unchanged average energy $U_0$, and thermal equilibrium with an environment at temperature $T$, the lower bound on the energetic cost $\Delta E_{min} = -T\Delta S_{sys} = T(S_{1} - S_{F-1})$ of moving from step 1 to step $F-1$ is given by,
\begin{equation}
    \label{eq:Landauer}
    \Delta E_{min} = k_BT \ln[2] (H(X_I)
 - H(X_O)).
\end{equation}
$\Delta E_{min}$ will be referred to as the \textit{Landauer cost} (referred to as the `unconstrained' or `all-at-once' Landauer cost in \cite{Wolpert2019}).

We assume the computation is implemented `all-at-once', meaning the details of the thermodynamic cost of the intermediary transitions between 1 and $F-1$ states, or the use and re-initialization of any intermediary registers $\mathcal{X}_M$, are ignored. As discussed in \cite[p. 55]{Wolpert2019}, the cost of re-initializing the inputs and intermediary registers will be greater than the cost of simply re-initializing inputs. Focusing on the costs associated with input and output states allows us to obtain a lower bound on the energetic cost of the computation while leaving the analysis of more specific and complex intermediary steps to future work.

The Landauer cost provides a lower bound on the energetic cost of a logically irreversible computation by assuming the physical implementation is thermodynamically reversible (note that thermodynamic and logical reversibility are entirely distinct \cite{MARONEY_no_relationship}). A process is called \textit{thermodynamically reversible}, when the bound $Q \geq -T \Delta S_{sys}$ is saturated at $Q = Q_{min} = -T \Delta S_{sys}$ \cite{Callen,landau1980statistical, chandler}. However, real physical systems often incur energetic costs beyond Landauer's bound, increasing the heat generation and work requirements beyond those implied by a decrease in the computational system's entropy. Section~\ref{sec:mmc} discusses one type of energetic cost beyond the thermodynamically reversible lower bound: the mismatch cost (MMC) \cite{Wolpert2019, Yadav2025, KolchinskyWolpert2021, Manzano2024}.

\subsection{Exact Linear Regression}

In the case of exact linear regression, the cyclic device that computes the best fit slope takes input $X_I = Q(\mathcal{D}) = \mathcal{D}_Q$, which is composed of $2nR$ input bits. This is because each of $n$ data points has two discrete representations of continuous variables, $x$ and $y$, each of which is represented by $R$ bits. The output $X_O= \hat W_Q$ is a single $R$-bit discrete representation of the continuous variable $\Wh$. With each $(X,Y)$ data point being independent, the entropy of $\mathcal{D}_Q$ can be written as the sum of the entropy of the $n$ data points. Therefore, the Landauer cost for an irreversible cyclic device performing exact linear regression on independent data points is:

\begin{align}
    \label{eq:exact_landauer}
    \Delta E^{Ex}_{min} & = k_BT \ln[2] (H(\mathcal{D}_Q) - H(\Wh_Q)) \\ & = k_BT \ln[2] (n H(X_Q,Y_Q) - H(\Wh_Q)).
\end{align}

\subsection{Stochastic Gradient Descent}\label{sec:sgd_prelim}

For SGD, we assume the cyclic device completes one cycle for each update step specified by Eq.~\eqref{eq:SGD_update}. $2B$ registers are required to load each batch $\mathcal{B}_Q$, with $1$ additional register required to load the model parameter $\Wh_{Q,k}$ from step $k$, meaning a total of $(2B+1)R$ input bits for $X_I = (\Wh_{Q,k}, \mathcal{B}_Q)$. After the computation is completed, a single register of $R$ bits holding $X_O=\Wh_{Q,k+1}$ is saved, then re-loaded onto the input to the next step while assuming all other registers are re-initialized. Assuming each data point is independent, the Landauer cost of a single update is given by

\begin{align}
    \label{eq:one_step_L}
    \Delta E_{min,k} &  = k_BT \ln[2](H(\Wh_{Q,k}, \mathcal{B}_Q) - H(\Wh_{Q,k+1})) \\ & = k_BT \ln[2](H(\Wh_{Q,k}) + B H(X_Q,Y_Q) - H(\Wh_{Q,k+1})),
\end{align}

After $\tau$ updates the process terminates and saves the final model parameter $\hat W_{Q,\tau}$. We start with $\Wh$ initialized precisely to $\wh_0$, so that $H(\Wh_{Q,0}) = 0$. Summing over the energetic cost of each update Eq.~\eqref{eq:one_step_L}, $-H(\Wh_{Q,k+1})$ from step $k$ cancels with $H(\Wh_{Q,k})$ from step $k+1$. The resulting Landauer cost after $\tau$ updates is

\begin{equation}
    \begin{aligned}
        \Delta E^{SGD}_{min} & = \sum_{k=0}^{\tau-1} \Delta E_{min,k} = \sum_{k=0}^{\tau-1} k_BT \ln[2](H(\Wh_{Q,k}) + B H(X_Q,Y_Q) - H(\Wh_{Q,k+1})) \\
        & = k_BT \ln[2](\tau B H(X_Q,Y_Q) - H(\Wh_{Q,\tau})).
    \end{aligned}
\end{equation}

\section{The Entropy of Floating-point Numbers}\label{sec:fpn}

In digital hardware, numerical values are stored as bit strings in finite registers, coming in many forms. Two examples are integers, uniformly spaced values with a range specified by the number of bits, and floating-point numbers, which uses a finite register of binary bits to represent a real value $x$ in binary scientific notation. 
Here, we focus on floating-point numbers, due to their ubiquitous role in representing continuous variables, and because floating point numbers play a central role in modern deep learning algorithms and hardware \cite{nvidia_nvidia_2025, advanced_micro_devices_inc_amd_2025}. 
As deep neural network training and inference loads attempt to saturate the computational resources of specialized hardware, there has been renewed interest in developing novel floating-point formats that find the right tradeoffs between precision, dynamic range, memory, computational speed and numerical stability \cite{kuzmin_fp8_2022}.
Some examples are microscaling (MX) formats \cite{rouhani_microscaling_2023, darvish_rouhani_shared_2023, su_characterization_2025} where $k$ numbers are represented with a single shared exponent that encodes the scale of the block, and $k$ scalar elements that are encoded using 8 bit or less floating point or integer formats.\footnote{An example is MXFP8 where a block has $32$ elements with a shared 8-bit exponent and each scalar element is encoded with FP8 using four exponent bits and three significand bits.} 

\subsection{The Structure of a Floating-point Number}\label{sec:fpn_structure}

In this paper, we will use a simple normalized floating-point format with midpoint rounding that does not use subnormal numbers (as defined in \cite{muller2010handbookfpn}) and where zero is not in the representable set.

Let $p \in \N{}$, $E \in \N{} \cup \{0\}$, and $\alpha \in \R{}$. The following map is defined on $\R{}\setminus\{0\}$; since all distributions used in this paper are absolutely continuous, the event that the value zero is realized is probability-zero. Define \begin{equation}
    \mathrm{round}_{p}(\alpha) \triangleq \paren{\argmin\limits_{i \in \{0, \dots, 2^{p-1}\}}\abs{\alpha - i2^{-(p-1)}}}2^{-(p-1)},
\end{equation} where ties in the $\argmin$ function are broken towards the higher index $i$. Assuming rounding to the nearest representable value, the value stored on a floating-point number can be written as
\begin{equation}
    x_{fp}(x) \triangleq (-1)^{s_{fp}(x)} \times 2^{e_{fp}(x)} \times \paren{1 + m_{fp}(x)},
\end{equation}
where \begin{equation}
\label{eq:sem_fp}
    \begin{aligned}
        s_{fp}(x) &\triangleq \indic{x < 0},\\
        \tilde e_{fp}(x) &\triangleq -\paren{2^{E-1} - 1}\indic{\log|x| < -\paren{2^{E-1} - 1}}\\
        &+ \sum\limits_{i=0}^{2^E - 1}\bracket{i - \paren{2^{E-1} - 1}}\indic{i - \paren{2^{E-1} - 1} \leq \log|x| < i+1 - \paren{2^{E-1} - 1}} + 2^{E-1}\indic{\log|x| \geq 2^{E-1} + 1}\\
          m_{fp}(x) &\triangleq \mathrm{round}_{p}\paren{|x|2^{-\tilde e_{fp}(x)} - 1}\indic{\log|x| < 2^{E-1}+1}\indic{\mathrm{round}_{p}\paren{|x|2^{-\tilde e_{fp}(x)} - 1
     } \neq 1} \\
          &+ \paren{1 - 2^{-(p-1)} }\paren{\indic{\log|x| \geq 2^{E-1}+1} + \indic{2^{E-1} \leq \log|x| < 2^{E-1}+1}\indic{\mathrm{round}_{p}\paren{|x|2^{-\tilde e_{fp}(x)}
     - 1} = 1}}\\e_{fp}(x) &\triangleq \min\paren{2^{E - 1}, \tilde e_{fp}(x) + \indic{\mathrm{round}_{p}\paren{|x|2^{-\tilde e_{fp}(x)} - 1} = 1}}.
       \end{aligned}\footnote{$\tilde e_{fp}(x)$ must be introduced to handle cases where mantissa value rounds up to the next exponent level. This occurs when $\mathrm{round}
      _{p}\paren{|x|2^{-\tilde e_{fp}(x)} - 1} = 1$.}
   \end{equation}
The quantities in Eq.~\eqref{eq:sem_fp} are decoded numerical values: $s_{fp}(x)$ is a single sign bit, $e_{fp}(x)$ is the binary exponent for $E \geq 1$ (for $E=0$, the format has a single exponent level with implicit exponent value $1/2$), and $m_{fp}$ is the significand (or mantissa) which encodes the binary significant digits of $x$ in the form $1.m_{fp}(x)$; in the physical registers these are encoded as binary integers using $1$, $E$, and $(p-1)$ bits respectively \cite{muller2010handbookfpn, Goldberg1991_FPN_Intro}. Since each stored field is a finite set of non-negative integers constrained by the number of bits in the register, this structure can only represent a finite set $U_{fp}$ of representable numbers along the real number line. For example, the standard single-precision IEEE-754 format uses one bit for $s_{fp}$, 23 bits for $m_{fp}$, and eight bits for $e_{fp}$. Note that the format defined by Eq.~\eqref{eq:sem_fp} is an idealized normalized floating-point format: it does not include zero, subnormal numbers, infinities, or NaNs, and its exponent range ($e_{\min} = -(2^{E-1}-1)$, $e_{\max} = 2^{E-1}$) differs slightly from the IEEE-754 standard. 
The structure of a floating-point number is illustrated in Fig.~\ref{fig:fp_diag}. See \cite{muller2010handbookfpn, Goldberg1991_FPN_Intro} for further details on the floating-point representation.

\subsection{Computing the Exact Discrete Entropy of Quantized Random Variables}\label{sec:exact_entropy}

The purpose of a floating-point number is to store a real number in a discrete state. Thus for a continuous random variable $X$, our goal will be to find the discrete entropy $H(X_{fp})$, where $X_{fp}$ represents the discrete random variable resulting from the floating-point quantization of $X$. We first directly compute the entropy of a random variable that is clipped and midpoint quantized to representable values in Theorem~\ref{thm:quantized_ent}. Then, Corollary~\ref{cor:gauss_fp_ent} finds the representable values of the floating-point representation given in Eq.~\eqref{eq:sem_fp}, which can be directly applied to Theorem~\ref{thm:quantized_ent} to directly compute the entropy of $\Xfp$. Due to their length, these exact expressions can be found in App.~\ref{app:exact_entropy_float}. Nevertheless, their numerical evaluations can be seen in Figures~\ref{fig:fp_exact_approx} and~\ref{fig:full_fp_exact_approx_p_E} which support the analysis of following approximations.

\subsection{Illustrative and Computationally Tractable Approximations}\label{sec:illustrative_approx}

We can gain further insight into the entropy of floating-point numbers, and a significant computational advantage for the entropy of large floating-point numbers, if we make three approximations which allow for the derivation of analytic expressions for the floating point entropy: relating the discrete entropy of the quantized variable to the differential entropy of the continuous variable, smoothing and extending the bin size function, and approximating non-zero mean distributions. The approximations applied here are discussed for a broader class of distributions in \cite{daniels2026entropyfloatingpointnumbers}.

\subsubsection*{Approximation 1 -- Relating discrete and differential entropy}

The first approximation relates the differential entropy of a continuous random variable to the discrete entropy of its counterpart discrete representation. This has been studied by \cite{CoverThomas, renyi_dimension_1959, jaynes_information_1962, Jaynes1968-limiting, linder_asymptotic_1994, gish_asymptotically_1968, kostina_data_2017}, focusing on when the quantization uses uniform bins of size $\Delta \in \R{}$. This results in the well known relationship $H(X_Q) \approx h(X) - \log \Delta$, where $h(X) = - \int f_X(x) \log f_X(x) dx$ is the differential entropy. 

However, the formalism in \cite{kostina_data_2017} can be extended to non-uniform bins, and applied to the case of floating-point numbers. In App.~\ref{app:FPN_app1}, the relationship between the differential and discrete entropy is derived for the non-uniform quantization of a $d$-dimensional distribution on $d$ independent registers. For simplicity, here we can focus on the case of a univariate distribution where $d=1$. Assuming bounded support on the region of $\R{}$ bounded by the granular region (see \cite{GrayNeuhoff1998Quantization}) $\mathbb{U} \triangleq [-2^{e_{max} + 1} + 2^{e_{max} - p}, \,\, 2^{e_{max} + 1} - 2^{e_{max} - p}]$, where $e_{max} = 2^{E-1}$. Approximation 1 is $H(X_Q) \approx \tilde H(X_Q)$, where
\begin{equation}
    \label{eq:diff_approx}
    \tilde H(X_Q) \triangleq -\int_\mathbb{U} f_X(x) \log [f_X(x) \Delta(x)] dx = h(X) - \Esub{\log[\Delta(x)]}{\mathbb{U}},
\end{equation}
and $\Esub{\cdot}{\mathbb{U}}$ denotes expectation over the granular region $\mathbb{U}$. As discussed in App.~\ref{app:FPN_app1} and \cite{kostina_data_2017}, approximation 1 is valid if $f_X(x)$ varies slowly over each bin, meaning $f_{X}(x)$ is well approximated by a piecewise distribution which is uniform over each bin (see Fig.~\ref{fig:fp_joint_exact_approx_p_E_snr_100}). 

This expression can be understood through E.T. Jaynes' work on the ``limiting density of discrete points" \cite{Jaynes1968-limiting}. Jaynes emphasized that differential entropy by itself is not an absolute quantity\footnote{While the primary motivation for introducing a discrete representation is the reality that modern computers use discrete representation, the known pathologies of differential entropy also make its application to a cyclic device difficult. Differential entropy is coordinate dependent (as noted by Jaynes \cite{Jaynes1968-limiting}), and it can diverge to $h(X) \rightarrow - \infty$ for sufficiently narrow distributions. Since the standard accounting convention assumes registers are re-initialized to specific, i.e. infinitely narrow states during the computation, use of the differential entropy would not be possible.}, depending implicitly on a reference measure that specifies how the continuous space is discretized. The position dependent bin density provides us with this underlying measure, by specifying the density of states along the number line as $1/\Delta(x)$ \cite{Jaynes1968-limiting}.

\subsubsection*{Approximation 2 -- Smoothing the bin size function and extending its domain}

For floating-point numbers, the spacing between adjacent representable values is constant inside each exponent block. When the value to be quantized $x$ is within the range $[2^{e_{fp}(x)}, 2^{e_{fp}(x)+1})$, a floating-point number can represent $2^{(p-1)}$ equally spaced values. Thus for $2^{e_{fp}(x)} \leq x < 2^{e_{fp}(x)+1}$, meaning that $x$ is in a bin in the interior of an exponent block, the bin size is, 
\begin{equation}
    \label{eq:true_bins}
    \Delta(x) = \frac{2^{e_{fp}(x)+1} - 2^{e_{fp}(x)}}{2^{p-1}}= 2^{e_{fp}(x) - (p-1)} \text{ if }x \in \mathbb{U}.
\end{equation}
In the main text, only the bin size when $x$ is in an interior bin is given for readability; the true midpoint bin-size function $\Delta(x)$ from App.~\ref{app:FPN_app1} takes on different values when $x$ is in the two outer clipping bins, the exponent-boundary bins and the two bins adjacent to zero. These special bins are treated explicitly in App.~\ref{app:FPN_app2}.

Our second approximation will alter the bin size. First we can smooth the steps by introducing a best-fit linear approximation as shown as the dashed red line in Fig.~\ref{fig:fp_diag}, which uses $e_{s}(x) \triangleq \log[|x|/\sqrt{2}]$. This gives us
\begin{equation}
    \label{eq:step_approx}
    \Delta_s(x) \triangleq \frac{1}{\sqrt{2}}|x| \cdot 2^{1-p} \approx \Delta(x).
\end{equation}
By Theorem~\ref{thm:steps}, the error caused by introducing this smooth approximation is bounded by $d/2$ plus the small contribution from the bins adjacent to zero; for most distributions it will be much smaller due to cancellation in the overestimation and underestimation.

Unlike $\Delta(x)$, it is simple to define $\Delta_s(x)$ for any $x \in \mathbb{R}$. With this, and by assuming the probability of overflow and underflow values is small, $\int_{\mathbb{R} \setminus \mathbb{U}} f_X(x) dx \approx 0$, we extend the domain of integration from $\mathbb{U}$ to $\mathbb{R}$. Since $\Delta_s(x)$ grows only linearly in $|x|$, the omitted tail integral $\int_{\mathbb{R}\setminus\mathbb{U}} f_X(x)|\log[f_X(x)\Delta_s(x)]|\,dx$ is near zero for all distributions used in this paper (Gaussian tails decay exponentially and Student's $t$ tails decay as a power law, both of which dominate the linear growth of $\Delta_s$) when the granular region $\mathbb{U}$ is sufficiently large. Approximation 2 states that $\Ht_s(x) \approx \tilde H(x)$ :

\begin{equation}
    \label{eq:first_two_approx}
    \Ht_s(X_{fp}) \triangleq h(X) + (p-1)- \int^\infty_{-\infty} f_X(x) \log\bracket{\frac{|x|}{\sqrt{2}}} dx,
\end{equation}
where $h(X) = -\int^\infty_{-\infty} f_X(x) \log f_X(x) dx$. Approximation 2 applies when $f_X(x)$ has low overflow or underflow probability. Beyond this, the bin-size smoothing error is bounded by $d/2$ plus the small contribution from the bins adjacent to zero (Theorem~\ref{thm:steps}); the domain extension from $\mathbb{U}$ to $\mathbb{R}$ is a separate approximation controlled by the probability mass outside $\mathbb{U}$. Simple and predictive analytic expressions can be obtained using just Eq.~\eqref{eq:first_two_approx} when continuous random variables are Gaussian and centered on 0, as discussed in Section~\ref{sec:gaussians}. 

\subsubsection*{Approximation 3 -- Evaluating $\mathbb{E}[\log[|x|/\sqrt{2}]]$}

We can also consider an approximation which can assist with distributions that are not centered at 0. Here we can assume the distribution is offset by a shift, $\mu$. In this case, $\Ht_s(X_{fp})$ will depend on a convolution of $\log[|x|/\sqrt{2}]$ with the probability distribution $f_X(x - \mu)$, 
\begin{equation}
    \label{eq:G_int_non_gauss}
    G_{f_X}(\mu) \triangleq \int^{\infty}_{-\infty} f_X(x-\mu) \log[|x|/\sqrt{2}] dx.
\end{equation}
For a smooth distribution which decays rapidly around its mean $\mu$, the resulting integral will approximate a smoothed $\log\bracket{|x|/\sqrt{2}}$ function, as shown for a Gaussian distribution in Fig.~\ref{fig:fp_off_zero}. Approximation 3 states that $\tilde H_s(X_{fp}) \approx \tilde H_s^\mu(X_{fp})$, with
\begin{equation}
    \label{eq:three_approxes}
    \Ht^\mu_s(\Xfp) \triangleq h(X) + (p-1) - \log\bracket{|\mu| / \sqrt{2}}.
\end{equation}
This approximation replaces the integral term $G_{f_X}(\mu)$ from Eq.~\eqref{eq:G_int_non_gauss} with $\log[|\mu|/\sqrt{2}]$, and applies for $\mu \neq 0$ when $X \sim \mathcal{N}(\mu, \sigma^2)$ with $|\mu| \gg \sigma$, as well as non-Gaussian distributions which approximate a delta function as a parameter $\epsilon$ defining the distribution's width is decreased toward zero ($\epsilon \to 0$), making the distribution increasingly concentrated around $\mu \neq 0$, as shown in App.~\ref{app:FPN_app3}. The full approximation $H(X_{fp}) \approx \Ht^\mu_s(X_{fp})$ additionally requires that Approximation~1 remains valid, which can fail when the distribution becomes narrower than the quantization bin size. 
Fig.~\ref{fig:fp_exact_approx_mu} shows this approximation along with the true value of $H(\Xfp)$ as $\mu$ varies.

\subsection{Approximating the entropy of a floating-point quantized univariate and bivariate Gaussian}

\label{sec:gaussians}

\subsubsection{Univariate Gaussian}
Applying Eq.~\eqref{eq:first_two_approx}, one can obtain an analytic expression for the approximate discrete entropy of a zero-mean normally distributed continuous random variable $X \sim \mathcal{N}(0, \sigma^2)$ as represented on a floating-point number. By Theorem~\ref{thm:single_gaussian}, $\Ht_s(X_{fp}) = \tilde H^0_s(p)$, where
\begin{equation}
    \label{eq:single_mean_0}
    \Ht^0_{s}(p) \triangleq p + \frac{1}{2} \log[2 \pi e] + \frac{\gamma_e}{2\ln[2]} \approx p + 2.46 \text{ bits}.
\end{equation}
$p$ is the precision, and $\gamma_e \approx 0.5772$ is the Euler--Mascheroni constant.

\begin{figure}[htbp]
  \centering
  \begin{subfigure}{0.32\textwidth}
    \centering
    \includegraphics[width=\linewidth]{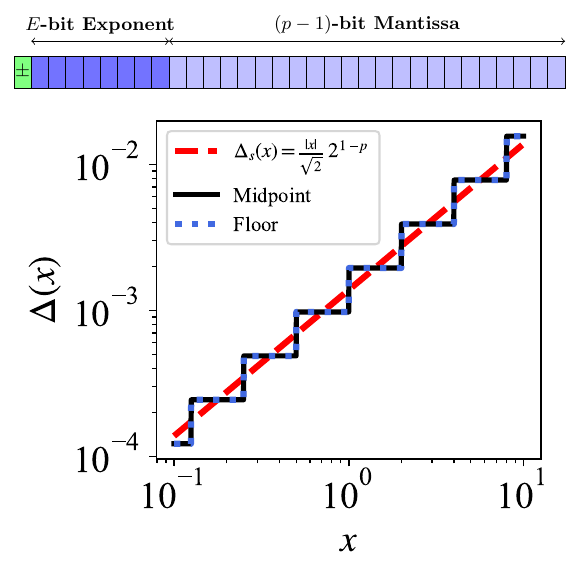}
    \caption{}
    \label{fig:fp_diag}
  \end{subfigure}
  \begin{subfigure}{0.32\textwidth}
    \centering
    \includegraphics[width=\linewidth]{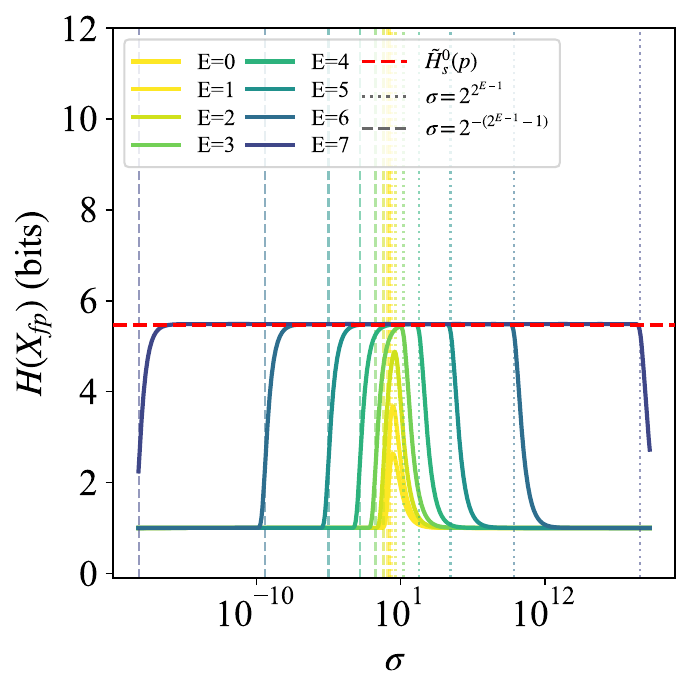}
    \caption{}
    \label{fig:fp_exact_approx_sigma}
  \end{subfigure}
  \begin{subfigure}{0.32\textwidth}
    \centering
    \includegraphics[width=\linewidth]{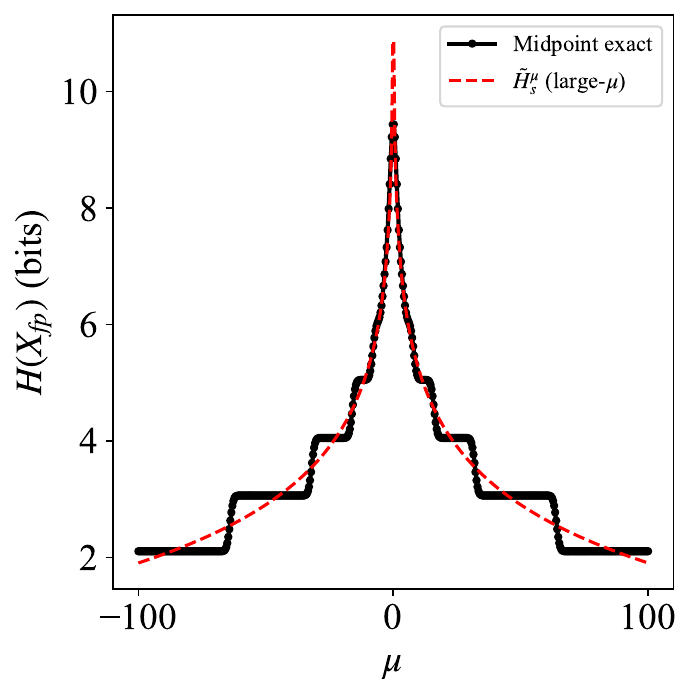}
    \caption{}
    \label{fig:fp_exact_approx_mu}
  \end{subfigure}
\caption{\textit{Floating-point structure and Gaussian approximations}. (\ref{fig:fp_diag}) The structure of a floating-point number where each box represents one bit. The true bin size function $\Delta(x)$ is plotted on log-log scale for both midpoint (black solid curve) and floor quantization (blue dotted curve) with $p = 10$ and $E=4$ along with the smooth approximation $\Delta_s(x)$ (red dashed curve). (\ref{fig:fp_exact_approx_sigma}) shows the entropy of $\Xfp$ which is a discrete representation of the random variable $X \sim \gauss{0}{\sigma^2}$ that has been clipped and midpoint quantized onto a floating-point representation with precision $p=3$ and various numbers of exponent bits $E$. The x-axis shows the standard deviation $\sigma$ of the underlying continuous distribution. The horizontal red line shows the approximate entropy $\Ht^0_s(p)$. The vertical dashed lines mark $\sigma = 2^{e_{\min}}$, while the vertical dotted lines mark $\sigma = 2^{e_{\max}}$ for each $E$. Notice that exact entropy closely follows the approximate entropy until $\sigma$ approaches these boundaries for each $E$. (\ref{fig:fp_exact_approx_mu}) shows the entropy of $\Xfp$ when $p = 7$ and $E = 7$, while the mean $\mu$ of the underlying continuous random variable $X \sim \gauss{\mu}{1}$ varies from $-100$ to $100$. The approximate entropy $\Ht^\mu_s(\Xfp)$ is plotted as the red dashed curve alongside the exact entropy (solid black curve) and smoothly passes through the stepwise drops in the exact entropy as $|\mu|$ increases.}
  \label{fig:fp_exact_approx}
\end{figure}

$\Ht^0_s$ predicts that the floating-point entropy will be constant with respect to the variance of the continuous random variable, seen by the red dashed line in Fig.~\ref{fig:fp_exact_approx_sigma}. This shows that the approximation $H(X_{fp}) \approx \tilde H_s(X_{fp})$ works well for a wide range given there are few overflows or underflows (see App.~\ref{app:exact_FPN_extrafigs}), even for 2-bit mantissa when $p=3$. Fig.~\ref{fig:fp_exponents} in App.~\ref{app:FPN_app_gauss} displays a histogram of the exponents for numbers sampled from the distribution $\mathcal{N}(0,\sigma^2)$. Additionally, from Fig.~\ref{fig:fp_exact_approx_mu} we see $H(X_{fp}) \approx \tilde H^\mu_s(X_{fp})$ works well beyond values near 0, with errors remaining within $1/2$ of a bit despite oscillatory errors due to the steps in the true bin size.

\begin{figure}[htbp]
  \centering
  \begin{subfigure}{0.32\textwidth}
    \centering
    \includegraphics[width=\linewidth]{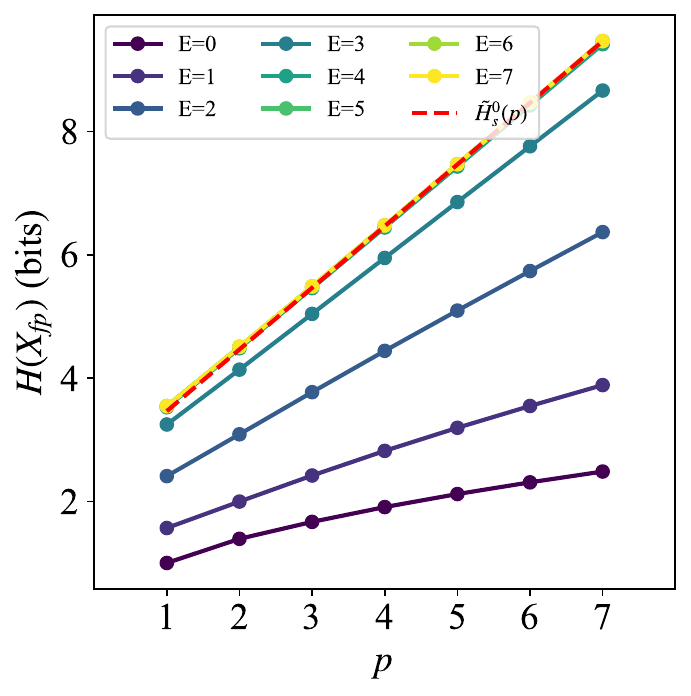}
    \caption{}
    \label{fig:fp_exact_approx_p_E}
  \end{subfigure}
  \begin{subfigure}{0.32\textwidth}
    \centering
    \includegraphics[width=\linewidth]{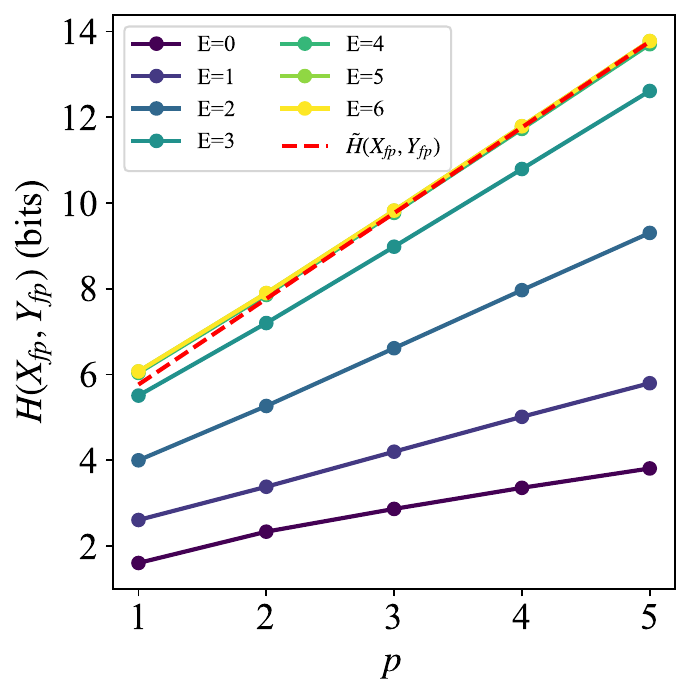}
    \caption{}
    \label{fig:fp_joint_exact_approx_p_E_snr_2}
  \end{subfigure}
  \begin{subfigure}{0.32\textwidth}
    \centering
    \includegraphics[width=\linewidth]{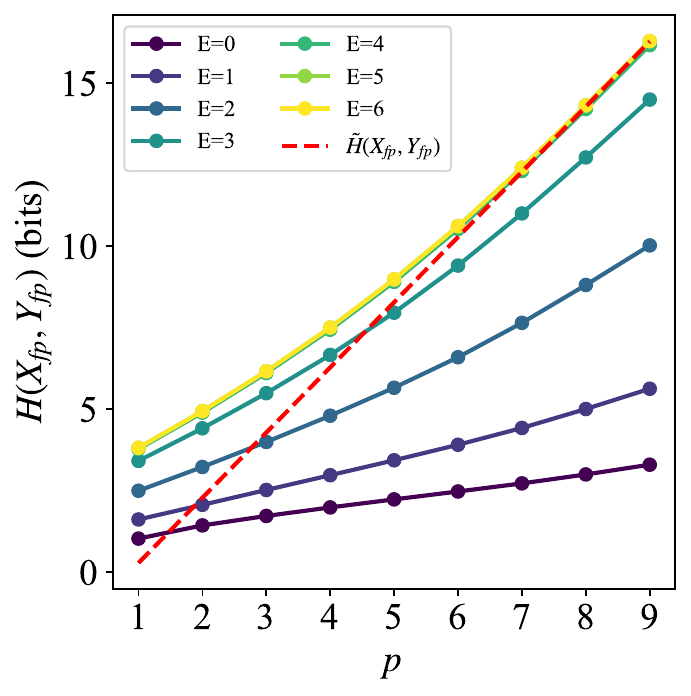}
    \caption{}
    \label{fig:fp_joint_exact_approx_p_E_snr_100}
  \end{subfigure}
\caption{\textit{Regimes where the approximations hold}. (\ref{fig:fp_exact_approx_p_E}) Shows the exact discrete entropy of $\Xfp$ as $p$ varies for each $E$ when the underlying continuous random variable is $X\sim \gauss{0}{1}$. The approximation $\Ht^0_s(p)$ is also plotted as the dashed red line. Notice that the curves for $E \geq 4$ are directly on top of each other, showing that $\Ht^0_s(p)$ is close to $H(\Xfp)$ when $E$ is large enough to keep the probability of overflows and underflows low (see App.~\ref{app:exact_FPN_extrafigs}). The approximation stays close to the true entropy with a large exponent and low precision. (\ref{fig:fp_joint_exact_approx_p_E_snr_2}) Shows the joint entropy of $(\Xfp, \Yfp)$ where $X \sim \gauss{0}{1}$, $\xi \sim \gauss{0}{0.25}$ and $Y = X + \xi$. The approximation $\Ht_s(\Xfp, \Yfp)$ is the red dashed line, and it again follows closely to $H(\Xfp, \Yfp)$ for $E \geq 4$. (\ref{fig:fp_joint_exact_approx_p_E_snr_100}) shows the same quantities as Fig.~\ref{fig:fp_joint_exact_approx_p_E_snr_2}, except $SNR \triangleq w^2\sigx^2/\sigxi^2 = 10000$ as opposed to $4$. In this high $SNR$ regime, the $\Ht_s(\Xfp, \Yfp)$ underestimates $H(\Xfp, \Yfp)$ for $p \leq 6$ due to the quantization noise that is present when the precision is low. Like the other settings, $\Ht_s(\Xfp, \Yfp)$ follows closely to $H(\Xfp, \Yfp)$ when $p >6$ and $E \geq 4$. Corollary~\ref{cor:joint_ent} reviews the exact calculation of the discrete entropy for the bivariate Gaussian case.}
  \label{fig:full_fp_exact_approx_p_E}
\end{figure}

\subsubsection{Bivariate Gaussian} 

We can apply a similar approach to approximating the discrete entropy for the joint entropy of $X$ and $Y = wX + \xi$ as represented over two floating-point numbers, defined by Eq.~\eqref{eq:initial_dist}. This makes use of the multivariate results presented in App.~\ref{app:FPN_app1}. For a given underlying slope $w$, by Theorem~\ref{thm:double_gaussian}, we obtain 
\begin{equation}
    \label{eq:approx_joint_ent_SNR}
    \begin{aligned}
        \Ht_s(\Xfp, \Yfp) &= 2\Ht^0_{s}(p) - \frac{1}{2} \log \bigg [1 + \frac{\sigma_x^2 w^2}{\sigma^2_\xi} \bigg ].
    \end{aligned}
\end{equation}

This expression in general depends on $w$, and if $w \neq 0$, then the discrete entropy will in general be lower than $2 \Ht^0_{s}(p)$, to a degree determined by the signal-to-noise ratio $\snr = \sigma_x^2 w^2/\sigma_{\xi}^2$. One issue to notice is that for an infinite signal-to-noise ratio the expression diverges, similar to an infinite capacity scalar Gaussian channel \cite{CoverThomas, shannon_communication_1949}. This is because Eq.~\eqref{eq:diff_approx} only applies for distributions which are sufficiently broad with respect to the bin size. For infinite signal-to-noise ratio, the distribution will be a sharp line in the joint space, which will vary quickly over individual bins, making approximation 1 inaccurate. 

This is seen in Fig.~\ref{fig:fp_joint_exact_approx_p_E_snr_100} where the dashed line showing the approximation only closely follows the exact joint entropy once the precision $p$ is $7$ for a signal-to-noise ratio of $10000$, while for a signal-to-noise ratio of $4$ Fig.~\ref{fig:fp_joint_exact_approx_p_E_snr_2} the approximation closely follows the exact joint entropy once the precision $p$ is $2$. In particular, the approximation underestimates the exact entropy in the high-$\snr$ small-$p$ regime. When there is infinite $\snr$, $X$ is perfectly recoverable from $Y$. Nonetheless, depending on the slope $w$ and the precision $p$, there could be two different $(\Xfp)_1$ and $(\Xfp)_2$ where $(\Xfp)_1 \neq (\Xfp)_2$, but $(\Yfp)_1 = Q(wX_1 + \xi_1) = Q(wX_2 + \xi_2) = (\Yfp)_2$. This would mean $(\Xfp)_1$ is not perfectly recoverable from $(\Yfp)_1$ even in the infinite-$\snr$ case. This additional entropy from the discrete binning in the high-$\snr$ small-$p$ regime cannot be captured by the smooth approximation $\Delta_s(x)$.

Not only do these approximations show close correspondence with the true entropy of floating-point quantized Gaussian random variables, but they can be numerically evaluated quickly for large $p$ and $E$. This contrasts the formulas for the exact entropy in Corollary~\ref{cor:gauss_fp_ent} that require the enumeration of $2^{p+E}$ values. For example, a IEEE-754 single-precision floating point number uses $p=24$ and $E=8$ which would require the enumeration of $2^{32} = 4.3\times 10^9$ values.

\section{Landauer Cost of Exact Linear Regression}\label{sec:exact}

We now will compute the Landauer cost of exact simple linear regression, assuming all continuous variables are implemented on single-precision floating-point numbers, with $E = 8$, $p=24$. The optimal fit line $\wh$ is determined by exact linear regression formula Eq.~\eqref{eq:exact_formula}, and is computed from the entire data set. With each of $n$ data points being independent and identically distributed, if we assume the distribution is sufficiently slow varying for the $p=24$ case (as supported by Fig.~\ref{fig:fp_joint_exact_approx_p_E_snr_100}, where the $p=4$ case will have even finer bins), and that the overflow and underflow probabilities are low, the entropy of the input data can be approximated using Eq.~\eqref{eq:approx_joint_ent_SNR}: 
\begin{equation}\label{eq:input_ent_ex}
    H(\mathcal{D}_{fp}) = nH(\Xfp, \Yfp) \approx n \Ht_s(X_{fp}, Y_{fp}).
\end{equation}

For the output data, $\Wh = \frac{\bX^T\bY}{\bX^T\bX}$ is the predicted slope using the exact formula for linear regression. The differential entropy of the predicted slope is
\begin{align}
h(\Wh) &= h\paren{\frac{\bX^T\bY}{\bX^T\bX}} = h\paren{\frac{\bX^T(w\bX + \bXi)}{\bX^T\bX}} = h\paren{w + \frac{\bX^T\bXi}{\bX^T\bX}} = h\paren{\frac{\bX^T\bXi}{\bX^T\bX}},
\end{align}
since $w$ is a constant. From Lemma~\ref{lem:compute_f_z}, if $Z = \frac{\bX^T\bXi}{\bX^T\bX}$, the probability density function of $Z$ is
\begin{align}
    f_Z(z) = \sqrt{\frac{1}{\pi\paren{\sigx^2}^n\sigxi^2}}\frac{\Gamma\paren{\frac{n+1}{2}}}{\Gamma\paren{\frac{n}{2}}}\paren{\frac{\sigx^2\sigxi^2}{\sigx^2z^2 + \sigxi^2}}^{\frac{n + 1}{2}}\label{eq:exact_final_dist}
\end{align}
which is the probability density function of a scaled Student's t-distributed random variable with $n$ degrees of freedom and a scale of $\frac{\sigxi}{\sigx\sqrt{n}}$ \cite{JohnsonNormanLloyd1994Cud}. 
See Appendix~\ref{sec:compute_f_z} for a proof. Fig.~\ref{fig:z_fit} in App.~\ref{sec:compute_f_z} shows $f_Z(z)$ for different $n$ along with a histogram of the empirical $z$ over 50000 trials. With $\E{Z} = 0$ for $n > 1$ and $\mathrm{Var}(Z) = \frac{\sigxi^2}{\sigx^2}\paren{\frac{1}{n-2}}$ for $n > 2$ \cite{JohnsonNormanLloyd1994Cud}, as $n$ increases the distribution narrows and peaks around zero.
The discrete entropy of $\Wh_{fp}$ is given by Theorem~\ref{thm:quantized_ent} where the cumulative distribution of the scaled and translated Student's t \cite{JohnsonNormanLloyd1994Cud}
\begin{equation}\label{eq:wh_cdf}
F_{\Wh}(\wh) = \frac{1}{2} + \frac{(\wh-w) \sigma_x \Gamma \left(\frac{n+1}{2}\right)}{\sqrt{\pi } \sigma_\xi \Gamma \left(\frac{n}{2}\right)} \, _2F_1\left(\frac{1}{2},\frac{n+1}{2};\frac{3}{2};-\frac{(\wh-w)^2 \sigma_x^2}{\sigma_\xi^2}\right)
\end{equation}
is used and the representable values of the floating-point representation given in Corollary~\ref{cor:gauss_fp_ent}.

In this case, for $n > 2$ and $|w| \gg \frac{\sigxi}{\sigx\sqrt{n-2}}$, $H(\Wh_{fp}) \approx \Ht_s^w(\Wh_{fp})$ is applied, where the resulting entropy is shown in Fig.~\ref{fig:H_fp_wh} as $n$ increases for various values of the ground truth $w$. The quality of the approximation in this case is justified by its offset from 0, and for high $n$ the distribution approximates a delta function (App.~\ref{app:FPN_app3}), as well as by the fine bin size provided by $p=24$. Fig.~\ref{fig:H_fp_wh_exact} in App.~\ref{sec:compute_h_z} shows the quality of the approximation $\Ht^\mu_s(\Wh_{fp})$ for $p=4$ floating point numbers when the exact entropy is computable, that the approximation still applies fairly well even for low $n$, which should only improve as $p$ increases. Taking the difference between the input and output entropies, we find the Landauer cost of exact linear regression as specified by Eq.~\eqref{eq:exact_landauer}, and applying Eqs.~\eqref{eq:first_two_approx} and~\eqref{eq:three_approxes}:
\begin{equation}
    \Delta E_{min}^{Ex} \approx k_BT\ln[2]\paren{n\Ht_s(X_{fp},Y_{fp}) - \Ht^w_{s}(\Wh_{fp})}, \label{eq:landauer_exact_lr}
\end{equation}

This result and the entropies that contribute to it are shown in Fig.~\ref{fig:ex_landauer}. Note from Eqns.~\eqref{eq:three_approxes}, \eqref{eq:single_mean_0}, and~\eqref{eq:approx_joint_ent_SNR} that the precision $p$ of the floating-point representation is a significant contributor to the entropy. In Fig.~\ref{fig:Ex_dH} we plot this contribution alone as well, given by $\Delta E^{Ex}_{p} = k_BT \ln[2](2n - 1)p$.

\begin{figure}[htbp]
  \centering
  \begin{subfigure}{0.4\textwidth}
    \centering
    \includegraphics[width=\linewidth]{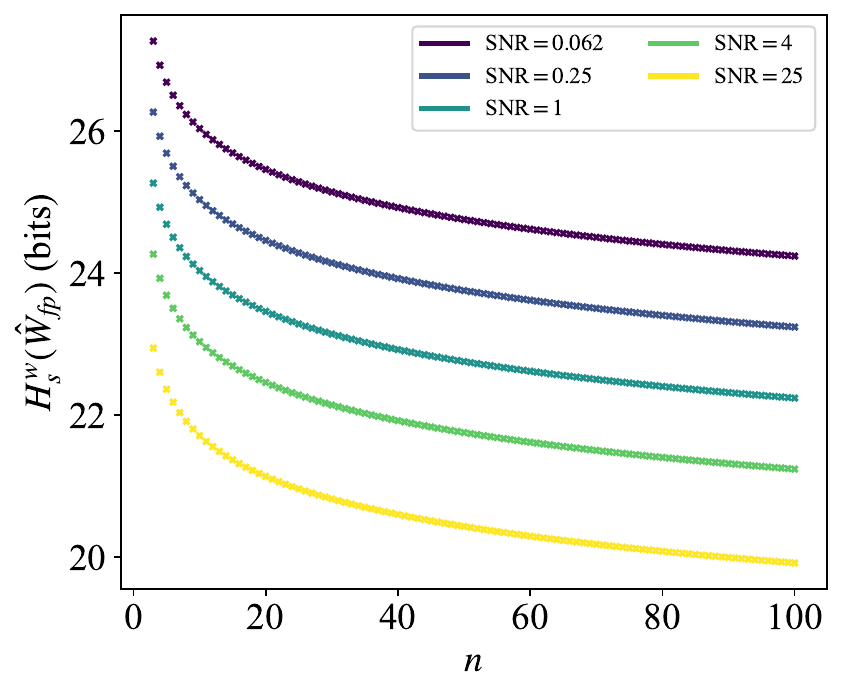}
    \caption{}
    \label{fig:H_fp_wh}
  \end{subfigure}
  \begin{subfigure}{0.4\textwidth}
    \centering
    \includegraphics[width=\linewidth]{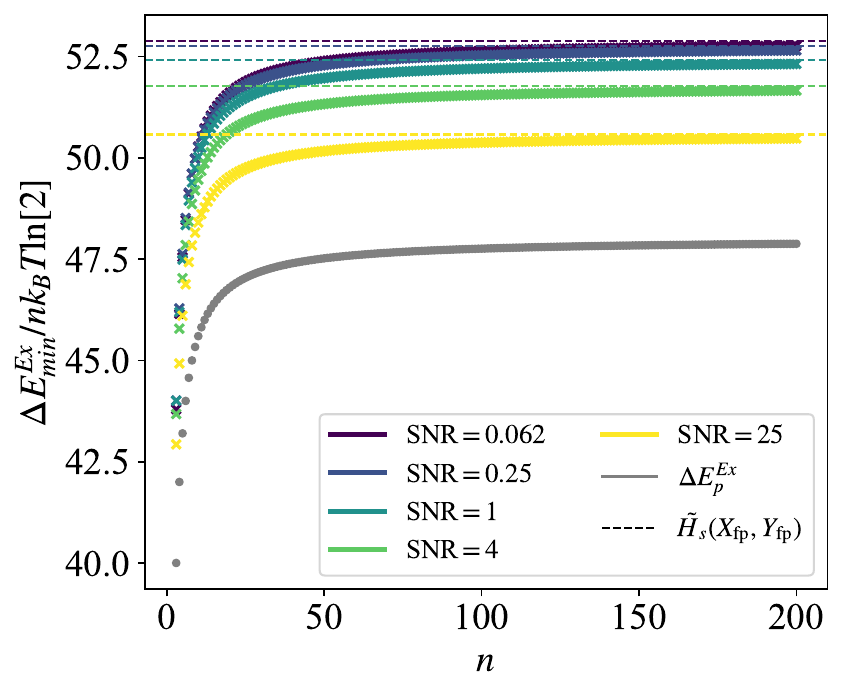}
    \caption{}
    \label{fig:Ex_dH}
  \end{subfigure}\hfill
  \caption{\textit{The Landauer cost for exact zero-intercept simple linear regression.} Input and output states are floating-point numbers with $p = 24$ and $E=8$. The candidate values of the $SNR = w^2\sigx^2/\sigxi^2$ are $0.062$, $0.25$, $1$, $4$, and $25$. (\ref{fig:H_fp_wh}) The output model approximate entropy $\Ht_s^w(\Wh_{fp})$, and its dependence on the number of data samples $n$. (\ref{fig:Ex_dH}) The approximate entropy difference rate between the input data and the output model for various values of $n$. The contribution of the precision of the data $p$ plotted here is $\Delta E^{Ex}_{p}/( k_BT \ln[2]n) =(2n - 1)p/n$ as the gray dots. The joint entropy of the input data is plotted as the dashed horizontal lines for each $SNR$ value.} 
  \label{fig:ex_landauer}
\end{figure}

Fig.~\ref{fig:ex_landauer} shows us the main contributors to the entropy of the output state, and to the overall lower bound on thermodynamic cost. As $n$ increases, the entropy of the output model $\hat w$ decreases (after $n>2$ Fig.~\ref{fig:H_fp_wh}), while the entropy of the input increases with more data. This confirms that thermodynamic costs are higher with more data, and that for large $n$, the entropy of the input data dominates that of the output model. The Landauer cost is nearly linear in $n$ beyond small values with the slope approaching the joint entropy of the input data, as shown in Fig.~\ref{fig:Ex_dH} where Eq.~\eqref{eq:landauer_exact_lr} is scaled by $1/n$ to demonstrate the contribution of the output entropy at low $n$. 

The contribution of the precision $p$ to the Landauer cost is
\begin{equation}
    \label{eq:precision_contribution_Ex}
    \Delta E^{Ex}_{p} \triangleq k_BT \ln[2]\,(2n - 1)p.
\end{equation}
For high $n$, at $p=24$ this is the main contributor to the cost, by approximately $90\%$.\footnote{This effect is even more extreme when analyzing the Landauer cost of averaging and summing normally distributed variables in Appendix~\ref{sec:avg_sum}. In this simpler setting, the only contributors to the Landauer cost are the size of the dataset $n$ and the precision of the data $p$.} Beyond this, the higher the signal-to-noise ratio, the lower the thermodynamic cost.\footnote{For arbitrarily high signal-to-noise ratio, the approximation will eventually break down due to quantization noise as seen in Fig.~\ref{fig:fp_joint_exact_approx_p_E_snr_100}. However, for $p=24$ and the highest SNR considered here being 25, we are well within the bounds of the approximation's applicability.}

\section{Landauer Cost of Linear Regression via Stochastic Gradient Descent}\label{sec:sgd}

For stochastic gradient descent, we again compute the Landauer cost assuming the algorithm is implemented on single-precision floating-point numbers. For stochastic gradient descent, the final distribution $f_{\Wh}(\wh)$ is more difficult to compute directly as was possible in the exact case. However, there are approximations which can allow us to obtain predictive analytic expressions, following \cite{MandtSGDBayes, MandtSGDVariational, Jastrzebski2017ThreeFI}. 

\subsection{Asymptotic Behavior}\label{sec:stationarysgd}

We can begin by considering the expression for $f_{\Wh_\tau}(\hat w_\tau)$ that applies in the case where SGD has been relatively successful after $\tau$ steps, meaning $\hat w_\tau \approx w$. Under the infinite-stream model of Section~\ref{sec:prelim}, the population gradient at $w$ is zero, and the mean-squared error loss function is $\sigma_x^2$-strongly convex so with a sufficiently small step size $\eta$ the iterates concentrate near $w$ \cite{bottou2018optimization}.

Following \cite{MandtSGDBayes, MandtSGDVariational, Jastrzebski2017ThreeFI}, the update rule for $\hat w_k$ given by Eq.~\eqref{eq:SGD_update} can be seen as a discretization of a continuous-time Ornstein-Uhlenbeck process. The limit stationary distribution of the continuous-time Ornstein-Uhlenbeck process can then be analyzed and any time discretization errors are neglected. We can first consider the asymptotic behavior of this process to approximate $f_{\Wh_\tau}(\hat w_\tau)$ as the stationary distribution of this process. Rewriting Eq.~\eqref{eq:SGD_update} with $\ell_i(\wh) \triangleq \ell(\wh, (x_i, y_i))$, $g_i(\hat w_k) = \partial \ell_i(\hat w_k)/\partial \hat w_k = x_i(\hat w_k x_i-y_i)$, and $\hat g_B(\hat w_k) = \frac{1}{B} \sum_{i \in \mathcal{B}} g_i(\hat w_k)$, our update rule becomes
\begin{equation}
    \label{eq:short_SGD}
    \hat w_{k + 1} = \hat w_k - \eta \hat g_B(\hat w_k),
\end{equation}
where $\eta$ is the step size and $k$ is the step number. When $B$ is sufficiently large, following \cite{MandtSGDBayes, Jastrzebski2017ThreeFI}, the central limit theorem can be invoked to assume the gradient noise is normally distributed, therefore we assume 
\begin{equation}
    \hat g_B(\hat w_k) \approx \E{g_i(\hat w_k)} +\frac{1}{\sqrt{B}}\Delta g(\hat w_k), \,\,\, \Delta g(\hat w_k) \sim \mathcal{N}(0,C(\hat w_k)),
\end{equation}
where $C(\hat w)$ is the variance of the gradient.

The batch gradient at the optimum $w$ is $g_i(w) = X_i(wX_i-(wX_i+\xi_i)) = -X_i\xi_i$. 
Therefore, when the step size $\eta$ is small, we can assume $\hat w_\tau \approx w$, and we have $g_i(\wh_\tau) \approx -X_i\xi_i$. With $X_i$ and $\xi_i$ independent, we can assume the gradient variance is thus given by the product of the variances $C(w) = \sigma_x^2\sigma_\xi^2$. Under these approximations, we can rewrite Eq.~\eqref{eq:short_SGD} as a Langevin equation in terms of $\phi=\wh_k - w$
\begin{equation}
    \frac{d\phi}{dt} = - \sigma_x^2 \, \phi + \sqrt{\frac{\eta}{B}} \sigma_x\sigma_\xi\, \epsilon(t),\label{eq:langevin}
\end{equation}
where $\epsilon(t)$ is a Gaussian white noise process with $\E{\epsilon(t)\epsilon(t')} = \delta(t-t')$ and the continuous time variable $t$ is identified with the rescaled discrete SGD step index (i.e., $t = \eta k$). The coefficient $\sigma_x^2$ in front of $\phi$ on the right hand side is given by expected Hessian of the population loss function \cite{MandtSGDBayes}, that is $\partial^2 \E{L(\wh)}/\partial \hat w^2 = \sigma^2_x$. 
In the Langevin equation $\sigx^2$ acts as the effective force constant of a harmonic potential, while $T_{\text{eff}} = \eta \sigma^2_x \sigma^2_\xi/2B$ acts as an effective temperature. Let $f_{\phi}$ be the probability density of $\phi(\tau)$. The resulting stationary distribution is
\begin{equation}
    \label{eq:SGD_final_dist}
    f_{\Wh_\tau}(\wh_\tau)=f_{\phi(\tau)}(\hat w_\tau-w) \approx \sqrt{\frac{B}{\pi \eta \sigma_\xi^2}} \cdot \exp\bigg(- \frac{B (\hat w_\tau - w)^2}{\eta \sigma_\xi^2}
    \bigg).
\end{equation}

This gives us an approximation for the output distribution of regression via SGD. We see that $f_{\Wh_\tau}(\wh_\tau) \sim \gauss{w}{\frac{\eta\sigxi^2}{2B}}$. For $|w| \gg \sqrt{\frac{\eta\sigxi^2}{2B}}$, App.~\ref{app:FPN_app3} tells us $H(\Wh_{fp,\tau}) \approx \Ht^{w}_s(\Wh_{fp,\tau})$. 
Otherwise, for $w \approx 0$ we assume $\tilde H_s(\Wh_{fp,\tau}) \approx \Ht^0_s(p)$.
Again approximating $H(X_{fp},Y_{fp}) \approx \tilde H_s(X_{fp}, Y_{fp})$, we can evaluate the resulting Landauer cost as, 
\begin{equation}
    \Delta E^{SGD}_{min} \approx k_BT \ln[2] \paren{\tau B \Ht_s(\Xfp,\Yfp) - \Ht_s^w(\Wh_{fp,\tau})}.
\end{equation}

Notice in Eq.~\eqref{eq:SGD_final_dist} that the distribution of $\Wh_\tau$ has no dependence on $\tau$. This means that $\tilde H_s(\Wh_{fp,\tau})$ is a constant with respect to $\tau$, while $\tau B \Ht_s(\Xfp,\Yfp)$ grows linearly with $\tau$. This shows that in the asymptotic regime of SGD, the Landauer cost is dominated by the entropy of the input data samples and grows linearly with the number of steps taken by the algorithm.

\subsection{Nonasymptotic Behavior}\label{sec:nonstationarysgd}

Next we can consider the case where SGD has not necessarily converged to the stationary distribution, and consider how the step number affects the thermodynamic cost. In Section~\ref{sec:stationarysgd}, we showed that the asymptotic SGD dynamics can be approximated by the continuous Ornstein-Uhlenbeck process, Eq.~\eqref{eq:langevin}. This picture can further allow us to approximate how the entropy difference evolves throughout the running of the SGD algorithm. In Section~\ref{sec:stationarysgd}, the assumption that $\wh \approx w$ allowed for the derivation of Eq.~\eqref{eq:langevin}. This assumption is no longer valid when the model is initialized far from the global optimum. However, while Fig.~\ref{fig:sgd_w_2x2} confirms that there is large approximation error for small $k$, for the SGD parameters for which the Landauer cost was computed it also shows that Eq.~\eqref{eq:langevin} tracks the dynamics of SGD well once $k \approx 100$.

Let the initial model parameter be initialized deterministically to $\wh_0 \in \R{}$. From \cite{pavliotis2014stochasticprocesses}, the solution to Eq.~\eqref{eq:langevin} allows us to evaluate the distribution over the model parameter at step $k$, again with $\phi = \wh_k - w$
\begin{align}
    \phi \sim \gauss{(\wh_0-w) e^{-\sigx^2\eta k}}{\frac{\eta\sigxi^2}{2B}\paren{1 -e^{-2\sigx^2 \eta k}}}.
\end{align}

Let $\tilde \mu(k) \triangleq w\paren{1- e^{-\sigx^2\eta k}} + \wh_0 e^{-\sigx^2\eta k}$. For $\abs{\tilde \mu} \gg \sqrt{\frac{\eta\sigxi^2}{2B}\paren{1 -e^{-2\sigx^2 \eta k}}}$ and large $k$, the entropy of this distribution is well approximated by $H(\Wh_{fp,k}) \approx \Ht_s^{\tilde \mu(k)}(\Wh_{fp,k})$ from Eq.~\eqref{eq:three_approxes}, by App.~\ref{app:FPN_app3}.

Note that the condition $|\tilde\mu(k)| \gg \sqrt{\eta\sigxi^2/(2B)}$ used below is sufficient, since $\sqrt{\eta\sigxi^2/(2B)} \geq \sqrt{(\eta\sigxi^2/(2B))(1 - e^{-2\sigx^2\eta k})}$ for all $k > 0$. From this we can approximate the Landauer cost after running $k$ iterations of SGD,
\begin{align}
     \Delta E^{SGD}_{min}(k) \approx k_BT \ln[2] \paren{k B \tilde H_s(\Xfp,\Yfp) - \tilde H_s^{\tilde \mu(k)}(\Wh_{fp,k})}.\label{eq:sgd_landauer_noneq}
\end{align}

\begin{figure}[htbp]
\centering
\begin{subfigure}{0.40\textwidth}
    \includegraphics[width=\linewidth]{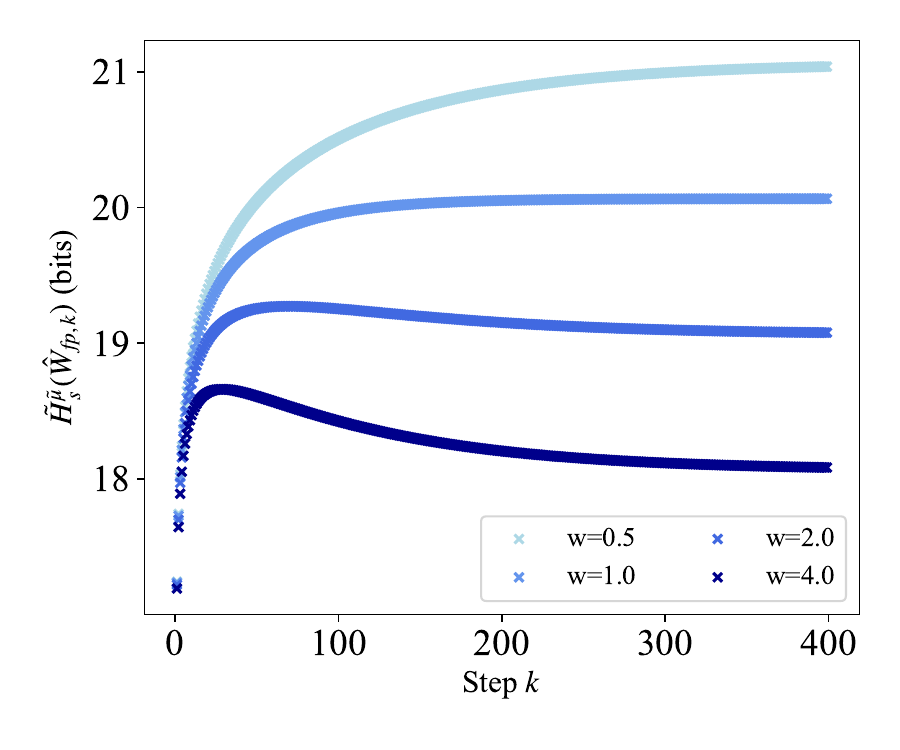}
    \caption{}
    \label{fig:sgd_hfp_vs_n}
\end{subfigure}
\begin{subfigure}{0.40\textwidth}
    \includegraphics[width=\linewidth]{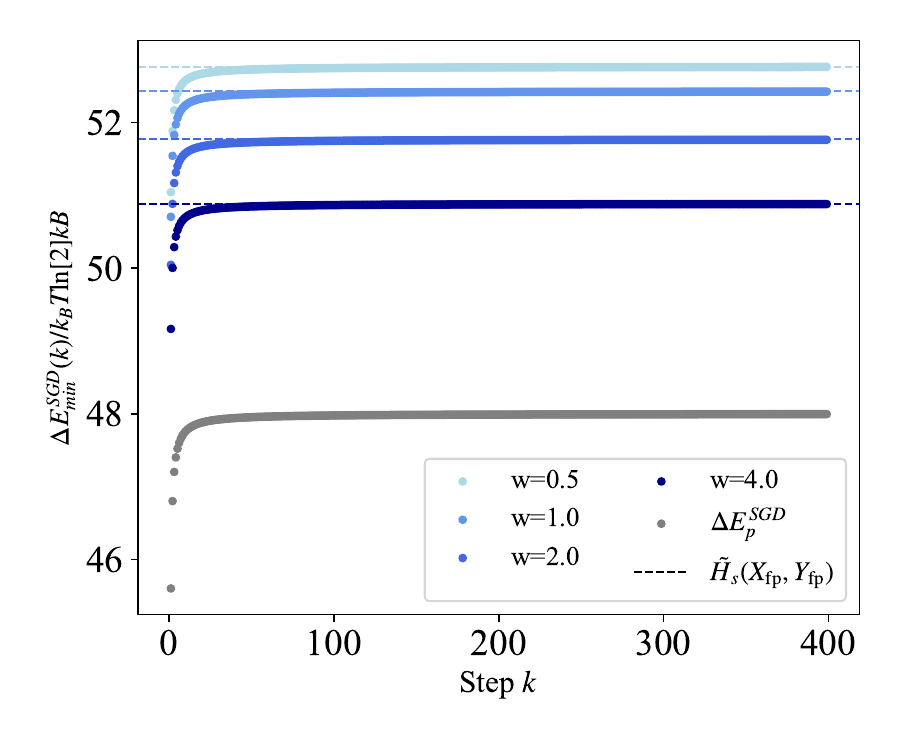}
    \caption{}
    \label{fig:sgd_entropy_diff_rate}
\end{subfigure}
\caption{\textit{Entropy dynamics of SGD.} Input and output states are assumed to be single-precision floating-point numbers, with $p = 24$. Here $\wh_0 = 1$, $\sigx^2 = 1$, $\sigxi^2 = 1$, $\eta = 10^{-2}$, and $B=10$.  (\ref{fig:sgd_hfp_vs_n}) shows the approximate floating-point entropy of $\Wh$ vs SGD step number $k$. (\ref{fig:sgd_entropy_diff_rate}) shows the entropy difference between the input states and the SGD predictor at step $k$. Here the precision contribution is given by $\Delta E^{SGD}_{p}/(k_BT \ln[2]k B) = (2k B - 1)p/(kB)$. The approximate joint entropy of an input sample $\Ht_s(\Xfp, \Yfp)$ are the dashed lines corresponding to each $w$.}
\label{fig:sgd_nonasymp}
\end{figure}

From Fig.~\ref{fig:sgd_entropy_diff_rate} we see how the Landauer cost increases with step number. Similar dependencies on the ground truth $w$ and precision $p$ are seen, comparable with Figs.~\ref{fig:ex_landauer}. Just as in the exact linear regression case in Section~\ref{sec:exact}, we see that the Landauer cost quickly becomes linear with the number of input data samples where the slope of the line is the input data's joint entropy. Again, the precision $p$ of the data is the primary contributor to the Landauer cost.

\section{Energy-cost Aware Scaling Laws}\label{sec:scaling}

Empirical scaling laws \cite{kaplan2020scalinglawsneurallanguage, hoffmann2022trainingcomputeoptimallargelanguage} have become tools for deep learning practitioners to set the best model and dataset size for a fixed training compute budget. 
However, \cite{patterson2021carbonemissionslargeneural} shows that the cost of running inference on a model after it has already been trained is a more significant burden on those who deploy models. 
\cite{sardana2025chinchillaoptimalaccountinginferencelanguage} addresses this by deriving a scaling law that finds the best model and dataset size that will achieve a certain value of pretraining loss while minimizing the total training and inference costs over the lifetime of the model. 

Here, we derive a scaling law in this simple setting of linear regression that finds the optimal dataset size that will maximize profit given prices of energy and inference, and a generalization error dependent user demand for inference \cite{sardana2025chinchillaoptimalaccountinginferencelanguage}. 
With continuous-valued regression, training loss and generalization error are both evaluated using mean-squared error, so we optimize directly over generalization error. 
This correspondence is not guaranteed in language modeling, where the relationship between pretraining loss and downstream performance is more complex \cite{schaeffer_are_2023, du_understanding_2025}. 
Also, in our single-parameter setting, the smooth continuous approximation predicts near-zero Landauer cost for inference when $\wh \neq 0$; however, the exact quantized map need not be injective and therefore need not be logically reversible. 
Although a more complicated multi-parameter regression model setting is necessary to better capture the effect of inference costs, we can still demonstrate tradeoffs between training data size and generalization error in this single-parameter setting.

\subsection{Landauer Cost of Running Inference}
Running inference on the trained model $\wh$, consists of multiplying a data sample independent of the training data $X^{test} \sim \gauss{0}{\sigx^2}$ by $\wh$. This gives $\hat{Y} = \wh X^{test}$. Assuming inference is run on a cyclic device, the input registers now hold $X_I = (X_{fp}^{test}, \wh_{fp})$, while the output state holds $X_O = (\hat Y_{fp}, \wh_{fp})$. Since $\wh$ is fixed during inference, $H(\wh_{fp}) = 0$, and the Landauer cost Eq.~\eqref{eq:Landauer} reduces to
\begin{equation}
    \Delta E_{min}^{inf} = k_BT \ln[2] (H(X^{test}_{fp}) - H(\hat Y_{fp}))
\end{equation}

However, since $\hat{Y} = \wh X^{test} \sim \gauss{0}{\wh^2\sigx^2}$, under our approximation Eq.~\eqref{eq:first_two_approx} and when $\wh \neq 0$, we see
\begin{align}
    \tilde H_s(\Yhfp) = \tilde H_s(X_{fp}^{test}) \approx \Ht^0_{s}(p),
\end{align}
which gives us $\Delta E_{min}^{inf} \approx 0$.
This implies that under the smooth entropy approximation Eq.~\eqref{eq:first_two_approx}, the entropy difference between the input and output is near zero for running inference on this single-parameter model, predicting nearly zero Landauer cost for $\wh \neq 0$. This is expected, since for fixed $\wh \neq 0$ the ideal real-valued map $x \mapsto \wh x$ is bijective, which is why the continuous approximation predicts zero entropy change. However, the actual quantized map $x_{fp} \mapsto Q(\wh x)$ need not be injective, so the implemented discrete computation is not generally logically reversible and can still incur nonzero Landauer cost. For multi-parameter models where the input is not recoverable from the output even at the continuous level, we would expect inference to incur a nonzero Landauer cost.

\subsection{Optimal Dataset Size}

\begin{figure}[bthp]
  \centering
  \begin{subfigure}{0.49\textwidth}
    \centering
    \includegraphics[width=\linewidth]{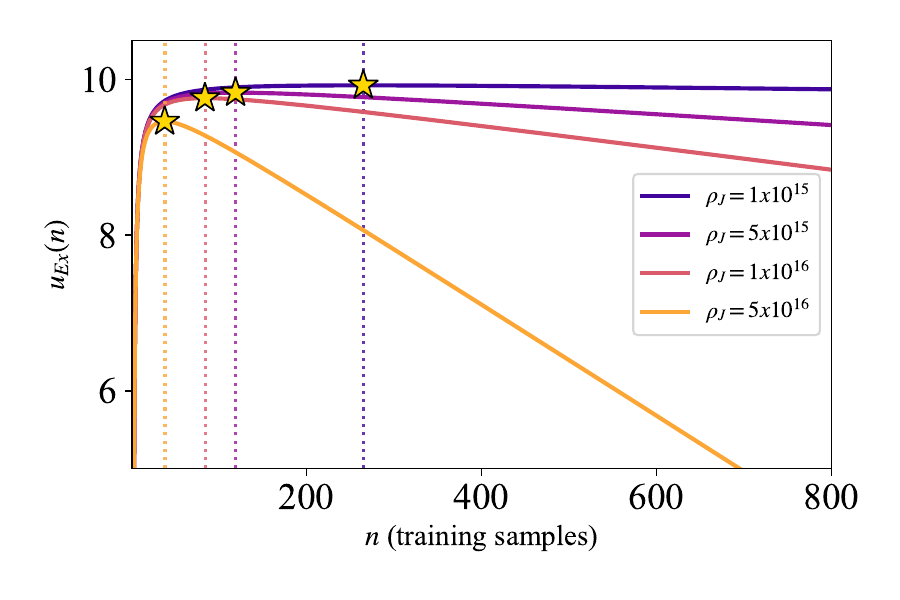}
    \caption{}
    \label{fig:profit_ex}
  \end{subfigure}
  \begin{subfigure}{0.49\textwidth}
    \centering
    \includegraphics[width=\linewidth]{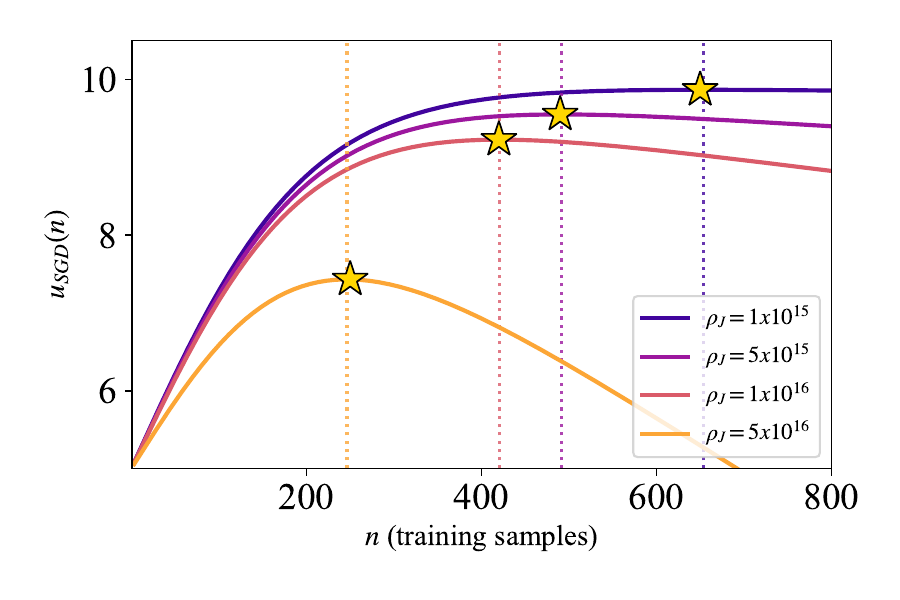}
    \caption{}
    \label{fig:profit_sgd}
  \end{subfigure}
  \begin{subfigure}{0.49\textwidth}
    \centering
    \hspace*{-0.35cm} 
    \includegraphics[width=\linewidth]{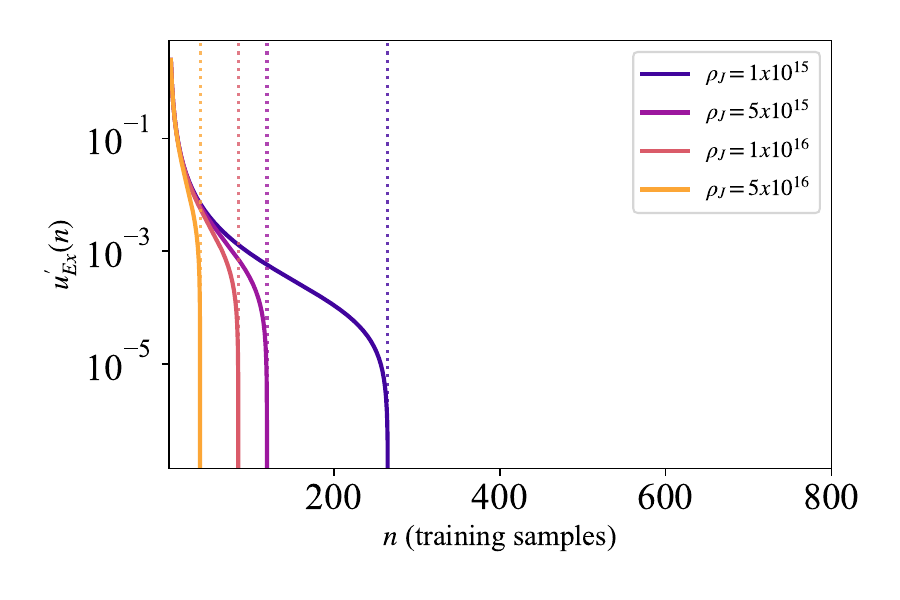}
    \caption{}
    \label{fig:ddn_profit_ex}
  \end{subfigure}
  \begin{subfigure}{0.49\textwidth}
    \centering
    \hspace*{-0.35cm} 
    \includegraphics[width=\linewidth]{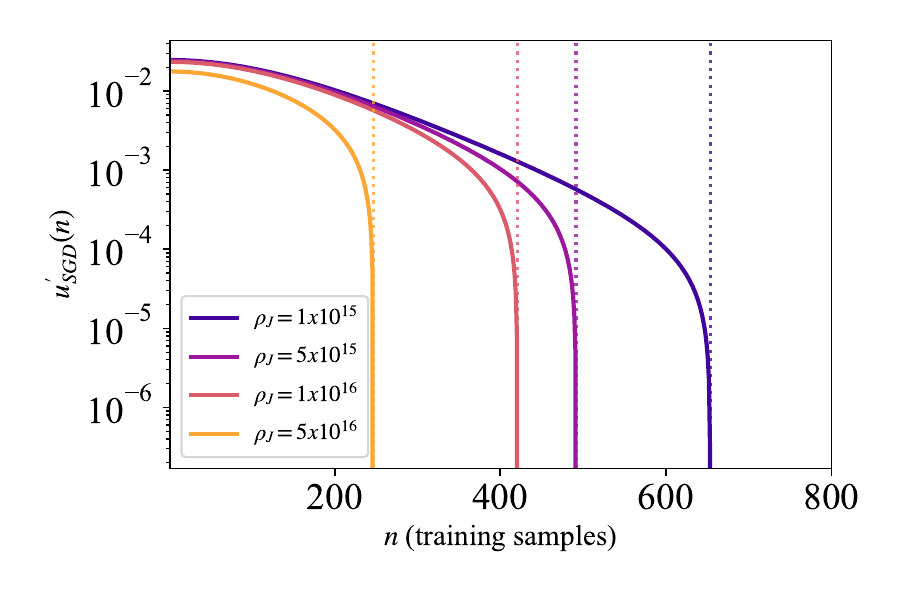}
    \caption{}
    \label{fig:ddn_profit_sgd}
  \end{subfigure}
  \caption{\textit{Optimal dataset size for the exact linear regression formula and for stochastic gradient descent.} (\ref{fig:profit_ex}) shows the profit gained versus the dataset size $n$ for the exact linear regression formula $u_{Ex}(n)$ given in Eq.~\eqref{eq:u_ex}. (\ref{fig:ddn_profit_ex}) shows $u'_{Ex}(n)$, the derivative of the profit function with respect to $n$  as given in Eq.~\eqref{eq:ddn_profit_ex}. (\ref{fig:profit_sgd}) shows the profit gained versus the dataset size $n$ for SGD $u_{SGD}(n)$ given in Eq.~\eqref{eq:u_sgd}. (\ref{fig:ddn_profit_sgd}) shows the derivative with respect to $n$ of the profit gained versus the dataset size for SGD as given in Eq.~\eqref{eq:ddn_profit_sgd}. For the profit plots, the gold stars are the maximum values with respect to each algorithm's feasible set for each value of $\rho_J$. For the derivative plots, the vertical dotted lines show the point where each curve crosses zero. For all figures, $\sigxi^2 =\sigx^2 = 1$, $k_BT = 4\times 10^{-21}$, $w = 2$, $\wh_0 = 1$, and $\rho_I = 10$. For the SGD plots, $\eta = 0.05$ and $B = 10$.}
  \label{fig:scaling_laws}
\end{figure}

Since Lemma~\ref{lem:compute_f_z} proves that $Z = \wh - w$ is a scaled Student's t distributed random variable, the generalization error for the exact formula versus training samples is (where the test pair $(X, Y)$ is drawn independently from the same data-generating distribution and independently of the training data and $\wh$)
\begin{align}
    \mse_{Ex}(n) =\E{(\wh X - Y)^2} =  \E{((\wh-w)X - \xi)^2} = \E{Z^2}\sigx^2 + \sigxi^2 =\sigxi^2\paren{ \frac{n-1}{n-2}},\label{eq:exact_error}
\end{align}
when $n > 2$ and infinite for $n = 1,2$. Note the cross term $\E{2ZX\xi} = 2\E{Z}\E{X}\E{\xi} = 0$ vanishes because $Z = \wh - w$ depends only on the training data, while the test pair $(X, \xi)$ is drawn independently with $\E{X} = \E{\xi} = 0$.

Following the above logic, we can approximate the generalization error of SGD with batch size $B$ versus total training samples $n = kB$ (where $k$ is the number of SGD steps) as follows, with $\phi_k = \wh_k - w$ as defined in Section~\ref{sec:nonstationarysgd},
\begin{align}
\mse_{SGD}(n) = \E{\phi_{n/B}^2}\sigx^2 + \sigxi^2 = \sigx^2\bracket{\frac{\eta\sigxi^2}{2B}\paren{1 -e^{-\frac{2\sigx^2 \eta n}{B}}} + \paren{\tilde \mu\paren{\frac{n}{B}} - w}^2} + \sigxi^2.
\end{align}
As in Eq.~\eqref{eq:exact_error}, the cross term $\E{2\phi_{n/B} X \xi} = 0$ vanishes because $\phi_{n/B}$ depends only on training randomness while the test pair $(X, \xi)$ is drawn independently with $\E{X} = \E{\xi} = 0$.

Let $d(\mse)$ be the inference demand given the model's generalization performance measured by mean-squared error. $d(\mse)$ will be a function that decreases as the mean-squared error increases and vice versa. Let $\rho_{J}$ be the price of energy, and $\rho_{I}$ be the price per sample that a user would pay to run inference on the model. We can now set up an optimization problem to find the optimal training dataset size $n^\star \in \N{}$:
\begin{align}
n^\star = \arg\max_{n \in \{3,4,\ldots\}} \left\{ \rho_{I}d(\mse(n)) - \rho_{J}\Delta E_{min}(n) \right\} \label{eq:scaling_opt}
\end{align}
where $\rho_{I}d(\mse)$ is the revenue gained from deployment, and $\rho_{J}\Delta E_{min}\paren{n}$ is the cost of energy for training. For exact regression $n \in \{3, 4, \ldots\}$; for SGD $n = kB$ with $k \in \mathbb{N}$. Appendix~\ref{sec:scaling_laws_app} specializes the feasible set accordingly.

To illustrate what an $n^\star$ might look like in our simple setting, let us assume there is a constant price per sample $\rho_{I}$ and furthermore, let $d(\mse) = \frac{1}{\mse}$. A more complicated price model might have a price for inference that increases as the user demand for inference increases, leading to a larger $n^\star$ than what will be derived in this section.

We can define the objective for each learning algorithm, $u_{Ex}(n)$ for exact linear regression and $u_{SGD}( n)$ for SGD as defined in App.~\ref{sec:scaling_laws_app}. Notice that for both objectives, as $\rho_J$ increases, the optimal dataset size $n^\star$ decreases. In general, the value of $n^\star$ will depend on the units of $\rho_I$, $\rho_J$ and $k_B T$, so we present Fig.~\ref{fig:scaling_laws} just as a demonstration of the trends visible in this trade-off. App.~\ref{sec:scaling_laws_app} presents derivations for solving a continuous relaxation of Eq.~\eqref{eq:scaling_opt} numerically, the results are presented in Fig.~\ref{fig:scaling_laws} in this section.

 When we set the constants that are shared between the exact formula and SGD equal to each other, we can compare between their optimal dataset sizes according to Eq.~\eqref{eq:nstar_ex} and Eq.~\eqref{eq:nstar_sgd}. This comparison serves as a one-parameter theoretical baseline that isolates how two training procedures for the same learning task differ in Landauer cost. In Figs.~\ref{fig:profit_ex} and \ref{fig:profit_sgd} we see that for each $\rho_J$ the optimal $n$ for the exact formula is less than or equal to the optimal $n$ for SGD. Nevertheless, the learning rate $\eta$ and the batch size $B$ are two hyperparameters for SGD that are not present in the exact linear regression algorithm. The effect of these hyperparameters on the optimal dataset size for SGD is shown in Fig.~\ref{fig:scaling_laws_sgd_eta_B}.
Although this simple profit model already displays differences between the exact method and SGD, more complicated profit models could also include other aspects of these algorithms that affect profitability, like their memory footprint and computational speed.

\FloatBarrier
\section{Bounding the Mismatch Cost of a Continuous Variable Algorithm}\label{sec:mmc}

Now we turn to the mismatch cost (MMC). While Landauer's bound assumes thermodynamic reversibility, there will in general be additional costs beyond $Q = -T \Delta S_{sys}$. Beyond this bound, the average heat flow $Q$ will depend on physical details such as the speed with which bits are manipulated \cite{Seifert2007, Sivak2012}, or if the bits are instantiated by physical systems that continuously dissipate heat \cite{Esposito2021}. A clarifying equivalent definition of a thermodynamically reversible process in terms of the second law of thermodynamics can help us understand energetic costs beyond reversibility. The second law of thermodynamics states \textit{total entropy production} $\Delta S_{tot}$ must be non-negative, defined as the total change in entropy of the system $\Delta S_{sys}$ and its environment $\Delta S_{env}$, $\Delta S_{tot} \triangleq \Delta S_{env} + \Delta S_{sys} \geq 0$. A process is thermodynamically reversible when $\Delta S_{tot} = 0$ \cite{Callen,landau1980statistical, chandler}. Mismatch cost quantifies the additional total entropy production that arises when the distribution of the system at the start of the computation differs from the optimal starting distribution that minimizes total entropy production. 

Here, we discuss a method to lower bound the MMC in the case of parameterized continuous input distributions. While $\Delta S_{sys}$ can be assessed at the algorithmic level based on the analysis above, $\Delta S_{env}$ and $\Delta S_{tot}$ will depend on the explicit time-dependent dynamics of the physically implemented computational system \cite{Wolpert2019, KolchinskyWolpert2021, Yadav2025}. Because of this, we cannot compute the MMC directly without more physical detail beyond the scope of this paper. Instead, this section serves the purpose of
demonstrating a bound on the MMC for computational systems that take a parameterized
continuous distribution as input (such as with linear regression), only given an assumed $\Delta S_{env}$, while leaving the true determination of $\Delta S_{env}$ to future work. While it remains undetermined, Appendix~\ref{app:mmc} discusses common assumptions which allow one to conceptually understand the assumed $\Delta S_{env}$, associating it with heat by $\Delta S_{env} = Q/T$. Under these assumptions, we see $Q = Q_{min} + T\Delta S_{tot} = -T\Delta S_{sys} + T\Delta S_{tot}$, meaning $\Delta S_{tot}$ can be understood as contributing to the energetic cost beyond the minimum implied by the Landauer cost.

With $\compsimplex$ as the probability simplex over the joint logical state of the computer $\mathcal{X}_c$, we can think of the computation as taking us from $p_1 \in \compsimplex$ to $p_{F-1} \in \compsimplex$ (with $p_s$ being the probability density at step $s$) where $p_{F-1}$ is fixed by a conditional distribution $\pi(x_c|x_c')$, where
\begin{equation}
    p_{F-1}(x_c) = \sum_{x_c'\in \mathcal{X}_c} \pi(x_c|x_c') p_1(x_c').
\end{equation}
The conditional distribution $\pi(x_c|x_c')$ can be interpreted as resulting from a physical process, as discussed in App.~\ref{app:mmc}. We can use the notation $p_{F-1} = \pi p_1$ to emphasize that $p_{F-1}$ is purely a function of the input state $p_1$ and the physical manipulations performed by the computation which determine $\pi(x_c|x_c')$.

Assuming a fixed physical implementation of the algorithm defining $\pi(x_c|x_c')$, a portion of the total entropy production $\Delta S_{tot}$ will be due to the MMC, which results from a divergence between the actual input distribution $p_1$ and the optimal initial distribution $q_1 \in \compsimplex$ which minimizes entropy production. With $\Delta S_{sys}(p_1) = k_B \ln[2] (H(\pi p_1) - H(p_1))$ the optimal input distribution $q_1$ defined by the minimum of
\begin{equation}
    \label{eq:mmc_exact_opt}
    q_1 \triangleq \arg \min_{r_1 \in \compsimplex}[\Delta S_{tot}(r_1)] = \arg \min_{r_1\in \compsimplex}[k_B \ln[2] \big (H(\pi r_1) - H(r_1) \big )+\Delta S_{env}(r_1)].\footnote{The optimum $q_1$ is defined over the full simplex and may assign nonzero probability to states with $x_M\neq x_M^0$ or $x_O\neq x_O^0$. Since at step $0$ $x_M = x_M^0$ and $x_O = x_O^0$ this loading can still be done cost free. If the actual input distribution $p_1$ has the restricted form $p(x_I)\delta_{x_M,x_M^0}\delta_{x_O,x_O^0}$ as imposed by the standard accounting convention, $q_1$ can minimize $\Delta S_{tot}$ over a superset while bound $\mathrm{MMC}(p_1)\geq 0$ remains valid.}
\end{equation}
Kolchinsky and Wolpert show that given a distribution with minimal entropy production $q_1$,
\begin{equation}
    \Delta S_{tot}(p_1) = \Delta S_{tot}(q_1) + k_B\ln[2] \big(D(p_1||q_1) - D(\pi p_1|| \pi q_1) \big ),
\end{equation}
in \cite{Kolchinsky2016DependenceOD, Wolpert_2020, KolchinskyWolpert2021} (where $D(p||q) = \sum_{x_c \in\mathcal{X}_c}p(x_c) \log[p(x_c)/q(x_c)]$ is the Kullback-Leibler (KL) divergence). The MMC is defined by
\begin{equation}
    \label{eq:MMC}
    \text{MMC}(p_1) \triangleq k_B \ln[2] \big ( D(p_1||q_1) - D(\pi p_1||\pi q_1) \big ) = \Delta S_{tot}(p_1) - \Delta S_{tot}(q_1),
\end{equation}
There may still be non-zero minimal entropy production for the optimal initial state $\Delta S_{tot}(q_1)$ called the `residual entropy production', however its character depends on further physical details which we ignore here \cite{Wolpert2019, Yadav2025}.

\textit{Continuous Variable Algorithms: }Several barriers arise for analyzing the MMC in a continuous valued context. Finding $q_1$ requires an optimization across the entire probability simplex of possible distributions $\compsimplex$. However, this becomes intractable for reasonably sized discretely represented continuous variables: a computational system with just one single-precision floating-point number contains $2^{32}$ states. Additionally, the input to such an algorithm is often parameterized by a variational continuous probability density such as Eq.~\eqref{eq:initial_dist}, restricting the accessible subset of $\compsimplex$. 

Instead, we can show that a positive lower bound on the MMC can be determined purely by considering a restricted set of variational input distributions. This can be specified by restricting the optimization to a sub-family of distributions specified by the quantization of a variational continuous probability density such as Eq.~\eqref{eq:initial_dist}. Let $\mathcal{V} = \{p_{I,\theta}(x_c)\delta_{x_M,x_M^0} \delta_{x_O,x_O^0}, \theta \in \Theta \} \subseteq \compsimplex$, where $\Theta$ corresponds to a bounded subregion of possible variational parameters, and $p_{I, \theta}(x_c)$ is the input distribution induced by quantizing the distribution parameterized by $\theta$. For example, in regression we can take Eq.~\eqref{eq:initial_dist} as our variational family parameterized by $\theta = \{\sigma_x,\sigma_\xi\}, \Theta = [\sigma_x^{min}, \sigma_x^{max}] \times [\sigma_\xi^{min}, \sigma_\xi^{max}] \subset \mathbb{R^+} \times \mathbb{R^+}$ (perhaps including a distribution over the initial SGD parameter $\Wh_0$). From this we can define the \textit{variational MMC}, MMC$_v$, as
\begin{equation}
    \text{MMC}_v(p_1) = \Delta S_{tot}(p_1) - \Delta S_{tot}(q_{1,v}),
\end{equation}
where $q_{1,v} \in \mathcal{V}$. The optimal variational distribution $q_{1,v}$ is found by an optimization over $\mathcal{V}$, in which the parameters $\theta$ are varied through the restricted subregion $\Theta$:
\begin{equation}
    \label{eq:mmc_var_opt}
    q_{1,v} \triangleq \arg \min_{r_{1,v} \in \mathcal{V}}[\Delta S_{tot}(r_{1,v})] = \arg \min_{r_{1,v} \in \mathcal{V}}[k_B \ln[2] \big (H(\pi r_{1,v}) - H(r_{1,v}) \big )+\Delta S_{env}(r_{1,v})]
\end{equation}

We can show that MMC$_v$ provides a positive lower bound on the true MMC of a variational input distribution $p_{1,v} \in \mathcal{V}$. First, we can note that since $\mathcal{V} \subseteq \compsimplex$, the variational optimum $q_{1,v}$ must have a larger or equal entropy production than the true optimum is $q_1$, meaning that $q_{1,v}$ itself must have a mismatch cost
\begin{equation}
    \text{MMC}(q_{1,v}) = k_B \ln [2] \big (D(q_{1,v}||q_1) - D(\pi q_{1,v}||\pi q_1) \big) = \Delta S_{tot}(q_{1,v}) - \Delta S_{tot}(q_1).
\end{equation}
Since $q_1$ minimizes $\Delta S_{tot}$ over the full simplex $\compsimplex$ and $q_{1,v} \in \mathcal{V} \subseteq \compsimplex$, we have
\begin{equation}
    \label{eq:qvqstar}
    \Delta S_{tot}(q_{1,v}) \geq \Delta S_{tot}(q_1) \geq 0.
\end{equation}
Subtracting $\Delta S_{tot}(p_1)$ from Eq.~\eqref{eq:qvqstar}, we see that for any distribution $p_1 \in \compsimplex$, MMC$_v$ forms a lower bound on the true MMC
\begin{equation}
    \text{MMC}_v(p_1) = \Delta S_{tot}(p_1) - \Delta S_{tot}(q_{1,v}) \leq \Delta S_{tot}(p_1) - \Delta S_{tot}(q_1) = \text{MMC}(p_1).
\end{equation}
In general, the variational MMC$_v$ can be negative if $\Delta S_{tot}(p_1) < \Delta S_{tot}(q_{1,v})$. However, this is only an issue if $p_1 \notin \mathcal{V}$. If $p_1 = p_{1,v} \in \mathcal{V}$, then by definition of $q_{1,v}$ as the variational optimum for $\mathcal{V}$ we know $\Delta S_{tot}(p_{1,v}) \geq \Delta S_{tot}(q_{1,v})$. This shows us that MMC$_v$ provides a positive lower bound on the total MMC of a variational input:
\begin{equation}
    \label{eq:var_bound}
    0 \leq \text{MMC}_v(p_{1,v}) \leq \text{MMC}(p_{1,v}).
\end{equation}

Wolpert and Kolchinsky show that while the optimized $q_1$ is not necessarily unique, it will have a unique MMC with any other distribution $p_1$\cite{Wolpert_2020}. It is important to note that while this uniqueness is broken by the variational bound as discussed in App.~\ref{app:mmc}, the lower bound Eq.~\eqref{eq:var_bound} continues to apply even for a non-unique variational optimum $q_{1,v}$. 

While without more specific physical details the true distribution of MMC$_v$ cannot be determined, App.~\ref{app:mmc} further discusses exact and SGD linear regression as illustrative examples.

\section{Discussion}\label{sec:disc}

In this work, we studied the thermodynamic costs of single-parameter linear regression represented on floating-point registers, comparing the exact analytic solution with stochastic gradient descent.
For this simple model, the Landauer cost is dominated by the size of the data, and the main contributor to this cost is the precision $p$, which determines the number of bits included in the mantissa. 
From this we see that mantissa bits are thermodynamically expensive, exponent bits are thermodynamically cheap. Fig.~\ref{fig:fp_exact_approx_sigma} shows that adding exponent bits only increases the range of representable values while leaving the entropy relatively unchanged once overflows and underflows are avoided. 
This implies that for many data sets, additional exponent bits will barely affect thermodynamic costs. 
Additionally, these results provide a theoretical foundation for empirical results on machine learning model compression via entropy coding \cite{bordawekar2022efloatentropycodedfloatingpoint, hao_neuzip_2024}.

The low thermodynamic cost of exponent bits also aligns with other innovations in machine learning systems, like the use of the bfloat16 number format \cite{googleBFloat16Secret} which has the same number of exponent bits as FP32, but only 7 mantissa bits as opposed to 23. 
Another important innovation is model quantization \cite{gholami2021surveyquantizationmethodsefficient, gupta2015deeplearninglimitednumerical, krishnamoorthi2018quantizingdeepconvolutionalnetworks} where deep learning models can be trained with very limited numerical precision, especially when using techniques like quantization aware training \cite{ash2024efqat, hasan2024optimizinglargelanguagemodels}.

Additionally, the signal-to-noise ratio of the data contributes significantly to the Landauer cost, as seen from Eq.~\eqref{eq:approx_joint_ent_SNR} and in Figs.~\ref{fig:ex_landauer} and~\ref{fig:sgd_nonasymp} where increasing the signal-to-noise ratio decreases the thermodynamic cost. This implies that less noisy data and more structured data may also yield lower fundamental energy costs.

We also derived scaling laws for the exact formula and SGD that demonstrate the trade-off between generalization error and the energy cost of training. When the demand for inference scales inversely with the generalization error and the price for inference is constant, the irreducible noise creates a threshold where training on more data will not increase the model's accuracy enough to justify the associated energy costs.

Finally, we presented a method for lower bounding the mismatch cost entropy production for variational input distributions. 
Future work will include studying specific thermodynamic models of bit implementations like CMOS \cite{CMOS}, such that an entropy flow to the environment $\Delta S_{env}$ can be determined and the mismatch cost specified.

The natural extension of this work is to multi-parameter models, where connections to kernel ridge regression 
via the Neural Tangent Kernel \cite{jacot2020neuraltangentkernelconvergence,karkada2025predictingkernelregressionlearning, pmlr-v80-belkin18a} suggest our analysis may generalize to certain regimes of MLP training. In these cases, the output entropy of the models may be more significant than in the single-parameter case. 
% For example, an overparameterized model may retain more information and have a correspondingly higher entropy output state, and thus a lower thermodynamic cost. 
% Furthermore, in the setting of overparameterized models, inference may incur a nonzero Landauer cost. 
% This would require inference costs to be included in any scaling laws that could be derived for multi-parameter models.

\bibliographystyle{IEEEtran}
\bibliography{thermo_lin_reg}
\newpage
\appendices
\FloatBarrier

\section{Summary Tables}\label{app:summary_tables}

\subsection{Standard Accounting Convention: Cyclic-Device Protocol}\label{app:cyclic_device_table}

Table~\ref{tab:cyclic_device} summarizes the standard accounting convention used throughout this paper. Under this protocol, the Landauer cost Eq.~\eqref{eq:Landauer} gives the unconstrained (``all-at-once'') lower bound on the energetic cost associated with the endpoint distributions $p_1$ and $p_{F-1}$. It does not include the costs of intermediate garbage erasure, intermediary-register resets, or time-resolved control operations. Any such implementation details can only increase the physical energetic cost above this lower bound.

Additionally, we assume each computational state $x_c$ has equal internal entropy $S_0$, that the computation begins and ends with the same average internal entropy $U_0$, and that it remains in thermal equilibrium with its environment at temperature $T$ throughout the computation.

\begin{table}[htbp]
\centering
\caption{Cyclic-device protocol constraints, where the joint computational state is $x_c = (x_I, x_M, x_O) \in \mathcal{X}_c$.}
\label{tab:cyclic_device}
\begin{tabular}{c l l l}
\toprule
\textbf{Step} & \textbf{Input $x_I$} & \textbf{Intermediary $x_M$} & \textbf{Output $x_O$} \\
\midrule
$0$ & Fixed at $x_I^0$ & Fixed at $x_M^0$ & Fixed at $x_O^0$ \\
$1$ & Loaded: $p(x_I)$ & Fixed at $x_M^0$ & Fixed at $x_O^0$ \\
$F{-}1$ & Reset to $x_I^0$ & Reset to $x_M^0$ & Computed: $p(x_O)$ \\
$F$ & Fixed at $x_I^0$ & Fixed at $x_M^0$ & Offloaded, reset to $x_O^0$ \\
\bottomrule
\end{tabular}
\vspace{4pt}

\noindent \footnotesize{Steps $0 \to 1$ (input loading) and $F{-}1 \to F$ (output offloading) are energetically cost free \cite[p. 34]{Wolpert2019}. Only steps $1 \to F{-}1$ incur an energetic cost, yielding $\Delta E_{min} = k_BT\ln[2](H(X_I) - H(X_O))$.}
\end{table}

\subsection{Approximation Roadmap}\label{app:approx_audit}

Table~\ref{tab:approx_audit} provides a compact summary of each approximation used in this paper: its mathematical status, the assumptions required, where it is validated, and where it is used downstream.

\begin{table}[htbp]
\centering
\caption{Roadmap for entropy approximations. ``Exact'' means a closed-form expression without approximation; ``Theorem-backed'' means an error bound is proved; ``Asymptotic'' means valid in a stated limit with numerical support.}
\label{tab:approx_audit}
\small
\begin{tabular}{p{3.0cm} p{1.6cm} p{3.6cm} p{2.4cm} p{2.4cm}}
\toprule
\textbf{Quantity / Eq.} & \textbf{Status} & \textbf{Assumptions} & \textbf{Validated by} & \textbf{Used in} \\
\midrule
Exact $H(\Xfp)$, $H(\Xfp, \Yfp)$ \newline (Thm.~\ref{thm:quantized_ent}, Cors.~\ref{cor:gauss_fp_ent}, \ref{cor:joint_ent}) & Exact & None & Definition & Figs.~1--2; benchmarks for Approx.~1--3 \\
\addlinespace
Approx.~1: $H(X_Q) \approx \tilde H(X_Q)$ \newline Eq.~\eqref{eq:diff_approx} & Asymptotic & $f_X$ slowly varying relative to $\Delta(x)$ & \cite{kostina_data_2017}; App.~\ref{app:FPN}, Figs.~\ref{fig:fp_exact_approx}, \ref{fig:full_fp_exact_approx_p_E} & All downstream expressions \\
\addlinespace
Approx.~2: $\tilde H(X_Q) \approx \tilde H_s(X_{fp})$ \newline Eq.~\eqref{eq:first_two_approx} & Theorem-backed & (a) Bin-smoothing: error $\leq d/2 + \varepsilon_0$ (Thm.~\ref{thm:steps}); (b) Domain extension: $\Pr(\text{overflow/underflow}) \approx 0$, i.e.\ $E \geq 4$ & Thm.~\ref{thm:steps}; Figs.~\ref{fig:fp_exact_approx}, \ref{fig:full_fp_exact_approx_p_E} & All downstream expressions \\
\addlinespace
Approx.~3: $\tilde H_s(X_{fp}) \approx \tilde H_s^\mu(X_{fp})$ \newline Eq.~\eqref{eq:three_approxes} & Asymptotic & $|\mu| \gg \sigma$; distribution concentrated around $\mu \neq 0$ & App.~\ref{app:FPN_app3}; Fig.~\ref{fig:fp_exact_approx_mu} & Output entropies (Secs.~IV, V) \\
\addlinespace
$\tilde H_s^0(p)$: zero-mean Gaussian \newline Eq.~\eqref{eq:single_mean_0}, Thm.~\ref{thm:single_gaussian} & Exact (given Approx.~1--2) & Approx.~1--2 hold; $X \sim \mathcal{N}(0, \sigma^2)$ & Thm.~\ref{thm:single_gaussian}; Fig.~\ref{fig:fp_exact_approx_sigma} & Input entropies; inference cost \\
\addlinespace
$\tilde H_s(\Xfp, \Yfp)$: bivariate Gaussian \newline Eq.~\eqref{eq:approx_joint_ent_SNR}, Thm.~\ref{thm:double_gaussian} & Exact (given Approx.~1--2) & Approx.~1--2 hold; $\mathrm{SNR}$ not too large (see Fig.~\ref{fig:fp_joint_exact_approx_p_E_snr_100}) & Thm.~\ref{thm:double_gaussian}; Figs.~\ref{fig:fp_joint_exact_approx_p_E_snr_2}, \ref{fig:fp_joint_exact_approx_p_E_snr_100} & Input entropy for training cost \\
\addlinespace
Exact LR output: $\tilde H_s^w(\Wh_{fp})$ \newline Eq.~\eqref{eq:landauer_exact_lr} & Asymptotic & Approx.~1--3; $n > 2$; $|w| \gg \sigxi/(\sigx\sqrt{n-2})$; $E \geq 4$ & Fig.~\ref{fig:ex_landauer} & Sec.~IV, Sec.~VI \\
\addlinespace
SGD stationary: $\Wh_\tau \sim \mathcal{N}(w, \eta\sigxi^2/(2B))$ \newline Eq.~\eqref{eq:SGD_final_dist} & Numerically validated & Large $B$ (CLT); small $\eta$; $\wh \approx w$; OU continuous-time limit; discretization error neglected & Figs.~10, 11 & Sec.~V-A \\
\addlinespace
SGD nonasymptotic: $\tilde H_s^{\tilde\mu(k)}(\Wh_{fp,k})$ \newline Eq.~\eqref{eq:sgd_landauer_noneq} & Numerically validated & All of the above; $|\tilde\mu(k)| \gg \sqrt{\eta\sigxi^2/(2B)}$; $k$ large enough for OU to track SGD ($k \gtrsim 100$ empirically) & Fig.~\ref{fig:sgd_w_2x2} & Sec.~V-B, Sec.~VI \\
\bottomrule
\end{tabular}
\end{table}

\FloatBarrier
\pagebreak
\section{Computing the exact entropy of quantized random variables}\label{app:exact_entropy_float}

Here we compute the exact discrete entropy of quantized random variables, such as normally distributed random variables stored as floating point numbers.

\begin{theorem}[Entropy of a Clipped and Arbitrarily Midpoint Quantized Random Variable]\label{thm:quantized_ent}
Let $X$ be an absolutely continuous random variable with cumulative distribution function $F$. Assume there are $K$ representable values in the quantization scheme and denote the set of these values as $\{u_1, \dots, u_K\}$, where $u_1 < \dots < u_K$. If $X_Q$ is the resulting clipped and midpoint quantized representation of $X$, then the discrete entropy of $X_Q$ is
\begin{equation*}
    \begin{aligned}
        H(X_Q) &= -F\paren{\frac{u_1 + u_2}{2}}\log\bracket{F\paren{\frac{u_1 + u_2}{2}}}\\
        &- \sum_{i = 2}^{K-1}\bracket{F\paren{\frac{u_{i+1} + u_i}{2}} - F\paren{\frac{u_{i} + u_{i-1}}{2}}}\log\bracket{F\paren{\frac{u_{i+1} + u_i}{2}} - F\paren{\frac{u_{i} + u_{i-1}}{2}}}\\
        &-\bracket{1 - F\paren{\frac{u_{K-1} + u_K}{2}}}\log\bracket{1 - F\paren{\frac{u_{K-1} + u_K}{2}}}.
        \end{aligned}
\end{equation*}
when $K \geq 3$. When $K = 2$, $H(X_Q) = -F\paren{\frac{u_1 + u_2}{2}}\log\bracket{F\paren{\frac{u_1 + u_2}{2}}} - \bracket{1 -  F\paren{\frac{u_1 + u_2}{2}}}\log\bracket{1 -  F\paren{\frac{u_1 + u_2}{2}}}$, and when $K = 1$, $H(X_Q) = 0$.

\end{theorem}

\begin{proof}
Since $X_Q$ is a truncated and midpoint quantized representation of $X$, we have,
\begin{equation}
    X_Q = \begin{cases}
    \sum\limits_{i = 2}^{K-1} u_i\indic{X \in \left[\frac{u_{i} + u_{i-1}}{2},\frac{u_{i+1} + u_i}{2}\right)} + u_1\indic{X < \frac{u_1 + u_2}{2}} + u_K\indic{X \geq \frac{u_{K-1} + u_K}{2}} \text{ if } K \geq 3\\
    u_1\indic{X < \frac{u_1 + u_2}{2}} + u_2\indic{X \geq \frac{u_1 + u_2}{2}} \text{ if } K = 2\\
    u_1 \text{ if } K = 1.
    \end{cases}
\end{equation}

For $K = 1$, we see that $P\{X_Q = u_1\} = 1$ so $H(X_Q) = 0$. For $K = 2$, $P\{X_Q = u_1\} = F\paren{\frac{u_1 + u_2}{2}} = 1 - P\{X_Q = u_2\}$, so 
\begin{equation}
H(X_Q) = -F\paren{\frac{u_1 + u_2}{2}}\log\bracket{F\paren{\frac{u_1 + u_2}{2}}} - \bracket{1 -  F\paren{\frac{u_1 + u_2}{2}}}\log\bracket{1 -  F\paren{\frac{u_1 + u_2}{2}}}.
\end{equation}

Finally, for $K \geq 3$,

\begin{equation}
\begin{aligned}
    P\{X_Q = u_i\} &=  \begin{cases} \int\limits_{\frac{u_{i} + u_{i-1}}{2}}^{\frac{u_{i+1} + u_i}{2}} f_X(x)dx \text{ if } i \in [2, K-1]\\
    \int\limits_{-\infty}^{\frac{u_1 + u_2}{2}} f_X(x)dx \text{ if }  i = 1\\
    \int\limits_{\frac{u_{K-1} + u_K}{2}}^{\infty} f_X(x)dx \text{ if }  i = K\\ 
    \end{cases}\\
    &=\begin{cases} F\paren{\frac{u_{i+1} + u_i}{2}} - F\paren{\frac{u_{i} + u_{i-1}}{2}} \text{ if } i \in [2, K-1]\\
    F\paren{\frac{u_1 + u_2}{2}}  \text{ if } i = 1\\
    1 - F\paren{\frac{u_{K-1} + u_K}{2}} \text{ if } i = K
    \end{cases}.
\end{aligned}
\end{equation}

Inserting this into the formula for discrete entropy allows us to derive the following
\begin{equation}
    \begin{aligned}
        H(X_Q) & = - \sum_{i = 1}^K P\{X_Q = u_i\} \log\bracket{P\{X_Q = u_i\}} \\
        & = -F\paren{\frac{u_1 + u_2}{2}}\log\bracket{F\paren{\frac{u_1 + u_2}{2}}}\\
        &- \sum_{i = 2}^{K-1}\bracket{F\paren{\frac{u_{i+1} + u_i}{2}} - F\paren{\frac{u_{i} + u_{i-1}}{2}}}\log\bracket{F\paren{\frac{u_{i+1} + u_i}{2}} - F\paren{\frac{u_{i} + u_{i-1}}{2}}}\\
        &-\bracket{1 - F\paren{\frac{u_{K-1} + u_K}{2}}}\log\bracket{1 - F\paren{\frac{u_{K-1} + u_K}{2}}}.
        \end{aligned}
\end{equation}
    
\end{proof}

\begin{figure}[htbp]
\centering
\begin{tikzpicture}[>=Stealth, thick]

\draw[black, thick] (-3.5, 0) -- (5.5, 0);

\draw[black, thick] (-2, 2.8) -- (-2, 0);
  \draw[black, thick] ( 1, 2.8) -- ( 1, 0);
  \draw[black, thick] ( 4, 2.8) -- ( 4, 0);

\draw[blue, very thick] (-0.5, 2.8) -- (-0.5, 0);
  \draw[blue, very thick] ( 2.5, 2.8) -- ( 2.5, 0);

\draw[black, -{Stealth}, very thick] (-0.5, 1.4) -- ( 1.0, 1.4);
  \fill[black] (-0.5, 1.4) circle (3pt);
  \draw[black, -{Stealth}, very thick] ( 2.5, 1.4) -- ( 4.0, 1.4);
  \fill[black] ( 2.5, 1.4) circle (3pt);

\draw[black, -{Stealth}, very thick] (-0.5, 0.8) -- (-2.0, 0.8);
  \fill[white] (-0.5, 0.8) circle (3pt);
  \draw[black, very thick] (-0.5, 0.8) circle (3pt);
  \draw[black, -{Stealth}, very thick] ( 2.5, 0.8) -- ( 1.0, 0.8);
  \fill[white] ( 2.5, 0.8) circle (3pt);
  \draw[black, very thick] ( 2.5, 0.8) circle (3pt);

\draw[black, -{Stealth}, very thick] (-3.5, 1.1) -- (-2.0, 1.1);
  \draw[black, -{Stealth}, very thick] ( 5.5, 1.1) -- ( 4.0, 1.1);

\node at (-2, -0.4) {$u_1$};
  \node at ( 1, -0.4) {$u_2$};
  \node at ( 4, -0.4) {$u_3$};
  \node[blue] at (-0.5, -0.9) {$\dfrac{u_1+u_2}{2}$};
  \node[blue] at ( 2.5, -0.9) {$\dfrac{u_2+u_3}{2}$};

\end{tikzpicture}

\caption{Clipping and midpoint quantization with $K=3$ representable values $\{u_1, u_2,u_3\}$. The blue vertical lines represent the midpoints, and the arrows depict the regions of the real line that map to each representable value at a black vertical line.}\label{fig:midpoint_quantization}
\end{figure}
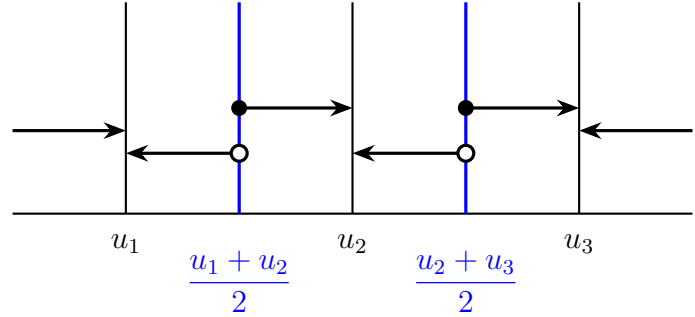

\begin{corollary}[Entropy of a Gaussian Random Variable Quantized to a Floating-point Number]\label{cor:gauss_fp_ent}
Let $X \sim \gauss{\mu}{\sigx^2}$. Let $X_{fp}$ be the clipped and midpoint quantized floating-point representation of $X$ with an $E$-bit exponent and a $(p-1)$-bit significand. The discrete entropy of $X_{fp}$ is given by Theorem~\ref{thm:quantized_ent} where $K = 2^{E + p}$, $F(x) = \frac{1}{2}\bracket{1 + \erf\paren{\frac{x -\mu}{\sqrt{2}\sigx}}}$,
\begin{equation*}
    u_i = \begin{cases}
    u^\prime_{i -2^{E + p - 1}}  \text{ if } i > 2^{E + p - 1}\\
    -u^\prime_{2^{E + p - 1} - (i-1)} \text{ otherwise},
\end{cases}
\end{equation*}
and 
\begin{equation*}
    u^\prime_i = 2^{\bracket{\left\lfloor (i-1)2^{-(p-1)}\right\rfloor - \paren{2^{E-1} - 1}}}\paren{1 + \bracket{(i-1)2^{-(p-1)} \bmod 1}}.
\end{equation*}
\end{corollary}

\begin{proof}
$F(x) = \frac{1}{2}\bracket{1 + \erf\paren{\frac{x -\mu}{\sqrt{2}\sigx}}}$ is the well-known Gaussian cumulative distribution function. All that is left to apply Theorem~\ref{thm:quantized_ent} is to explicitly construct the ordered sequence of representable values $\{u_1, u_2, \dots, u_K\}$ from the structure of the floating-point format given in Eq.~\eqref{eq:sem_fp}. There are $2^{p-1}$ distinct mantissa values for each exponent and $2^E$ exponent values, giving $2^{E+p-1}$ positive floating-point values in total. Including the negative values by symmetry, there are $K = 2^{E+p}$ representable values altogether.
The smallest value of the exponent $e_{\min} = -(2^{E-1}-1)$ while the largest value $e_{\max} = 2^{E-1}$.

We first enumerate the $2^{E+p-1}$ positive floating-point values $u^\prime_i$ in increasing order, indexed by $i = 1, 2, \ldots, 2^{E+p-1}$. For each fixed exponent value, we exhaust all $2^{p-1}$ mantissa values before incrementing the exponent.
Concretely, for index $i$, the exponent index is $\lfloor (i-1)/2^{p-1} \rfloor$, which starts at zero steps up by one exactly every $2^{p-1}$ values of $i$. The mantissa index within that exponent block is $(i-1) \bmod 2^{p-1}$, which cycles through $0, 1, \ldots, 2^{p-1}-1$ repeatedly. Translating the exponent index into the true exponent by subtracting $e_{\min}$, and the mantissa index into its fractional value by multiplying by $2^{-(p-1)}$, gives
\begin{equation*}
    u^\prime_i = 2^{\left\lfloor (i-1)2^{-(p-1)}\right\rfloor - (2^{E-1}-1)}\Bigl(1 + \bracket{(i-1)2^{-(p-1)} \bmod 1}\Bigr).
\end{equation*}

The full sequence $u_1 < u_2 < \cdots < u_K$ must enumerate the negative and positive floating-point values in increasing order. Since the negative floating-point values are the mirror image of the positive ones, the most negative value corresponds to $-u^\prime_{2^{E+p-1}}$ and the least negative to $-u^\prime_1$. Therefore, for $i \leq 2^{E+p-1}$, we set
\begin{equation*}
    u_i = -u^\prime_{2^{E+p-1} - (i-1)},
\end{equation*}
which enumerates the negative values in increasing order. For $i > 2^{E+p-1}$, we set
\begin{equation*}
    u_i = u^\prime_{i - 2^{E+p-1}},
\end{equation*}
which enumerates the positive values in increasing order. With these $u_i$ and $K = 2^{E+p}$, Theorem~\ref{thm:quantized_ent} can be used to compute $H(\Xfp)$, completing the proof.
\end{proof}

\begin{corollary}[Joint Entropy of Two Jointly Gaussian Random Variables Quantized to Floating-point Numbers]\label{cor:joint_ent}
    Let $w \in \R{}$, $X \sim \gauss{0}{\sigx^2}$ and $\xi \sim \gauss{0}{\sigxi^2}$ where $X$ and $\xi$ are independent. Define $Y = wX + \xi$. Let $X_{fp}$ and $Y_{fp}$ be the floating-point representation of $X$ and $Y$ respectively, where both representations have $E$ exponent bits and $(p-1)$ significand bits. Let $K = 2^{E+p}$ and denote the set of representable values as $\{u_1, u_2, \dots, u_K\}$ as given in Corollary~\ref{cor:gauss_fp_ent}. When $K \geq 2$, the joint discrete entropy of $X_{fp}$ and $Y_{fp}$ is
    \begin{equation*}
        H(X_{fp}, Y_{fp}) = H(\Xfp) - \sum_{i = 1}^{K}\sum_{j = 1}^{K} F_{XY}\paren{a_j, b_j; a_i, b_i}\log\bracket{\frac{F_{XY}\paren{a_j, b_j; a_i, b_i}}{F_X(b_i) - F_X(a_i)}}.
    \end{equation*}
where $H(\Xfp)$ is given from Corollary~\ref{cor:gauss_fp_ent}, $F_{XY}(a, b; c,d) \triangleq \int_a^b\int_c^d f_{XY}(x,y)dxdy$ where $f_{XY}$ is the bivariate Gaussian density function corresponding to Eq.~\eqref{eq:initial_dist}, $F_X(x)$ is the univariate Gaussian cumulative density function, and \begin{equation*}
    a_i = \begin{cases} -\infty & i = 1 \\ \frac{u_i + u_{i-1}}{2} & 2 \leq i \leq K\end{cases}, \qquad b_i =\begin{cases} \frac{u_{i+1}+u_i}{2} & 1 \leq i \leq K-1 \\ \infty & i = K\end{cases}. 
    \end{equation*} When $K = 1$, $H(\Xfp, \Yfp) = 0$.
\end{corollary}
\begin{proof}
\begin{equation}
     H(X_{fp}, Y_{fp}) = H(\Xfp) + H(\Yfp|\Xfp).
\end{equation}
$H(\Xfp)$ is given by Theorem~\ref{thm:quantized_ent}, so all that is left is to compute $H(\Yfp|\Xfp)$.
Let $\{u_1, \dots, u_K\}$ denote the $2^{E + p}$ representable values of the floating-point representation as specified in Corollary~\ref{cor:gauss_fp_ent}. We know that\begin{equation}
    \begin{bmatrix}
        X\\
        Y
    \end{bmatrix} \sim \gauss{0}{\begin{bmatrix}
        \sigx^2 & w\sigx^2\\
        w\sigx^2 & w^2\sigx^2 + \sigxi^2
    \end{bmatrix}}.
\end{equation}

Let $F_{XY}(a, b; c,d) \triangleq \int_a^b\int_c^d f_{XY}(x,y)dxdy$ where $f_{XY}$ is the bivariate Gaussian density function corresponding to Eq.~\eqref{eq:initial_dist}. Notice that,
\begin{equation}
    P\{Y \in [a,b)| X\in [c,d)\} = \frac{P\{Y \in [a,b), X\in [c,d)\}}{P\{X\in [c,d)\}} = \frac{F_{XY}(a, b; c,d)}{F_X(d) - F_X(c)}.\label{eq:cond_prob}\\
\end{equation}

When $K=1$, $P\{\Xfp = u_1 \} = P\{\Yfp = u_1 \} = 1$ so $H(\Yfp|\Xfp) = 0$. 
When $K \geq 2$, let $a_i$ and $b_i$ denote the left and right boundaries of the quantization bin for $u_i$:
\begin{equation}
    a_i \triangleq \begin{cases} -\infty & i = 1 \\ \frac{u_i + u_{i-1}}{2} & 2 \leq i \leq K\end{cases}, \qquad b_i \triangleq \begin{cases} \frac{u_{i+1}+u_i}{2} & 1 \leq i \leq K-1 \\ \infty & i = K\end{cases}.
\end{equation}
Then,
\begin{equation}
\begin{aligned}
    &H(\Xfp) + H(\Yfp|\Xfp) = H(\Xfp) - \sum_{i = 1}^{K}\sum_{j = 1}^K P\{\Yfp = u_j, \Xfp = u_i\}\log\bracket{P\{\Yfp = u_j | \Xfp = u_i\}}\\
    &= H(\Xfp) - \sum_{i = 1}^{K}\sum_{j = 1}^{K} P\left\{Y \in \left[a_j, b_j\right), X \in \left[a_i, b_i\right)\right\}\log\bracket{\frac{P\left\{Y \in \left[a_j, b_j\right), X \in \left[a_i, b_i\right)\right\}}{P\left\{ X \in \left[a_i, b_i\right)\right\}}}\\
    &= H(\Xfp) - \sum_{i = 1}^{K}\sum_{j = 1}^{K} F_{XY}\paren{a_j, b_j; a_i, b_i}\log\bracket{\frac{F_{XY}\paren{a_j, b_j; a_i, b_i}}{F_X(b_i) - F_X(a_i)}}.
\end{aligned}
\end{equation}

\end{proof}

\FloatBarrier
\pagebreak
\section{Approximating the entropy of floating point numbers}\label{app:FPN}

Here we review the details of calculations used to approximate the entropy of floating point numbers.

\subsubsection*{Approximation 1 -- Relating discrete and differential entropy for non-uniform bins}\label{app:FPN_app1}

Let there be a quantization scheme with $K$ representable values
$\{u_1, u_2, \ldots, u_K\}$, where $u_1 < u_2 < \cdots < u_K$.
Let $\mathbf{X} \sim f_\mathbf{X}$, where $f_\mathbf{X}$ is a differentiable probability density function with support on
\begin{equation}
\mathbb{U}^d \triangleq \left[-2^{e_{max} + 1} + 2^{e_{max} - p}, \,\, 2^{e_{max} + 1} - 2^{e_{max} - p}\right]^d,
\end{equation}
where we assume each component $X_j$ is independently clipped and quantized.

Let $\Delta: \left[-2^{e_{max} + 1} + 2^{e_{max} - p},\,\, 2^{e_{max} + 1} - 2^{e_{max} - p}\right] \to \mathbb{R}_+$
encode the bin sizes of the quantization scheme for a single component, where
\begin{equation*}
    \begin{aligned}
    \Delta(x_j) &\triangleq \sum_{i=2}^{K-1}\frac{u_{i+1} - u_{i-1}}{2}
    \indic{x_j \in \Big[\frac{u_i+u_{i-1}}{2}, \frac{u_{i+1}+u_i}{2}\Big)} + (u_2 - u_1)\indic{x_j \in [-2^{e_{max} + 1} + 2^{e_{max} - p}, \frac{u_{2} + u_{1}}{2}]}\nonumber\\
    & + (u_K - u_{K-1})\indic{x_j \in [\frac{u_{K} + u_{K-1}}{2}, 2^{e_{max} + 1} - 2^{e_{max} - p}]}.
    \end{aligned}
\end{equation*}
Since each component is stored on its own floating-point register, the quantization cell for $\mathbf{X}_Q$ is a $d$-dimensional rectangle with volume
\begin{equation}
\Delta(\mathbf{x}) \triangleq \prod_{j=1}^{d} \Delta(x_j),
\qquad
\log \Delta(\mathbf{x}) = \sum_{j=1}^{d} \log \Delta(x_j).
\end{equation}

Let $\mathbf{i} = (i_1,\dots,i_d)$ index the $d$-dimensional quantization cells, where $i_j \in \{1,\dots,K\}$. Define the cell boundaries
\begin{equation}
a_i=\begin{cases}
-2^{e_{max} + 1} + 2^{e_{max} - p} & i=1\\[2pt]
\frac{u_i+u_{i-1}}2 & 2\leq i\leq K
\end{cases},
\qquad
b_i=\begin{cases}
\frac{u_{i+1}+u_i}2 & 1\leq i\leq K-1\\[2pt]
2^{e_{max} + 1} - 2^{e_{max} - p} & i=K.
\end{cases}
\end{equation}
For each $d$-dimensional index $\mathbf{i}$, define the quantization cell
\begin{equation}
B_{\mathbf{i}} \triangleq \prod_{j=1}^d [a_{i_j}, b_{i_j}],
\end{equation}
and let
\begin{equation}
p_{\mathbf{i}} \triangleq \mathbb{P}\{\mathbf{X}\in B_{\mathbf{i}}\}
= \mathbb{P}\{\mathbf{X}_Q = \mathbf{u}_{\mathbf{i}}\},
\end{equation}
where $\mathbf{u}_{\mathbf{i}}$ is the $d$-dimensional representable value corresponding to the quantization cell $B_\mathbf{i}$.

Then the discrete entropy of $\mathbf{X}_Q$ is
\begin{equation}
H(\mathbf{X}_Q) = -\sum_{\mathbf{i}} p_{\mathbf{i}} \log p_{\mathbf{i}}.
\end{equation}

To relate this discrete entropy to the differential entropy of $\mathbf{X}$, define the piecewise-uniform density
\begin{equation}
g(\mathbf{x})
\triangleq
\sum_{\mathbf{i}}
\frac{p_{\mathbf{i}}}{|B_{\mathbf{i}}|}
\indic{\mathbf{x}\in B_{\mathbf{i}}},
\label{eq:piecewise_uniform_g}
\end{equation}
where $|B_{\mathbf{i}}| = \Delta(\mathbf{x})$ for $\mathbf{x} \in B_{\mathbf{i}}$ is the volume of each rectangular cell $B_{\mathbf{i}}$. By construction, $g$ is constant on each cell and integrates to one across all cells.

Following \cite{kostina_data_2017}, we can derive an exact identity relating the discrete and differential entropies. By the definition of Kullback-Leibler divergence,
\begin{align}
D(f_{\mathbf{X}}\|g)
&=
\int_{\mathbb{U}^d}
f_{\mathbf{X}}(\mathbf{x})
\log\frac{f_{\mathbf{X}}(\mathbf{x})}{g(\mathbf{x})}
\,d\mathbf{x}
\notag\\
&=
\int_{\mathbb{U}^d}
f_{\mathbf{X}}(\mathbf{x})\log f_{\mathbf{X}}(\mathbf{x})
\,d\mathbf{x}
-
\int_{\mathbb{U}^d}
f_{\mathbf{X}}(\mathbf{x})\log g(\mathbf{x})
\,d\mathbf{x}
\notag\\
&=
-h(\mathbf{X})
-
\int_{\mathbb{U}^d}
f_{\mathbf{X}}(\mathbf{x})\log g(\mathbf{x})
\,d\mathbf{x}.
\label{eq:KL_start}
\end{align}
Since $g(\mathbf{x}) = p_{\mathbf{i}}/|B_{\mathbf{i}}|$ for $\mathbf{x}\in B_{\mathbf{i}}$,
\begin{align}
\int_{\mathbb{U}^d}
f_{\mathbf{X}}(\mathbf{x})\log g(\mathbf{x})
\,d\mathbf{x}
&=
\sum_{\mathbf{i}}
\int_{B_{\mathbf{i}}}
f_{\mathbf{X}}(\mathbf{x})
\log\left(\frac{p_{\mathbf{i}}}{|B_{\mathbf{i}}|}\right)
d\mathbf{x}\\
&=
\sum_{\mathbf{i}}
 \bigg(
\int_{B_{\mathbf{i}}}
f_{\mathbf{X}}(\mathbf{x}) d\mathbf{x}\bigg)
\log\left(\frac{p_{\mathbf{i}}}{|B_{\mathbf{i}}|}\right)
=
\sum_{\mathbf{i}}
p_{\mathbf{i}}
\log\left(\frac{p_{\mathbf{i}}}{|B_{\mathbf{i}}|}\right)\\
&=
\sum_{\mathbf{i}} p_{\mathbf{i}}\log p_{\mathbf{i}}
-
\sum_{\mathbf{i}} p_{\mathbf{i}}\log |B_{\mathbf{i}}|.
\label{eq:KL_middle}
\end{align}
Substituting Eq.~\eqref{eq:KL_middle} into Eq.~\eqref{eq:KL_start} yields
\begin{align}
D(f_{\mathbf{X}}\|g)
&=
-h(\mathbf{X})
-
\sum_{\mathbf{i}} p_{\mathbf{i}}\log p_{\mathbf{i}}
+
\sum_{\mathbf{i}} p_{\mathbf{i}}\log |B_{\mathbf{i}}|
\notag\\
&=
-h(\mathbf{X})
+
H(\mathbf{X}_Q)
+
\mathbb{E}[\log \Delta(\mathbf{X})],
\end{align}
or equivalently,
\begin{equation}
H(\mathbf{X}_Q)
=
h(\mathbf{X})
-
\mathbb{E}[\log \Delta(\mathbf{X})]
+
D(f_{\mathbf{X}}\|g).
\label{eq:exact_identity}
\end{equation}
Approximation 1 corresponds to neglecting the nonnegative correction term $D(f_{\mathbf{X}}\|g)$:
\begin{equation}
H(\mathbf{X}_Q)
\approx
h(\mathbf{X})
-
\mathbb{E}[\log \Delta(\mathbf{X})].
\label{eq:approx1_clean}
\end{equation}

\cite{kostina_data_2017} shows rigorously that the correction term $D(f_{\mathbf{X}}\|g)$ becomes large if $f_{\mathbf{X}}$ varies noticeably within each quantization cell, and vanishes if the function is uniform across each cell. Here, we can understand this at a basic level by noting that $g(\mathbf{x})$ is equal to the mean value of $f_{\mathbf{X}}$ within a given bin $B_{\mathbf{i}}$. By the multivariate mean value theorem, $g(\mathbf{x}) = f_\mathbf{X}(\mathbf{\tilde x_{\mathbf{i}}})$ for some $\mathbf{\tilde x}_{\mathbf{i}} \in B_{\mathbf{i}}$. Taylor expanding $f_{\mathbf{X}}(\mathbf{x})$ within cell $B_{\mathbf{i}}$ to first order around $\mathbf{\tilde x}_{\mathbf{i}}$ we see:
\begin{equation}
f_{\mathbf{X}}(\mathbf{x}) \approx g (\mathbf{x})+ \nabla f_{\mathbf{X}}(\mathbf{\tilde x}_{\mathbf{i}}) \cdot (\mathbf{x}- \mathbf{\tilde x}_{\mathbf{i}}).
\end{equation}
Since $D(f_{\mathbf{X}}\|g) \approx 0$ when $f_{\mathbf{X}} \approx g$, a cell can contribute significantly to the error if the size of the cell and the gradient of $f_{\mathbf{X}}$ are simultaneously large, since $(\mathbf{x}-\mathbf{\tilde x}_{\mathbf{i}})$ is bounded by the size of the cell. However, when the cell diameter is small in regions where the density varies rapidly, or if the cell diameter is large where the density varies only slowly, the discrepancy between $f_{\mathbf{X}}$ and its cell average $g(\mathbf{x})$ remains small.

Finally, we can note that if each component is quantized independently, then
\begin{equation}
\mathbb{E}[\log \Delta(\mathbf{X})]
=
\sum_{j=1}^{d}\mathbb{E}_{X_j}[\log \Delta(X_j)],
\label{eq:separable}
\end{equation}
so that Approximation 1 takes the form
\begin{equation}
\tilde{H}(\mathbf{X}_Q)
\triangleq
h(\mathbf{X})
-
\sum_{j=1}^{d}\mathbb{E}_{X_j}[\log \Delta(X_j)].
\label{eq:multi_var_diff_disc}
\end{equation}

\subsubsection*{Approximation 2 -- Smoothing the bin size function and extending its domain}\label{app:FPN_app2}

In the multivariate case, with $\Delta_s(x) = \frac{|x|}{\sqrt{2}}\cdot 2^{1-p}$ from Eq.~\eqref{eq:step_approx}, we can define
\begin{equation}
    \Delta_s (\mathbf{x}) = \prod^d_{j=1} \Delta_s(x_j).
\end{equation}
Again, $\mathbb{U}^d = [-2^{e_{max} + 1} + 2^{e_{max} - p}, \,\, 2^{e_{max} + 1} - 2^{e_{max} - p}]^d$, where $e_{max} = 2^{E-1}$. Assuming that the tails of $f_\mathbf{X}$ are vanishingly small,
\begin{equation}
    \int_{\mathbb{R}^d \setminus \mathbb{U}^d} f_\mathbf{X}(\mathbf{x}) \, d\mathbf{x} \approx 0,
\end{equation}
we extend the domain from $\mathbb{U}^d$ to $\mathbb{R}^d$. As noted in the main text, this extension requires that the omitted tail integral $\int_{\mathbb{R}^d\setminus\mathbb{U}^d} f_\mathbf{X}(\mathbf{x})|\log[f_\mathbf{X}(\mathbf{x})\Delta_s(\mathbf{x})]|\,d\mathbf{x}$ is near zero, not merely that the tail mass is small. For the Gaussian and Student's $t$ distributions used in this paper, the tail decay dominates the linear growth of $\Delta_s$, so this tail integral is negligible when the granular region $\mathbb{U}$ is sufficiently large (i.e., $E$ is not too small). With this we can say from Eq.~\eqref{eq:multi_var_diff_disc} $\tilde H(\mathbf{X}_Q) \approx \tilde H_s(\mathbf{X}_Q)$, where
\begin{equation}
    \tilde H_s(\mathbf{X}_Q) = -\int_{\mathbb{R}^d}  f_\mathbf{X}(\mathbf{x})
    \log\left[f_\mathbf{X}(\mathbf{x})
    \prod_{j=1}^d \Delta_s(x_j)\right] d\mathbf{x}.
\end{equation}

Define the smooth-bin approximation restricted to $\mathbb U^d$ as
\begin{equation}
\label{eq:Ht_s_U}
\tilde H_{s,\mathbb U}(\mathbf X_{fp})
\;\triangleq\;
-\int_{\mathbb U^d} f_{\mathbf X}(\mathbf x)\,\log\left[f_{\mathbf X}(\mathbf x)\prod_{j=1}^d \Delta_s(x_j)\right]d\mathbf x.
\end{equation}
Theorem~\ref{thm:steps} bounds the smoothing error by the $d/2$ term plus the explicit contribution from the bins adjacent to zero.

Define the special-bin set
\begin{equation}
\label{eq:special_bins}
\begin{aligned}
\mathcal B_0 &\triangleq \left\{x\in \mathbb U:\ |x|<2^{e_{\min}}(1+2^{-p})\right\}
\end{aligned}
\end{equation}
where $\mathcal B_0$ is the union of the two bins adjacent to $0$. When $x \in \mathcal B_{0}$, $\Delta(x) = \Delta_0 = 2^{e_{\min}}(1+2^{-p})$. Define
\begin{equation}
\label{eq:special_error}
\varepsilon_{0}(\mathbf X)
\triangleq
\sum_{j=1}^{d}
\int_{\mathcal B_{0}}
f_{X_j}(x)
\left|
\log\bracket{\frac{\Delta_s(x)}{\Delta_0}}
\right|dx.
\end{equation}

\begin{theorem}[Error bound for smoothing the bin-size function]\label{thm:steps}
Let $\mathbf{X}\sim f$ be a $d$-dimensional random vector with the probability density $f$ supported on $\mathbb U^d$, and let $X_j$ denote its $j$th component. Let $\tilde H(\mathbf X_{fp})$ be the approximation from App.~\ref{app:FPN_app1} using the true midpoint bin-size function $\Delta$, and let $\tilde H_{s,\mathbb U}(\mathbf X_{fp})$ be the corresponding approximation on $\mathbb U^d$ using $\Delta_s$. Then
\begin{equation}
\abs{\tilde H(\mathbf X_{fp}) - \tilde H_{s,\mathbb U}(\mathbf X_{fp})} \leq \frac{d}{2} + \varepsilon_{0}(\mathbf X),
\end{equation}
where $\varepsilon_{0}(\mathbf X)$ is finite and defined in Eq.~\eqref{eq:special_error}.
\end{theorem}

\begin{proof}

There are 3 types of bins for which we will bound the error: bins on the interior of an exponent block, bins on the outer edges, and bins on the boundary between exponent blocks (regions for which a single exponent value applies). The bound can be shown by considering the ratio between bin sizes, case by case.

\paragraph{Interior Bins}
For interior bins strictly within each exponent block $e$, by Eq.~\eqref{eq:true_bins},
\begin{equation}
    \label{eq:real_bins}
    \Delta(x_j) = 2^{e-(p-1)} \quad \text{for } 2^e \leq |x_j| < 2^{e+1},
\end{equation}
and within this range $\Delta_s(x_j) = \frac{|x_j|}{\sqrt{2}} \cdot 2^{1-p}$ satisfies
\begin{equation}
    \begin{aligned}
    &\frac{2^e}{\sqrt{2}}2^{1-p} \leq \Delta_s(x_j) \leq \frac{2^{e+1}}{\sqrt{2}}2^{1-p}\\
    \implies &\frac{1}{\sqrt{2}} \leq \frac{\Delta_s(x_j)}{\Delta(x_j)} < \sqrt{2}.
    \end{aligned}
\end{equation}

\paragraph{Outer clipping bins}

For the outer clipping bins, by symmetry it suffices to analyze $[\frac{u_K + u_{K-1}}{2}, 2^{e_{\max}+1}-2^{e_{\max}-p}]$ which is the bin on positive side of the real line. When $p \geq 2$, we know that $u_K = 2^{e_{max}+1}(1 - 2^{-p})$, and $u_{K-1} = 2^{e_{max}+1}(1 - 2^{1-p})$, so the outer clipping bin size is \begin{equation}
\Delta(x_j) = 2^{e_{max} + 1} - 2^{e_{max} - p}- \frac{u_K + u_{K-1}}{2} = 2^{e_{max}}\paren{2^{1-p}} = 2^{e_{max}-p + 1}
\end{equation}
Now, we have
\begin{equation}
\begin{aligned}
  &\frac{(u_K + u_{K-1})}{2\sqrt{2}}2^{1-p} \leq \Delta_s(x_j) \leq \frac{2^{e_{max} + 1} - 2^{e_{max} - p}}{\sqrt{2}}2^{1-p}\\
  \implies &\frac{ 2^{e_{max}}(1 - 2^{-1-p} - 2^{-p})}{2^{e_{max} - p}\sqrt{2}}2^{1-p} \leq \frac{\Delta_s(x_j)}{\Delta(x_j)} \leq \frac{2^{e_{max}}(1 - 2^{-1 - p})}{2^{e_{max} - p}\sqrt{2}}2^{1-p}\\
  \implies &\frac{ 2(1 - 2^{-1-p} - 2^{-p})}{\sqrt{2}} \leq \frac{\Delta_s(x_j)}{\Delta(x_j)} \leq \frac{2(1- 2^{-1-p})}{\sqrt{2}}\\
  \implies &\sqrt{2}(1 - 2^{-1-p} - 2^{-p}) \leq \frac{\Delta_s(x_j)}{\Delta(x_j)} \leq \sqrt{2}(1- 2^{-1-p}).
\end{aligned}
\end{equation}
When $p = 1$, $u_K = 2^{e_{max}}$ and $u_{K-1} = 2^{e_{max}-1}$. The true bin width is 
\begin{equation}
\Delta(x_j) = 2^{e_{max}+1} - 2^{e_{max}-1} - \frac{1}{2}(2^{e_{max}} + 2^{e_{max}-1}) = 3\cdot2^{e_{max}-2}.
\end{equation}
Using the same bounding technique, we have
\begin{equation}
\begin{aligned}
  &\frac{2^{e_{max}} + 2^{e_{max}-1}}{2\sqrt{2}} \leq \Delta_s(x_j) \leq \frac{2^{e_{max}+1} - 2^{e_{max}-1}}{\sqrt{2}}\\
  \implies &\frac{3\cdot 2^{e_{max}-2}}{3\cdot2^{e_{max}-2}\sqrt{2}} \leq \frac{\Delta_s(x_j)}{\Delta(x_j)} \leq \frac{3\cdot 2^{e_{max}-1}}{3\cdot2^{e_{max}-2}\sqrt{2}}\\
  \implies &\frac{ 1}{\sqrt{2}} \leq \frac{\Delta_s(x_j)}{\Delta(x_j)} \leq \frac{2}{\sqrt{2}}.
\end{aligned}
\end{equation}
For the outer clipping bins, we see that for all $p\geq 1$, $\frac{1}{\sqrt{2}} \leq \frac{\Delta_s(x_j)}{\Delta(x_j)} \leq \sqrt{2}$.

\paragraph{Exponent boundary bins}

At exponent boundaries, let the last representable value of the $e$-th exponent block be $u_A = 2^{e+1} - 2^{e-(p-1)}$. The first representable value of the $e+1$-th exponent block be $u_B = 2^{e+1}$, and the second representable value of the $e+1$-th exponent block be $u_C = 2^{e+1} + 2^{e+1 - (p-1)}$. The left boundary of the midpoint-quantization bin for $u_B$ is $m_L = \frac{u_A + u_B}{2} = 2^{e+1}(1 - 2^{-p-1})$ while the right boundary is $m_R = \frac{u_B + u_C}{2} = 2^{e+1}(1+2^{-p})$. When $x_j \in [m_L, m_R)$, the bin width is
\begin{equation}\label{eq:exp_boundary}
\Delta(x_j) = m_R - m_L = 3 \cdot 2^{e-p},
\end{equation}
and the ratio $\Delta_s(x_j)/\Delta(x_j)$ is bounded by
\begin{equation}
    \begin{aligned}
    &\frac{2^{e+1}(1 - 2^{-p-1})}{\sqrt{2}}2^{1-p} \leq \Delta_s(x_j) \leq \frac{2^{e+1}(1+2^{-p})}{\sqrt{2}}2^{1-p}\\
    \implies & \frac{2^{e+1}(1 - 2^{-p-1})}{3 \cdot 2^{e-p}\sqrt{2}}2^{1-p} \leq \frac{\Delta_s(x_j)}{\Delta(x_j)} \leq \frac{2^{e+1}(1+2^{-p})}{3 \cdot 2^{e-p}\sqrt{2}}2^{1-p}\\
    \implies & \frac{2\sqrt{2}(1 - 2^{-p-1})}{3} \leq \frac{\Delta_s(x_j)}{\Delta(x_j)}  \leq \frac{2\sqrt{2}(1+2^{-p})}{3}.
    \end{aligned}
\end{equation}
When $p = 1$ the lower bound is equal to $1/\sqrt{2}$ and the upper bound is equal to $\sqrt{2}$. The lower bound is monotonically increasing in $p$ while the upper bound is monotonically decreasing in $p$, and both converge to $2\sqrt{2}/3$ as $p \rightarrow \infty$, which is between $1/\sqrt{2}$ and $\sqrt{2}$. This means both bounds lie within $[1/\sqrt{2},\, \sqrt{2}]$, so for within exponent block bins and exponent boundary bins,
\begin{equation}
    -\frac{1}{2} \leq \log\left[\frac{\Delta_s(x_j)}{\Delta(x_j)}\right] < \frac{1}{2}.
\end{equation}

\paragraph{Bound on entropy difference from bin ratios}

Hence, for every $x_j\in \mathbb U\setminus \mathcal B_{0}$,
\begin{equation}
\abs{\log\bracket{\frac{\Delta_s(x_j)}{\Delta(x_j)}}} \leq \frac{1}{2}.
\end{equation}

The remaining points are exactly the special bins collected in $\mathcal B_{0}$, where we do not claim a uniform pointwise bound and instead keep their contribution explicitly:
\begin{equation}
\begin{aligned}
&\left|
\tilde H(\mathbf X_{fp}) - \tilde H_{s,\mathbb U}(\mathbf X_{fp}) \right| = \Bigg | -\int_{\mathbb U^d} f_{\mathbf X}(\mathbf x)\,\log\left[f_{\mathbf X}(\mathbf x)\prod_{j=1}^d \Delta(x_j)\right]d\mathbf x  \\ & \qquad \qquad \qquad \qquad \qquad
+ \int_{\mathbb U^d} f_{\mathbf X}(\mathbf x)\,\log\left[f_{\mathbf X}(\mathbf x)\prod_{j=1}^d \Delta_s(x_j)\right]d\mathbf x \Bigg |\\
&\left|
\tilde H(\mathbf X_{fp}) - \tilde H_{s,\mathbb U}(\mathbf X_{fp}) \right| = \abs{\sum_{j=1}^{d} \Esub{\log\bracket{\frac{\Delta_s(x)}{\Delta(x)}}}{X_j}} \\
&\leq\sum_{j=1}^{d} \int_{\mathbb U\setminus \mathcal B_{0}} f_{X_j}(x) \abs{\log\bracket{\frac{\Delta_s(x)}{\Delta(x)}}}dx + \varepsilon_{0}(\mathbf X) \leq\frac d2 + \varepsilon_{0}(\mathbf X).
\end{aligned}
\end{equation}
The point $x=0$ has measure zero, and $\varepsilon_{0}(\mathbf X)$ is finite since each marginal density is continuous on $\mathbb U$ while $|\log|x||$ is locally integrable near $0$.
\end{proof}

Not only is $\varepsilon_{0}(\mathbf X)$ finite, as shown in Theorem~\ref{thm:steps}, but for the probability distributions that this paper studies, Corollary~\ref{cor:eps_sp_bound} shows that $\varepsilon_{0}(\mathbf X)$ is negligibly small.

\begin{corollary}[Bound on $\varepsilon_{0}$ for densities bounded near zero]
\label{cor:eps_sp_bound}
Under the hypotheses of Theorem~\ref{thm:steps}, with
$\mathcal{B}_{0} = \mathcal{B}_0$ (the two zero-crossing bins),
define
\begin{equation}
  a \;\triangleq\; 2^{e_{\min}}\bigl(1 + 2^{-p}\bigr),
  \qquad
  C_0 \;\triangleq\; 2\,a\left(p - \frac{1}{2} + \log e\right).
\end{equation}
Then
\begin{equation}
\label{eq:eps_sp_general}
  \varepsilon_{0}(\mathbf{X})
  \;\le\;
  C_0
  \sum_{j=1}^{d}
  \sup_{x \in \mathcal{B}_0} f_{X_j}(x).
\end{equation}
 
If every marginal density $f_{X_j}$ attains its mode at $0$,
then
\begin{equation}
\label{eq:eps_sp_mode}
  \varepsilon_{0}(\mathbf{X})
  \;\le\;
  C_0
  \sum_{j=1}^{d} f_{X_j}(0).
\end{equation}
 
\medskip
\noindent\textbf{(i)\; Zero-mean Gaussian marginals.}\;
If $X_j \sim \mathcal{N}(0,\sigma_j^2)$, then $f_{X_j}(0) = 1/(\sigma_j\sqrt{2\pi})$, and
\begin{equation}
\label{eq:eps_sp_gauss}
  \varepsilon_{0}(\mathbf{X})
  \;\le\;
  \frac{C_0}{\sqrt{2\pi}}
  \sum_{j=1}^{d} \frac{1}{\sigma_j}.
\end{equation}
 
\noindent\textbf{(ii)\; Scaled Student's $t$ marginals.}\;
If $X_j \sim t_\nu(0,\sigma_j)$ with density
$f_{X_j}(x)
= \frac{\Gamma\bigl(\frac{\nu+1}{2}\bigr)}
       {\sigma_j\sqrt{\nu\pi}\;\Gamma\bigl(\frac{\nu}{2}\bigr)}
  \bigl(1 + x^2/(\nu\sigma_j^2)\bigr)^{-(\nu+1)/2}$,
then $f_{X_j}(0) = \Gamma\bigl(\frac{\nu+1}{2}\bigr)
  \big/\bigl(\sigma_j\sqrt{\nu\pi}\;\Gamma\bigl(\frac{\nu}{2}\bigr)\bigr)$, and
\begin{equation}
\label{eq:eps_sp_student}
  \varepsilon_{0}(\mathbf{X})
  \;\leq\;
  \frac{C_0\;\Gamma\bigl(\frac{\nu+1}{2}\bigr)}
       {\sqrt{\nu\pi}\;\Gamma\bigl(\frac{\nu}{2}\bigr)}
  \sum_{j=1}^{d} \frac{1}{\sigma_j}.
\end{equation}
\end{corollary}
 
\begin{proof}
We see that $a$ is the value of the first interior midpoint for a positive bin adjacent to zero $(0, a]$. The true midpoint bin width is,
$\Delta(x_j) = a$,
while $\Delta_s(x_j) = x_j \cdot 2^{1-p}/\sqrt{2}$.
Since on this interval $x_j 2^{1-p}/\sqrt{2} \leq x_j \leq a$, the ratio $\Delta_s(x_j)/\Delta(x_j) < 1$
on the entire bin, so
\begin{equation}
  \abs{\log\bracket{\Delta_s(x)/\Delta(x)}}
  = \log\bracket{\frac{a \cdot 2^{p-1/2}}{x}}.
\end{equation}
By symmetry the negative bin $[-a, 0)$ gives the same integrand in $|x|$.
Therefore, for each marginal,
\begin{align}
  \int_{\mathcal{B}_0} f_{X_j}(x)
    \bigl|\log\bracket{\Delta_s(x)/\Delta(x)}\bigr|\,dx
  &\leq
  \sup_{t \in \mathcal{B}_0} f_{X_j}(t)
  \;\cdot\;
  2\int_0^{a} \log\bracket{\frac{a\cdot 2^{p-1/2}}{x}}dx.
\end{align}
The integral evaluates as
\begin{align}
  2\int_0^{a} \log\bracket{\frac{a\cdot 2^{p-1/2}}{x}}dx
  &= 2\int_0^{a}\Bigl[\bigl(p - \frac{1}{2}\bigr)
     + \log[a/x]\Bigr]dx \notag\\
  &= 2a\left(p - \frac{1}{2}\right)
     + \frac{2}{\ln 2}\int_0^{a}\ln[a/x]\,dx \notag\\
  &= 2a\left(p - \frac{1}{2}\right) + \frac{2a}{\ln 2}
  \;=\; 2a\left(p - \frac{1}{2} + \log e\right)
  \;=\; C_0,
\end{align}
where we used $\int_0^{a}\ln[a/x]\,dx = a$.
Summing over $j=1,\dots,d$ gives Eq.~\eqref{eq:eps_sp_general}.

When each $f_{X_j}$ has its mode at zero,
$\sup_{x\in\mathcal{B}_0} f_{X_j}(x) = f_{X_j}(0)$,
yielding Eq.~\eqref{eq:eps_sp_mode}.
Substituting $f_{X_j}(0)$ for $\mathcal{N}(0,\sigma_j^2)$ and
$t_\nu(0,\sigma_j)$ gives Eqs.~\eqref{eq:eps_sp_gauss}
and~\eqref{eq:eps_sp_student}.
\end{proof}
 
For candidate floating-point formats and identical marginals
($\sigma_j = \sigma$ for all $j$), Figure~\ref{fig:eps0_C0_vs_p}
shows $C_0$ and the per-dimension bound $\varepsilon_{0}/d$
as functions of the precision $p$ for several exponent widths $E$,
with $\sigma = 1$ and $d = 1$.
Gaussian marginals are shown as circles and
Student's $t_5$ marginals as squares; the two nearly overlap because
the ratio
$\Gamma\bigl(\frac{\nu+1}{2}\bigr)\sqrt{2/\nu}\big/\Gamma\bigl(\frac{\nu}{2}\bigr)$
is $O(1)$ for all $\nu \ge 1$.
By $E = 4$, $\varepsilon_0 \approx 0.16$ for $p = 24$. 

\begin{figure}[htbp]
  \centering
  \begin{subfigure}[t]{0.48\textwidth}
    \centering
    \includegraphics[width=\textwidth]{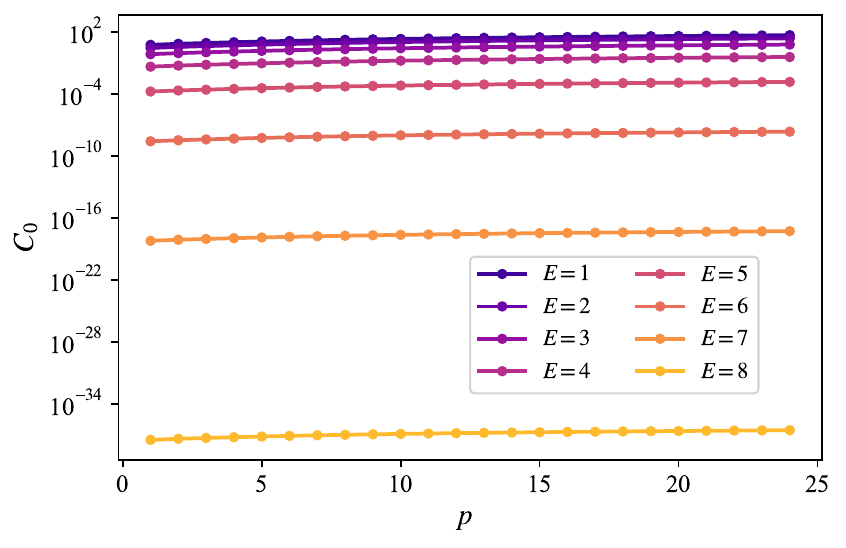}
    \caption{$C_0$ vs.\ precision $p$.}
    \label{fig:C0_vs_p}
  \end{subfigure}
  \hfill
  \begin{subfigure}[t]{0.48\textwidth}
    \centering
    \includegraphics[width=\textwidth]{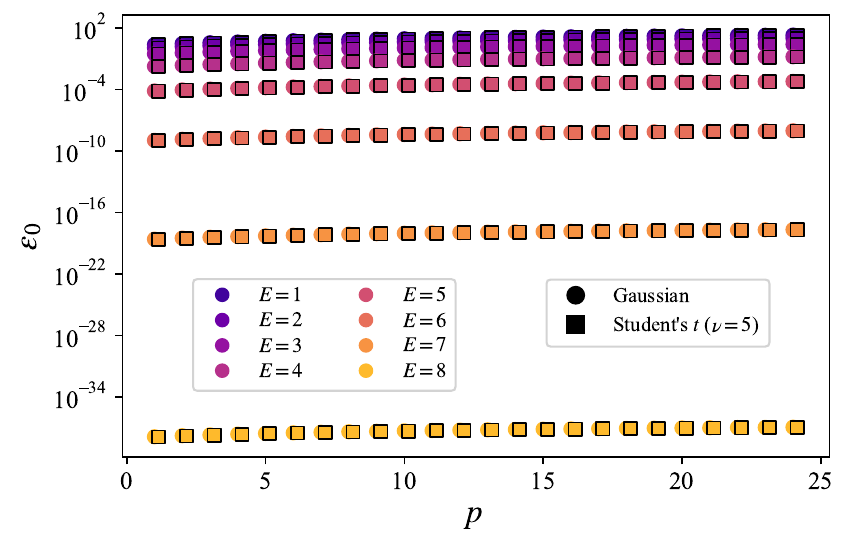}
    \caption{$\varepsilon_{0}$ bound vs.\ precision $p$.}
    \label{fig:eps0_vs_p}
  \end{subfigure}
  \caption{Numerical scale of $C_0$ and $\varepsilon_{0}$ (Corollary~\ref{cor:eps_sp_bound}) for $d=1$, $\sigma=1$.
  Each curve corresponds to a different exponent width $E$.
  Circles denote Gaussian marginals; squares denote Student's $t_5$ marginals.}
  \label{fig:eps0_C0_vs_p}
\end{figure}

\FloatBarrier
\subsubsection*{Approximation 3 -- $\mathbb{E}[\log[|x|/\sqrt{2}]]$}\label{app:FPN_app3}

\begin{figure}[htbp]
  \centering
    \includegraphics[trim={1cm 0.0cm 16cm 0.85cm}, clip,width=0.5\textwidth]{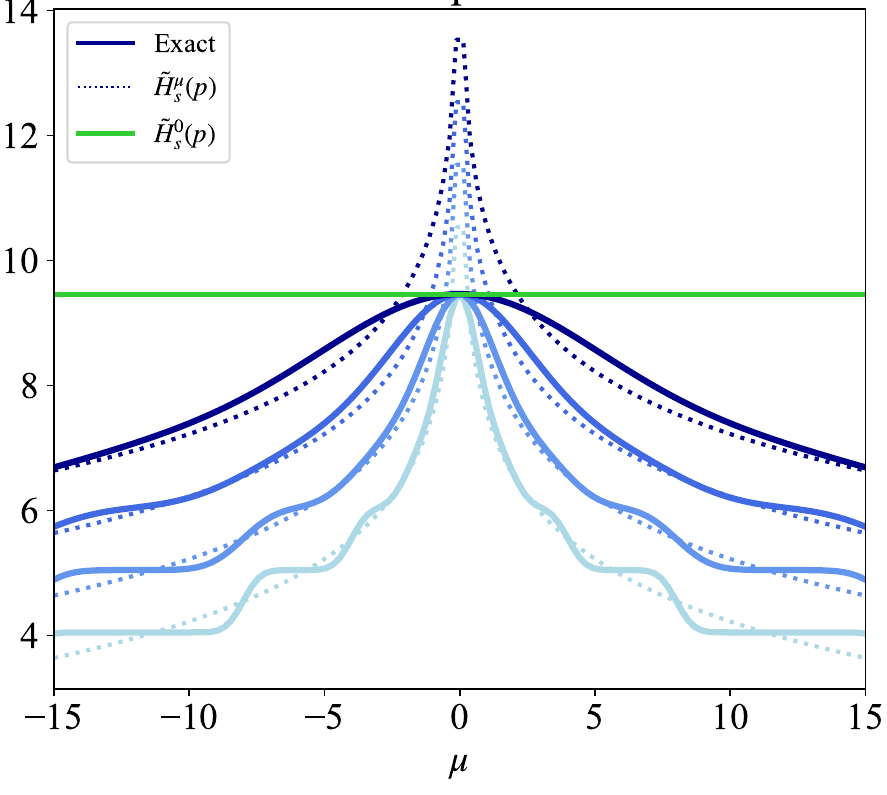}
\caption{\textit{Simulating approximation 3: } The dependence of the entropy of a normally distributed floating-point number on its standard deviation $\sigma$ and mean $\mu$. Dashed lines show Approximation 3 from Eq.~\eqref{eq:three_approxes}, while solid lines show the exact entropy from Corollary~\ref{cor:gauss_fp_ent}.}\label{fig:fp_off_zero}
\end{figure}

We can show that Eq.~\eqref{eq:first_two_approx} approximates Eq.~\eqref{eq:three_approxes} in two cases. First, for a Gaussian distribution where its mean is significantly greater than its standard deviation, $\mu \gg \sigma$. Second, for a distribution $X \sim f_{X}(x,\epsilon)$ that approximates a Dirac delta function as $\epsilon \rightarrow 0$. 

\textit{Non-zero mean Gaussian: }First, for a non-zero mean Gaussian, let $X = \mu + \epsilon$ where $\epsilon \sim \mathcal{N}(0, \sigma^2)$. Without loss of generality, assume $\mu > 0$ (the final result depends only on $|\mu|$). $\log\bracket{\abs{X/\sqrt{2}}}$ becomes\begin{align}
    \log\left|\frac{X}{\sqrt{2}}\right| = \log\left|\frac{\mu + \epsilon}{\sqrt{2}}\right| 
    = \log\left|\frac{\mu}{\sqrt{2}}\left(1 + \frac{\epsilon}{\mu}\right)\right| 
    = \log\left|\frac{\mu}{\sqrt{2}}\right| + \log\left|1 + \frac{\epsilon}{\mu}\right|.\label{eq:split_x}
\end{align}
To control the second term, we can split according to the events
\begin{equation}
    A=\left\{\abs{\frac{\epsilon}{\mu}} \leq \frac{1}{2}\right\} = \left\{X \in \bracket{\frac{\mu}{2}, \frac{3\mu}{2}}\right\}, \qquad A^c= \left\{|\frac{\epsilon}{\mu}| > \frac{1}{2}\right\}.
\end{equation}
Events in $A^c$ are exponentially unlikely when $\mu\gg\sigma$, since $P(A^c)$ is equal to
\begin{align}
    \mathbb P(A^c)& =\int_{-\infty}^{-\mu/2} \frac{\exp[-t^2/2\sigma^2]}{\sigma \sqrt{2\pi}} dt + \int_{\mu/2}^{\infty} \frac{\exp[-t^2/2\sigma^2]}{\sigma\sqrt{2\pi}}\,dt  = 2 \int_{\mu/2}^{\infty} \frac{\exp[-t^2/2\sigma^2]}{\sigma\sqrt{2\pi}}\,dt. \\
    &= 2\int_{\mu/2}^{\infty} \frac{\sigma}{t} \cdot
       \frac{t}{\sigma^2} \cdot
       \frac{\exp[-t^2/2\sigma^2]}{\sqrt{2\pi}}\,dt \leq \frac{4\sigma}{\mu} \int_{\mu/2}^{\infty}
       \frac{t}{\sigma^2} \cdot
       \frac{\exp[-t^2/2\sigma^2]}{\sqrt{2\pi}}\,dt \\
    & = \frac{4\sigma}{\mu\sqrt{2\pi}}\left[-\exp\left[-\frac{t^2}{2\sigma^2}
       \right]\right]_{\mu/2}^{\infty} = \frac{4\sigma}{\mu\sqrt{2\pi}}\,e^{-\mu^2/(8\sigma^2)}\label{eq:p_ac_bound}
\end{align}
meaning that for $\mu \gg \sigma$, $P(A^c)$ is small.

By the law of total expectation,
\begin{equation}\label{eq:total_expect}
    \mathbb{E}\left[\log\left|1 + \frac{\epsilon}{\mu}\right|\right] = \mathbb{E}\left[\log\left|1 + \frac{\epsilon}{\mu}\right|\,\middle|\, A\right]P(A) + \mathbb{E}\left[\log\left|1 + \frac{\epsilon}{\mu}\right|\,\middle|\, A^c\right]P(A^c).
\end{equation}
For the first term, on $A$ we have $1+\epsilon/\mu > 0$, so the Mercator series converges:
\begin{equation}
    \log\left[1 + \frac{\epsilon}{\mu}\right] = \frac{1}{\ln[2]} \left( \frac{\epsilon}{\mu} - \frac{\epsilon^2}{2\mu^2} + O \left (\frac{\epsilon^3}{\mu^3}\right ) \right).
\end{equation}
Since $1+\epsilon/\mu > 0$ on $A$, the absolute value can be dropped. On $A = \{|\epsilon/\mu| \le 1/2\}$, the Mercator remainder satisfies $|\ln[1+x] - x + x^2/2| \le \frac{2}{3}|x|^3$ for $|x| \le 1/2$. Therefore,
\begin{equation}
    \E{\left|\ln\left(1+\frac{\epsilon}{\mu}\right) - \frac{\epsilon}{\mu} + \frac{\epsilon^2}{2\mu^2}\right| \Bigg|A} \le \frac{2}{3}\frac{\E{|\epsilon|^3|A}}{|\mu|^3} = O\left(\frac{\sigma^3}{\mu^3}\right).
\end{equation}
$\mathbb{E}[\epsilon \mid A] = 0$ since the distribution of $\epsilon$ is symmetric over the interval $\bracket{-\mu/2, \, \mu/2}$ and by the monotone convergence theorem as $\mu/\sigma \to \infty$,  $\mathbb{E}[\epsilon^2 \mid A] \to \sigma^2$. Thus,
\begin{equation}
    \mathbb{E}\left[\log\left|1 + \frac{\epsilon}{\mu}\right|\,\middle|\, A\right]P(A) = \frac{P(A)}{\ln[2]} \left( - \frac{\sigma^2}{2\mu^2} + O\left(\frac{\sigma^3}{|\mu|^3}\right) \right) \approx -\frac{\sigma^2}{2\mu^2 \ln[2]}.
\end{equation}

For the term pertaining to $A^c$, Lemma~\ref{lem:Ac-bound} shows that $\mathbb{E}\left[\left|\log\left|1 + \epsilon/\mu\right|\right|\,\middle|\, A^c\right]\mathbb{P}(A^c) \to 0$ as $\mu/\sigma \to \infty$, which means taking the expectation of Eq.~\eqref{eq:split_x} gives
\begin{equation}
    \mathbb{E} \left[\log\left|\frac{X}{\sqrt{2}}\right| \right] \approx \log\left|\frac{\mu}{\sqrt{2}}\right| - \frac{\sigma^2}{2\mu^2 \ln[2]}.
\end{equation}
Since $\frac{\sigma^2}{2\mu^2 \ln[2]} \rightarrow 0$ as $\mu/\sigma \rightarrow \infty$, we see that the convolution integral can be approximated for $\mu \gg \sigma$ as 
\begin{align}
    G_{\mathcal{N}(\mu,\sigma^2)}(\mu)=\int^{\infty}_{-\infty} \frac{\exp(-\frac{(x-\mu)^2}{2\sigma^2})}{\sigma \sqrt{2\pi}} \, \log \bigg[\frac{|x|}{\sqrt{2}} \bigg] dx \approx \log\left[\frac{|\mu|}{\sqrt{2}}\right].
\end{align}
as shown in Fig.~\ref{fig:fp_off_zero}, which shows that the value of the convolution integral is insensitive to the value of $\sigma$ when $\mu \gg \sigma$. 

\noindent This means for a non-zero mean Gaussian random variable, for $\mu \gg \sigma$, $\Ht_s(\Xfp) \approx \Ht_s^\mu(X_{fp})$, with
\begin{equation}
    \label{eq:int_entropy}
    \Ht^\mu_s(\Xfp) \triangleq \frac{1}{2}\log[2\pi e \sigma^2] + (p-1) - \log\bracket{|\mu| / \sqrt{2}}.
\end{equation}

\begin{lemma}[Bound on the Mercator Series Remainder]\label{lem:mercator-remainder}
For all $|x| \le \frac{1}{2}$,
\begin{equation}
  \abs{\ln[1+x] - x + \frac{x^2}{2}} \leq \frac{2}{3}\,|x|^3.
\end{equation}
\end{lemma}
 
\begin{proof}
The Mercator series gives
\begin{equation}
  \ln\bracket{1+x} \;=\; \sum_{k=1}^{\infty} \frac{(-1)^{k+1}\, x^k}{k},
  \qquad |x| < 1,
\end{equation}
so the second-order remainder is
\begin{equation}
  R_2(x)
  \triangleq \ln\bracket{1+x} - x + \frac{x^2}{2}
  \;=\; \sum_{k=3}^{\infty} \frac{(-1)^{k+1}\, x^k}{k}.
\end{equation}
By the triangle inequality and using the inequality $1/k \le 1/3$ for every $k \ge 3$,
\begin{equation}
  |R_2(x)|
  \;\le\; \sum_{k=3}^{\infty} \frac{|x|^k}{k}
  \;\le\; \frac{1}{3} \sum_{k=3}^{\infty} |x|^k
  \;=\; \frac{1}{3} \cdot \frac{|x|^3}{1 - |x|},
\end{equation}
where the last equality is the geometric series
$\sum_{k=3}^{\infty} |x|^k = |x|^3 \sum_{j=0}^{\infty} |x|^j = |x|^3/(1-|x|)$,
valid for $|x| < 1$.
Since $|x| \le \frac{1}{2}$ implies $1 - |x| \ge \frac{1}{2}$, we obtain
\begin{equation}
  |R_2(x)|
  \;\le\; \frac{|x|^3}{3 \paren{\frac{1}{2}}}
  \;=\; \frac{2}{3}\,|x|^3.
\end{equation}
\end{proof}

\begin{lemma}\label{lem:Ac-bound}
Let $\epsilon \sim N(0,\sigma^2)$, $\mu \neq 0$, and $A = \{|\epsilon/\mu| \le \frac{1}{2}\}$.
Then,
\begin{equation}
  \E{\abs{\ln\abs{1+\frac{\epsilon}{\mu}}}\,\, \Bigr|\,A^c}P(A^c)
  \leq \,\paren{\frac{C_1\sigma}{|\mu|} + \frac{C_2|\mu|}{\sigma}}\, e^{-|\mu|^2/8\sigma^2}
\end{equation}
for some $C_1, C_2 > 0$.
In particular, the bound decays faster than any polynomial in~$|\mu|/\sigma$.
\end{lemma}

\begin{proof}
Write $Z = \epsilon/\sigma \sim N(0,1)$, so that $\epsilon/\mu = \sigma Z/\mu$
(taking $\mu>0$ without loss of generality; the case $\mu<0$ is symmetric).
Set $\delta \triangleq 1 + \epsilon/\mu = 1 + \sigma Z/\mu$ and split
\begin{equation}
  A^c = A_1^c \cup A_2^c, \qquad
  A_1^c \triangleq \left\{\abs{\frac{\epsilon}{\mu}}>\frac{1}{2}\right\} \cap \left\{\abs{\delta}\ge \frac{1}{2}\right\},\qquad
  A_2^c \triangleq \left\{\abs{\frac{\epsilon}{\mu}}>\frac{1}{2}\right\} \cap \left\{\abs{\delta} < \frac{1}{2}\right\}.
\end{equation}

\medskip
\noindent\textit{Contribution of~$A_1^c$.}\enspace
For $t \ge \frac{1}{2}$ we have $|\ln t| \le \ln 2 + t$
(since $-\ln t \leq \ln 2$ when $t \geq \frac{1}{2}$,
and $\ln t \le t$ when $t > 0$). Hence on~$A_1^c$,
\begin{equation}
  |\ln{\abs{\delta}}| \;\le\; \ln 2 + |\delta|
  \;\le\; \ln 2 + 1 + |\frac{\epsilon}{\mu}|
\end{equation}
almost surely.
Therefore
\begin{align}
&\abs{\ln|\delta|}\indic{A_1^c} \;\le\; \paren{\ln 2 + 1 + |\frac{\epsilon}{\mu}|}\indic{A_1^c}\\
  \implies &\E{\abs{\ln\bracket{|\delta|}}\indic{A_1^c}}
  \;\le\; \E{\paren{\ln 2 + 1 + |\frac{\epsilon}{\mu}|}\indic{A_1^c}}\\
  \implies &\E{\abs{\ln\bracket{|\delta|}}\indic{A_1^c}}
  \;\le\; (1+\ln 2)\,P(A_1^c) + \E{\abs{\frac{\sigma Z}{\mu}}\indic{A_1^c}}\\
    \implies &\E{\abs{\ln\bracket{|\delta|}}\indic{A_1^c}}
  \;\le\; (1+\ln 2)\,P(A^c) \;+\;
    \frac{\sigma}{\mu}\E{|Z|\,\mathbf{1}_{|Z|>\mu/2\sigma}},
\end{align}
since $A_1^c \subset A^c \implies P(A_1^c) \leq P(A^c)$ and $\indic{A_1^c} \leq \indic{A^c}$ almost surely.
From Eq.~\eqref{eq:p_ac_bound}, the Gaussian tail bound gives
$P(A^c) \leq \frac{4\sigma}{\mu\sqrt{2\pi}}\,e^{-\mu^2/(8\sigma^2)}$,
and the truncated first moment evaluates to
\begin{equation}
  \mathbb{E}\bigl[|Z|\,\mathbf{1}_{|Z|>a}\bigr]
  = \frac{2}{\sqrt{2\pi}}\int_a^\infty ze^{-\frac{z^2}{2}}dz = \frac{2}{\sqrt{2\pi}}\,e^{-a^2/2},
  \qquad a = \frac{\mu}{2\sigma}.
\end{equation}
Since \begin{equation}
    \E{\abs{\ln\abs{\delta}} \,|\, A_1^c}P(A_1^c) = \E{\abs{\ln\abs{\delta}}\indic{A_1^c}} \leq \paren{\ln 2 + \frac{3}{2}}\frac{4\sigma}{\mu\sqrt{2\pi}}\,e^{-\mu^2/(8\sigma^2)},
\end{equation}
the entire $A_1^c$~contribution is at most
$\frac{C_1\sigma}{\mu}\,e^{-\mu^2/(8\sigma^2)}$ for some constant~$C_1$.

\medskip
\noindent\textit{Contribution of~$A_2^c$.}\enspace
On the event $\{|\delta|<\frac{1}{2}\}$ we have
$\epsilon/\mu \in (-\frac{3}{2},-\frac{1}{2})$,
i.e.\ $Z \in (-\frac{3\mu}{2\sigma},-\frac{\mu}{2\sigma})$.
The density of~$\delta$ at a point~$t$ is
$\frac{\mu}{\sigma}f_Z\paren{\frac{\mu}{\sigma}(t-1)}$,
where $f_Z$ is the standard normal density.
For $\abs{t}<\frac{1}{2}$ we have $(t-1)^2 \ge \frac{1}{4}$, giving the
uniform density bound
\begin{equation}
  \frac{\mu}{\sigma}f_Z\paren{\frac{\mu}{\sigma}(t-1)}
  = \frac{\mu}{\sigma\sqrt{2\pi}}e^{-\frac{\mu^2(t-1)^2}{2\sigma^2}}
  \leq \frac{\mu}{\sigma\sqrt{2\pi}}\,e^{-\frac{\mu^2}{8\sigma^2}}.
\end{equation}
Since $\abs{\delta} < \frac{1}{2}$ almost surely implies $\abs{\frac{\epsilon}{\mu}}>\frac{1}{2}$ almost surely, we have $A_2^c = \{\abs{\delta} < \frac{1}{2}\}$ and
\begin{equation}
  \E{\abs{\ln\abs{\delta}}\,\indic{|\delta|<1/2}}
  = \int_{-1/2}^{1/2}\frac{\mu}{\sigma}f_Z\paren{\frac{\mu}{\sigma}(t-1)}\abs{\ln\abs{t}}dt
  \leq
  \frac{\mu}{\sigma\sqrt{2\pi}}\,e^{-\frac{\mu^2}{8\sigma^2}}
  \int_{-1/2}^{1/2} |\ln\abs{t}\,|\,dt.
\end{equation}
The integral evaluates to
$-2\int_0^{1/2}\ln\bracket{t}\,dt = -2\bracket{t\ln t -t}\Bigr|_0^{1/2}
= \ln 2 + 1$,
so the $A_2^c$~contribution is at most
$\frac{C_2\mu}{\sigma}\,e^{-\frac{\mu^2}{8\sigma^2}}$ for some constant $C_2$.

\medskip
\noindent\textit{Combining both contributions.}\enspace
Adding both contributions yields
\begin{equation}
  \E{\abs{\ln\abs{1+\frac{\epsilon}{\mu}}}\,\mathbf{1}_{A^c}}
  \;\le\; \paren{\frac{C_1\sigma}{\mu} + \frac{C_2\mu}{\sigma}} e^{-\frac{\mu^2}{8\sigma^2}}
\end{equation}
which decays faster than any polynomial in~$|\mu|/\sigma$.
\end{proof}

\textit{Approximating delta distribution: } Next, we can show the same result for a function which approximates a delta function as $\varepsilon \rightarrow 0$, provided $\mu \neq 0$. For a random variable $X$ where $f_X(x-\mu; \varepsilon) \rightarrow \delta (x-\mu)$ as $\varepsilon \rightarrow 0$ with $\mu \neq 0$, Eq.~\eqref{eq:G_int_non_gauss} becomes
\begin{equation}
    G_{f_X}(\mu) \triangleq \int^{\infty}_{-\infty} f_X(x-\mu; \epsilon) \log \bigg[\frac{|x|}{\sqrt{2}} \bigg] dx.
\end{equation}
As $\epsilon \rightarrow 0$, for any $\mu \neq 0$ such that $\log[|x| / \sqrt{2}]$ is continuous,
\begin{align*}
    G_{f_X}(\mu) & \approx \int^{\infty}_{-\infty} \delta(x-\mu) \log \bigg[\frac{|x|}{\sqrt{2}} \bigg] dx \\
    & = \log \bigg[\frac{|\mu|}{\sqrt{2}} \bigg].
\end{align*}
Thus for any distribution, which may be non-Gaussian, if it approaches $\delta(\cdot)$ as a parameter $\epsilon$ which defines it approaches 0 and $\mu \neq 0$, in the limit of $\epsilon \rightarrow 0$ the integral term satisfies $G_{f_X}(\mu) \approx \log[|\mu|/\sqrt{2}]$, giving $\Ht_s(X_{fp}) \approx \Ht_s^\mu(X_{fp})$, where
\begin{equation}
    \Ht^\mu_s(\Xfp) \triangleq h(X) + (p-1) - \log\bracket{|\mu| / \sqrt{2}}.
\end{equation}

\FloatBarrier
\pagebreak
\section{Applying approximations to the Gaussian case}\label{app:FPN_app_gauss}

Here the theorems for the floating point entropy of Gaussian random variables are reviewed. See Figure~\ref{fig:fp_exponents} for a directly computed histogram of exponent values. Directly computing the entropy of the histograms returns a discrete entropy of $\sim 2.54$ bits, which is close to the offset factor above the mantissa contribution $p$ in $\Ht^0_s(p) \approx p + 2.46$ bits. Since bin sizes are fairly large with respect to the underlying distribution of exponents, this divergence likely can be explained by quantization noise.

\begin{figure}[htbp]
  \centering
    \includegraphics[width=0.36\linewidth]{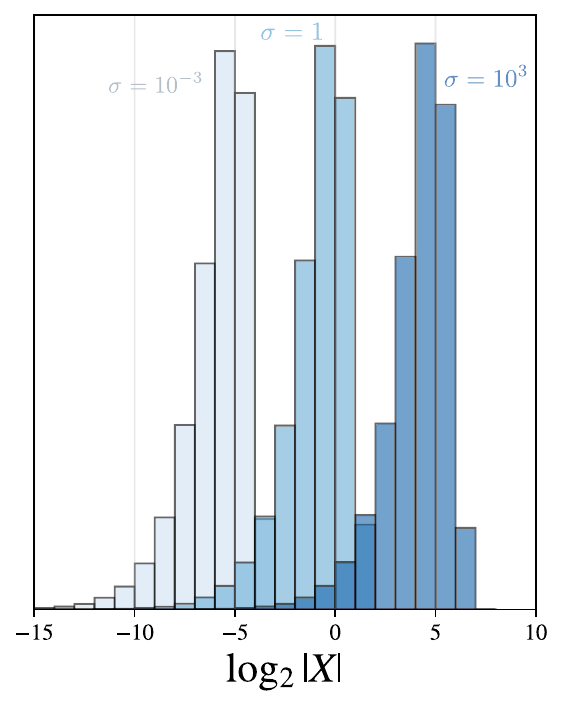}
\caption{\textit{Histogram of exponent values for a normally distributed random variable.} The distribution of exponent states for $X \sim \mathcal{N}(0,\sigma^2)$. From left to right, the distributions plotted have $\sigma = \{10^{-3}, 10^0, 10^3 \}$. All three have a discrete entropy $H(\log[X]) \approx 2.54$ bits. These distributions conform well to observational data in \cite{bordawekar2022efloatentropycodedfloatingpoint, hao_neuzip_2024}.}
    \label{fig:fp_exponents}
\end{figure}

\begin{theorem}[Approximating the entropy of a mean-zero univariate gaussian]
    \label{thm:single_gaussian}
    The discrete entropy of a zero-mean normally distributed continuous random variable $X \sim \mathcal{N}(0, \sigma_x^2)$ as represented on a floating-point number, with the approximations specified by Eq.~\eqref{eq:diff_approx} and~\eqref{eq:step_approx} applied, is a constant specified by the floating point number's precision, $p$.
    \begin{equation}
        \tilde H_s(X_{fp}) = (p-1) + \frac{1}{2} \log[2 \pi e] + 1 + \frac{\gamma_e}{2\ln[2]}
    \end{equation}
    
\end{theorem}
\begin{proof}
    Let $X \sim \mathcal{N}(0, \sigma_x^2)$.
    \begin{equation}
    \begin{aligned}
    \label{eq:1d-0m-fp}
        \Ht_{s}(\Xfp) &\triangleq h(X) - \E{\log[\Delta_s(x)]} \\
        &= \frac{1}{2}\log[2\pi e \sigma^2_x] - \int^{\infty}_{-\infty} \frac{\exp(-\frac{x^2}{2\sigma_x^2})}{\sigma_x \sqrt{2\pi}} \log [\Delta_s(x)] \, dx  \\
        &= \frac{1}{2}\log[2\pi e \sigma^2_x] - \int^{\infty}_{-\infty} \frac{\exp(-\frac{x^2}{2\sigma_x^2})}{\sigma_x \sqrt{2\pi}} (\log[|x|/\sqrt{2}] - (p - 1)) \, dx \\
        & = \frac{1}{2}\log[2\pi e \sigma^2_x] + (p - 1) - 2\int^{\infty}_0 \frac{\exp(-\frac{x^2}{2\sigma_x^2})}{\sigma_x \sqrt{2\pi}} \log[x/\sqrt{2}] \, dx\\
        & = \frac{1}{2}\log[2\pi e \sigma^2_x] + (p-1) + \bigg(1 + \frac{\gamma_e}{2\ln[2]} -\frac{1}{2} \log[\sigma^2_x] \bigg)\\
        & = (p-1) + \frac{1}{2} \log[2 \pi e] + 1 + \frac{\gamma_e}{2\ln[2]}\\
        & = p + \frac{1}{2} \log[2 \pi e] + \frac{\gamma_e}{2\ln[2]}.
    \end{aligned}
\end{equation}
\end{proof}

\begin{theorem}
    \label{thm:double_gaussian}
    Let $X \sim \mathcal{N}(0,\sigma_x^2)$, $\Xi \sim \mathcal{N}(0,\sigma_\xi^2)$, $w \in \mathbb{R}$, and $Y = wX + \Xi$. The approximate joint discrete entropy of $(\Xfp, \Yfp)$ is,
\begin{equation*}
    \Ht_s(\Xfp, \Yfp) = 2 \Ht^0_s(p)-\frac{1}{2} \log \bigg [ 1+\frac{\sigma_x^2 w^2}{\sigma_\xi^2} \bigg].
\end{equation*}
\end{theorem}
\begin{proof}
\begin{equation}
\label{eq:approx_joint_ent_SNR_start}
    \begin{aligned}
        \Ht_s(\Xfp, \Yfp) &\triangleq h(X, Y) - \iint_{-\infty}^\infty f_{XY}(x,y) ( \log[\Delta_s(x)] + \log[\Delta_s(y)])dxdy \\
        = h& (X)  + h(\Xi) - \iint_{-\infty}^\infty f_{XY}(x,y)  \log[\Delta_s(x)] dxdy - \iint_{-\infty}^\infty f_{XY}(x,y)\log[\Delta_s(y)]dxdy \\
        = \frac{1}{2}&\log[2\pi e \sigma_x^2] + \frac{1}{2}\log[2\pi e \sigma_\xi^2] - \int_{-\infty}^\infty f_{X}(x)  \log[\Delta_s(x)] dx - \int_{-\infty}^\infty f_{Y}(y)\log[\Delta_s(y)]dy.
    \end{aligned}
\end{equation}
where $h(X,Y) = h(X) + h(\Xi)$ is because $X$ is independent from $\Xi$ and $(X,\Xi) \mapsto (X,Y)$ is an invertible linear transformation with Jacobian 1. The expressions follow directly from the differential entropy of a normal distribution \cite{CoverThomas}. In the last line, $f_X(x)$ and $f_Y(y)$ are the marginal distributions over $f_{XY}(x,y)$,
\begin{equation}
    f_X(x) = \int_{-\infty}^\infty f_{XY}(x,y) \, dy\text{,  }f_Y(y) = \int_{-\infty}^\infty f_{XY}(x,y) \, dx,
\end{equation}
where from Eq.~\eqref{eq:initial_dist}
\begin{equation}
    f_X(x) =
\frac{1}{\sqrt{2\pi\sigma_x^2}}
\exp\left(-\frac{x^2}{2\sigma_x^2}\right)\text{,  } f_Y(y) =
\frac{1}{\sqrt{2\pi(\sigma_\xi^2+w^2\sigma_x^2)}}
\exp\left(-\frac{y^2}{2(\sigma_\xi^2+w^2\sigma_x^2)}\right).
\end{equation}
Applying the same steps as in Theorem~\ref{thm:single_gaussian} to evaluate the bin terms, we obtain
\begin{equation}
    \int_{-\infty}^{\infty} f_X(x)\,\log\bracket{\Delta_s(x)}\,dx
\approx
-\left(p+\frac{\gamma_e}{2\ln 2}\right)
+\frac{1}{2}\log \sigma_x^2,
\end{equation}

\begin{equation}
    \int_{-\infty}^{\infty} f_Y(y)\,\log\bracket{\Delta_s (y)}\,dy
    \approx
    -\left(p+\frac{\gamma_e}{2\ln 2}\right)
    +\frac{1}{2}\log\bracket{\sigma_\xi^2+w^2\sigma_x^2}.
\end{equation}
Plugging these back into the original expression, we obtain
\begin{equation}
    \label{eq:approx_joint_ent_SNR_app}
    \begin{aligned}
        \Ht_s(\Xfp, \Yfp) &= \frac{1}{2}\log[2\pi e \sigma_x^2] + \frac{1}{2}\log[2\pi e \sigma_\xi^2] + p+\frac{\gamma_e}{2\ln 2} - \frac{1}{2}\log \sigma_x^2\\
        & + p + \frac{\gamma_e}{2\ln[2]} - \frac{1}{2} \log[\sigma^2_\xi + \sigma_x^2 w^2].
    \end{aligned}
\end{equation}
Using the definition $H^0_s(p)$ from definition Eq.~\eqref{eq:single_mean_0} and Theorem~\ref{thm:single_gaussian}, we can obtain a clean final expression,
\begin{equation}
    \Ht_s(\Xfp, \Yfp) = 2 \Ht^0_s(p)-\frac{1}{2} \log \bigg [ 1+\frac{\sigma_x^2 w^2}{\sigma_\xi^2} \bigg]
\end{equation}
\end{proof}

\section{SGD Fit Quality}

Here we review the fit quality of approximations introduced to obtain analytic expressions for the entropy of model parameters along stochastic gradient descent. Fig.~\ref{fig:assymp_fit} shows the asymptotic fit quality, while Fig.~\ref{fig:non_assymp_fit} shows the nonasymptotic fit quality.

\begin{figure}[htbp]
  \centering
    \includegraphics[width=0.3\linewidth]{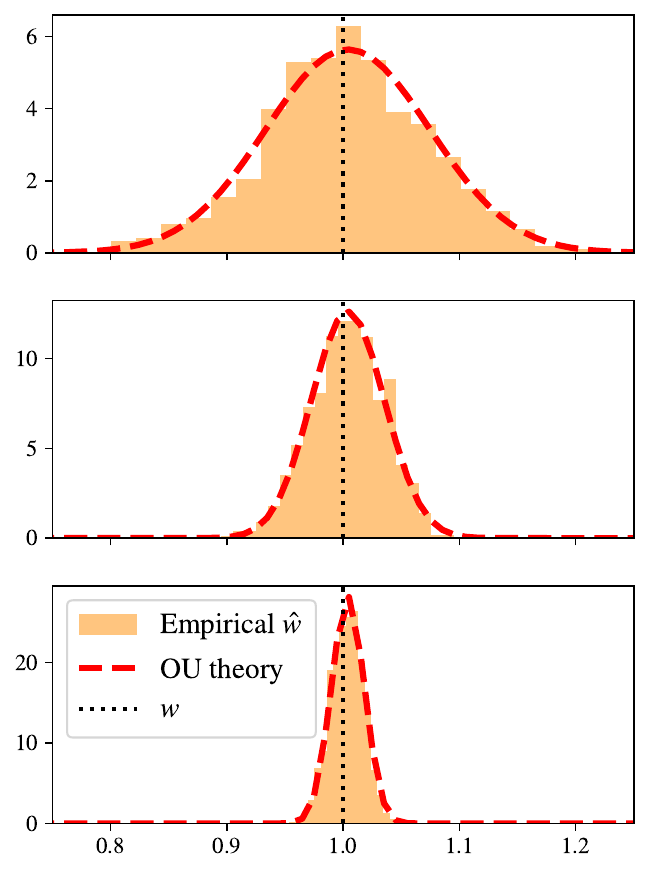}
  \caption{\textit{Fit quality for asymptotic stochastic gradient descent.} Empirical investigation of the validity of the continuous Ornstein-Uhlenbeck process approximation with a simulation of $\hat w$ with $\eta = 0.01$, $\tau = 200$, $\sigma_x^2 = 1$, and $\sigma_\xi^2=1$, for 1000 trials. From top to bottom, the batch sizes are $B = \{1,5,25\}$ while $\tau = 200$, showing the approximation is already effective for these parameters at low $B$. }
  \label{fig:assymp_fit}
\end{figure}

\begin{figure}[htbp]
  \centering
  \begin{subfigure}{0.42\textwidth}
    \centering
    \includegraphics[width=\linewidth]{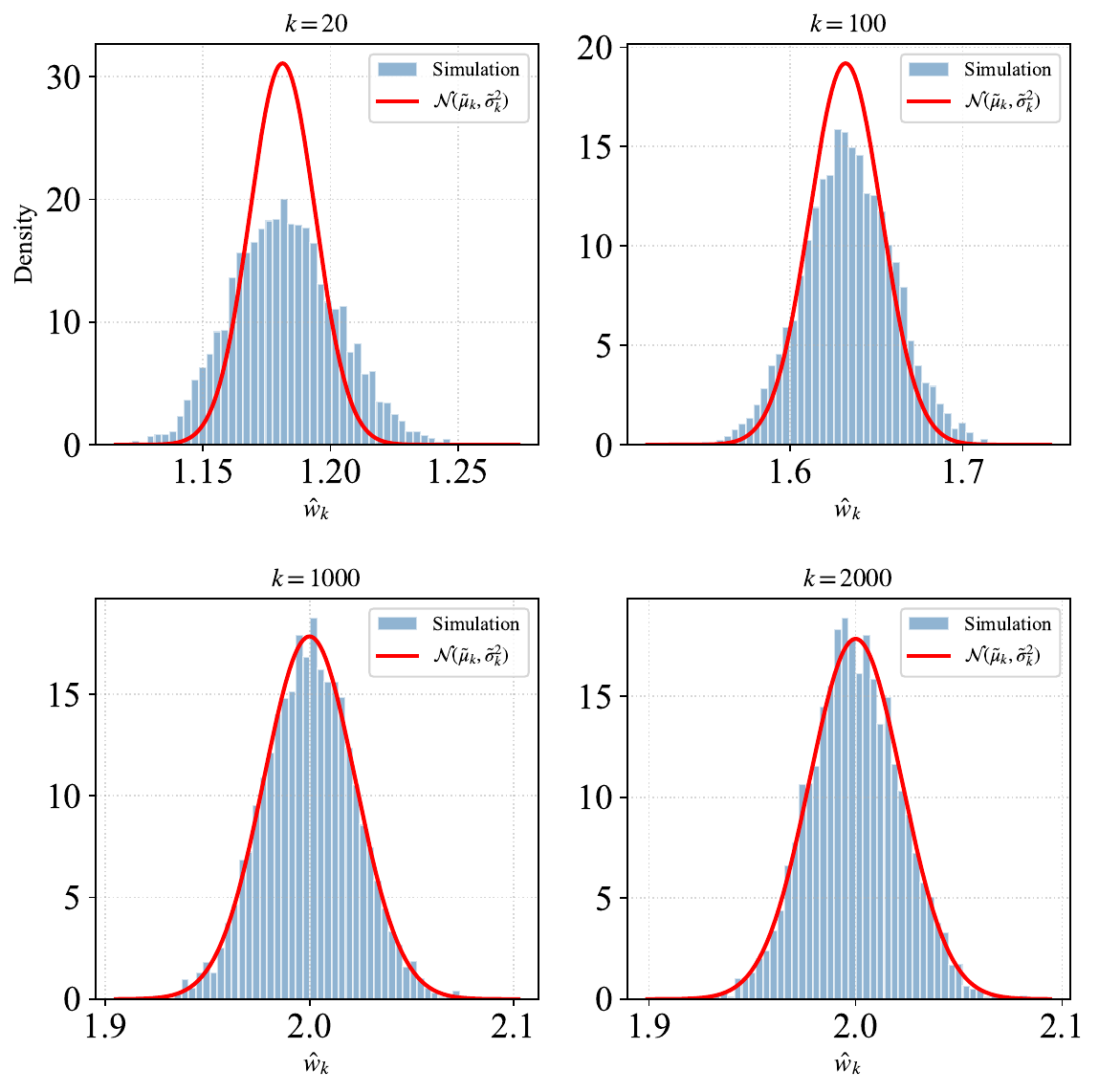}
    \caption{}
    \label{fig:non_assymp_fit}
  \end{subfigure}
  \begin{subfigure}{0.3\textwidth}
    \centering
    \includegraphics[width=\linewidth]{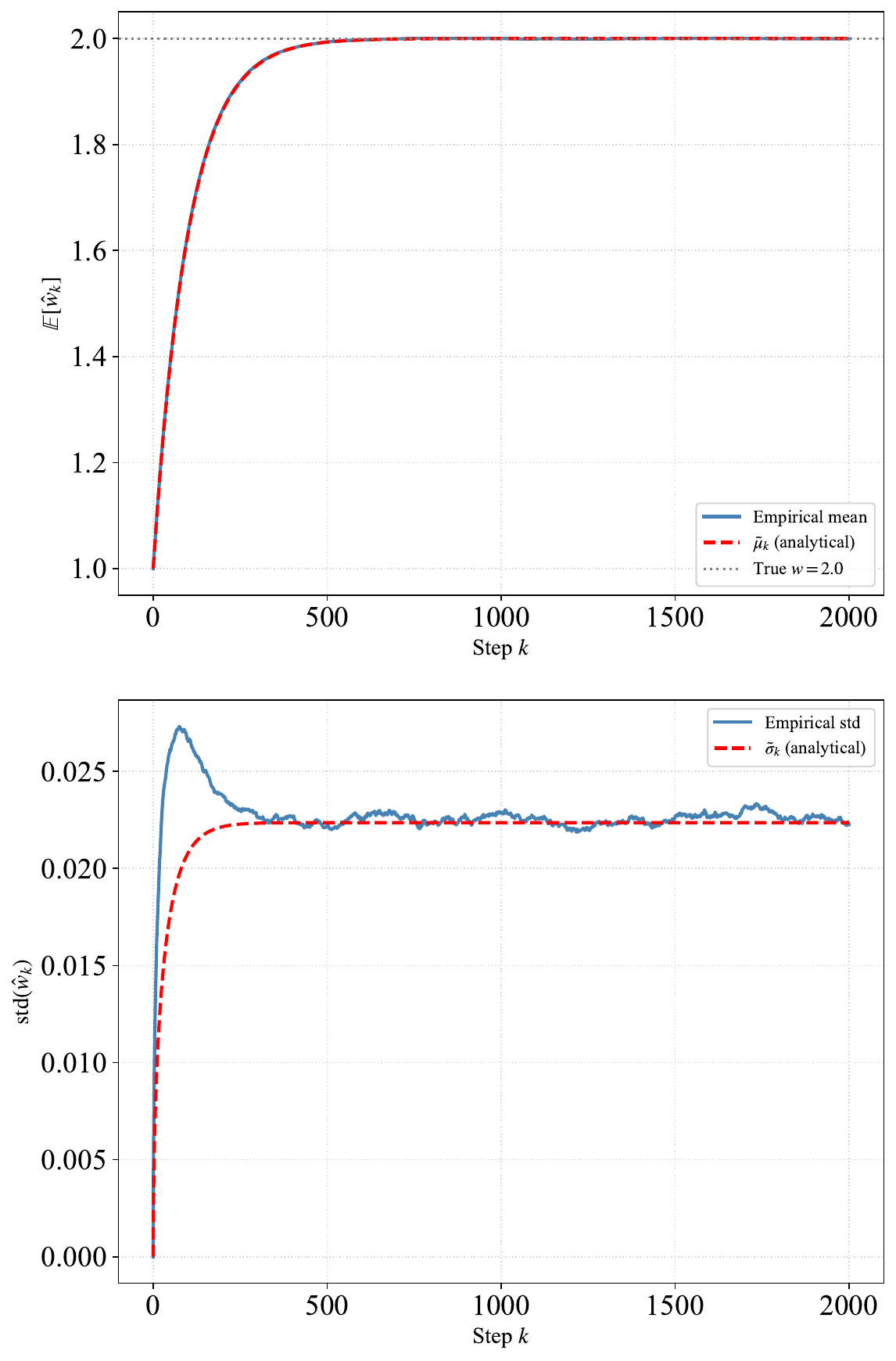}
    \caption{}
    \label{fig:mean_fit}
  \end{subfigure}
  \caption{\textit{SGD dynamics for zero-intercept simple linear regression.} (\ref{fig:non_assymp_fit}) SGD parameter $\hat w$ distribution at selected iterations, for 5000 separate trials of running SGD. $\sigx^2 = 1$, $\sigxi^2 = 1$, $\eta = 0.01$, $B = 10$. (\ref{fig:mean_fit}) Theoretical mean and standard deviation compared to the empirical mean and standard deviation as step number $k$ increases up to final value $k = \tau = 2000$.}
  \label{fig:sgd_w_2x2}
\end{figure}

\FloatBarrier
\pagebreak
\section{Landauer Cost of Averaging and Summing}\label{sec:avg_sum}
An even simpler prediction algorithm than linear regression is the average. Let the input to this algorithm be $X_1,\dots, X_n \iid f$, where $\E{X_i} = \mu$ and $\mathrm{var}\paren{X_i} = \sigma^2$. The differential entropy of the input is $h(X_1, \dots, X_n) = nh(X_1)$.

If $f$ is not Gaussian, for sufficiently large $n$, central limit theorem can be invoked to show that $\frac{1}{n}\sum\limits_{i = 1}^n X_i \Rightarrow \gauss{\mu}{\frac{\sigma^2}{n}}$ as $n \rightarrow \infty$. So the differential entropy of the final distribution is $h(\frac{1}{n}\sum\limits_{i = 1}^n X_i) \approx \frac{1}{2}\log\bracket{2\pi e \frac{\sigma^2}{n}}$. When $f_{X_i}$ is Gaussian, this is true even for small $n$.

For the case when $ \forall i, f_{X_i} \sim \gauss{0}{1}$ and the random variables are encoded in floating-point numbers with precision $p$, the discrete entropy difference is
\begin{align}
    \Delta H = n\Ht^0_{s}(p) - \Ht^0_{s}(p) = (n-1)\Ht^0_{s}(p) \approx (n-1)(p + 2.46) \text{ bits,}\label{eq:landauer_avg}
\end{align}
where $\Ht^0_{s}(p)$ is derived in Section~\ref{sec:fpn}. Notice that Eq.~\eqref{eq:landauer_avg} is dominated by the number of inputs. As stressed in Section~\ref{sec:fpn}, $\Ht^0_{s}(p)$ does not depend on the variance of the random variable. This means that for zero-mean Gaussian random variables, the Landauer cost of averaging $\frac{1}{n}\sum\limits_{i = 1}^n X_i$ and summing $\sum\limits_{i = 1}^n X_i$ is equivalent, since $\sum\limits_{i = 1}^n X_i \sim \gauss{0}{n\sigma^2}$.

\FloatBarrier
\pagebreak
\section{Computing the Probability Distribution of Z}\label{sec:compute_f_z}

\begin{figure}[htbp]
  \centering
  \begin{subfigure}{0.33\textwidth}
    \centering
    \includegraphics[width=\linewidth]{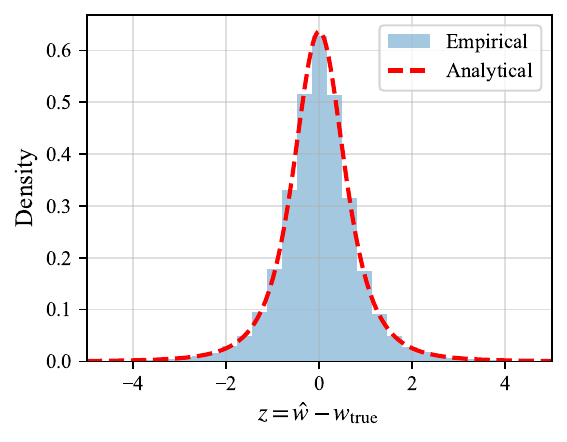}
    \caption{$n = 3$}
    \label{fig:z_fit_3}
  \end{subfigure}\hfill
  \begin{subfigure}{0.33\textwidth}
    \centering
    \includegraphics[width=\linewidth]{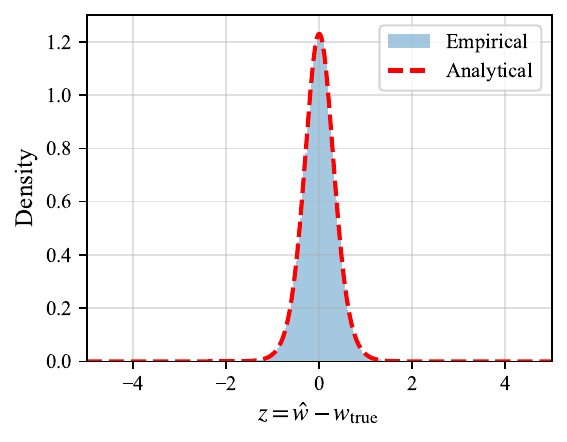}
    \caption{$n=10$}
    \label{fig:z_fit_10}
  \end{subfigure}\hfill
  \begin{subfigure}{0.33\textwidth}
    \centering
    \includegraphics[width=\linewidth]{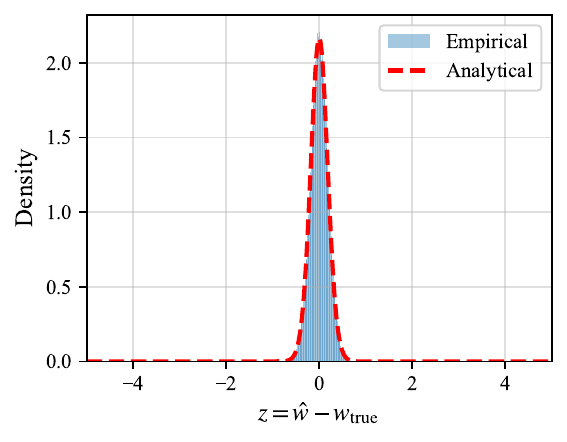}
    \caption{$n=30$}
    \label{fig:z_fit_30}
  \end{subfigure}
  \caption{\textit{The probability density function} $f_Z(z)$, with $\sigx = \sigxi = 1$. As $n$ increases, the distribution becomes more peaked around $z = 0$. The simulation is of 50000 trials.} 
  \label{fig:z_fit}
\end{figure}

\begin{lemma}
    \label{lem:compute_f_z}
    Let $\bX \sim \mathcal{N}(0, \sigx^2 I_n)$ and $\bXi \sim \mathcal{N}(0, \sigxi^2 I_n)$ be independent, with $n \in \mathbb{N}$, and define $Z = \frac{\bX^T\bXi}{\bX^T\bX}$. The probability density function of $Z$ is
\begin{align}
    f_Z(z) = \sqrt{\frac{1}{\pi\paren{\sigx^2}^n\sigxi^2}}\frac{\Gamma\paren{\frac{n+1}{2}}}{\Gamma\paren{\frac{n}{2}}}\paren{\frac{\sigx^2\sigxi^2}{\sigx^2z^2 + \sigxi^2}}^{\frac{n + 1}{2}}
\end{align}
which is the probability density function of a scaled Student's t-distributed random variable with $n$ degrees of freedom and a scale of $\frac{\sigxi}{\sigx\sqrt{n}}$ \cite{JohnsonNormanLloyd1994Cud}. 
\end{lemma}

\begin{proof}[Proof of Lemma~\ref{lem:compute_f_z}]
Recall that $Z = \frac{\xv^T\xiv}{\xtx}$. Using the law of total probability, we have
\begin{align}
    f_{Z}(z) &= \int f_{Z|\xv} (z|\xv)f_\xv(\xv)d\xv.
\end{align}
We see that $f_{Z|\xv} (z|\xv) \sim \gauss{0}{\frac{\sigxi^2}{\xv^T\xv}}$. Let $\xt \sim \gauss{0}{\frac{\sigx^2\sigxi^2}{\sigx^2z^2 + \sigxi^2}I}$. This means,
\begin{align}
    f_{Z}(z) &= \int \frac{\sqrt{\xtx}}{\sqrt{2\pi\sigxi^2}}e^{\frac{-(\xv^T\xv)z^2}{2\sigxi^2}}\paren{\frac{1}{\sqrt{2\pi\sigx^2}}}^ne^{-\frac{\xv^T\xv}{2\sigx^2}} d\xv\\
    &= \frac{1}{\sqrt{2\pi\paren{\sigx^2}^n\sigxi^2}}\paren{\frac{\sigx^2\sigxi^2}{\sigx^2z^2 + \sigxi^2}}^{\frac{n}{2}}\int \sqrt{\xtx}\paren{\frac{1}{\sqrt{2\pi}}}^n \paren{\frac{\paren{\sigx^2}^n\sigxi^2}{\sigx^2z^2 + \sigxi^2}}^{-\frac{n}{2}} e^{-\frac{\xtx}{2}\bracket{\frac{z^2}{\sigxi^2} + \frac{1}{\sigx^2}}}d\xv\\
    &= \frac{\E{\sqrt{\xttxt}}}{\sqrt{2\pi\paren{\sigx^2}^n\sigxi^2}}\paren{\frac{\sigx^2\sigxi^2}{\sigx^2z^2 + \sigxi^2}}^{\frac{n}{2}}.
\end{align}

Let $c \triangleq \frac{\sigx^2\sigxi^2}{\sigx^2z^2 + \sigxi^2}$ and let $\mathbf{g} \sim \mathcal{N}(0, I_n)$. We have $\E{\sqrt{\xttxt}} = \sqrt{c}\,\E{\sqrt{\mathbf{g}^T\mathbf{g}}} = \sqrt{2c}\frac{\Gamma\paren{\frac{n+1}{2}}}{\Gamma\paren{\frac{n}{2}}}$, where the last equality uses $\sqrt{\mathbf{g}^T\mathbf{g}} \sim \chi(n)$ \cite{JohnsonNormanLloyd1994Cud}. Therefore,
\begin{align}
    f_Z(z) &= \sqrt{\frac{1}{\pi\paren{\sigx^2}^n\sigxi^2}}\sqrt{\frac{\sigx^2\sigxi^2}{\sigx^2z^2 + \sigxi^2}}\frac{\Gamma\paren{\frac{n+1}{2}}}{\Gamma\paren{\frac{n}{2}}}\paren{\frac{\sigx^2\sigxi^2}{\sigx^2z^2 + \sigxi^2}}^{\frac{n}{2}}\\
    &= \sqrt{\frac{1}{\pi\paren{\sigx^2}^n\sigxi^2}}\frac{\Gamma\paren{\frac{n+1}{2}}}{\Gamma\paren{\frac{n}{2}}}\paren{\frac{\sigx^2\sigxi^2}{\sigx^2z^2 + \sigxi^2}}^{\frac{n + 1}{2}}\label{eq:f_z_appendix}
\end{align}
This distribution is depicted in Fig.~\ref{fig:z_fit} for varying values of $n$.
\end{proof}

\FloatBarrier
\newpage
\section{Computing the Differential Entropy of Z}\label{sec:compute_h_z}

\begin{lemma}\label{lem:diff_ent_Ex}
    Under the hypotheses of Lemma~\ref{lem:compute_f_z}, let $\psi(x) \triangleq \frac{d}{dx}\ln{\Gamma(x)} = \frac{\Gamma^\prime(x)}{\Gamma(x)}$ be the digamma function. The differential entropy of $Z$ is
\begin{align}
    \label{eq:diff_ent_Ex}
    h(Z) &= \log\bracket{e}\paren{\ln\bracket{\frac{\sigxi}{\sigx}} - \ln{\paren{\frac{\Gamma\paren{\frac{n+1}{2}}}{\sqrt{\pi}\Gamma\paren{\frac{n}{2}}}}} + \frac{n + 1}{2}\bracket{\psi\paren{\frac{n+1}{2}} - \psi\paren{\frac{n}{2}}}}.
\end{align}
\end{lemma}

\begin{proof}[Proof of Lemma~\ref{lem:diff_ent_Ex}]
    Using the distribution derived in Lemma~\ref{lem:compute_f_z}, the differential entropy of $Z$ is
\begin{align}
    &\frac{1}{\log\bracket{e}}h(Z) = -\int\limits_{-\infty}^\infty f_Z(z)\ln{f_Z(z)}dz\\
     &= -\int\limits_{-\infty}^\infty \paren{\sqrt{\frac{1}{\pi\paren{\sigx^2}^n\sigxi^2}}\frac{\Gamma\paren{\frac{n+1}{2}}}{\Gamma\paren{\frac{n}{2}}}\paren{\frac{\sigx^2\sigxi^2}{\sigx^2z^2 + \sigxi^2}}^{\frac{n + 1}{2}}}\\
     &\times\ln{\paren{\sqrt{\frac{1}{\pi\paren{\sigx^2}^n\sigxi^2}}\frac{\Gamma\paren{\frac{n+1}{2}}}{\Gamma\paren{\frac{n}{2}}}\paren{\frac{\sigx^2\sigxi^2}{\sigx^2z^2 + \sigxi^2}}^{\frac{n + 1}{2}}}}dz\\
     &= -\int\limits_{-\infty}^\infty \paren{\sqrt{\frac{1}{\pi\paren{\sigx^2}^n\sigxi^2}}\frac{\Gamma\paren{\frac{n+1}{2}}}{\Gamma\paren{\frac{n}{2}}}\paren{\frac{\sigx^2\sigxi^2}{\sigx^2z^2 + \sigxi^2}}^{\frac{n + 1}{2}}}\nonumber\\
     &\bracket{ \ln{\paren{\sqrt{\frac{1}{\pi\paren{\sigx^2}^n\sigxi^2}}\frac{\Gamma\paren{\frac{n+1}{2}}}{\Gamma\paren{\frac{n}{2}}}\paren{\sigxi^2}^{\frac{n + 1}{2}}}}- \frac{n + 1}{2}\ln\bracket{z^2 + \frac{\sigxi^2}{\sigx^2}}}dz\\
     &= -\sqrt{\frac{1}{\pi\paren{\sigx^2}^n\sigxi^2}}\frac{\Gamma\paren{\frac{n+1}{2}}}{\Gamma\paren{\frac{n}{2}}} \ln{\paren{\sqrt{\frac{1}{\pi\paren{\sigx^2}^n\sigxi^2}}\frac{\Gamma\paren{\frac{n+1}{2}}}{\Gamma\paren{\frac{n}{2}}}\paren{\sigxi^2}^{\frac{n + 1}{2}}}}\int\limits_{-\infty}^\infty \paren{\frac{\sigx^2\sigxi^2}{\sigx^2z^2 + \sigxi^2}}^{\frac{n + 1}{2}} dz\\
     &+ \frac{n + 1}{2}\sqrt{\frac{1}{\pi\paren{\sigx^2}^n\sigxi^2}}\frac{\Gamma\paren{\frac{n+1}{2}}}{\Gamma\paren{\frac{n}{2}}}\int\limits_{-\infty}^\infty \paren{\frac{\sigx^2\sigxi^2}{\sigx^2z^2 + \sigxi^2}}^{\frac{n + 1}{2}}\ln\bracket{z^2 + \frac{\sigxi^2}{\sigx^2}}dz\\
     &= - \ln{\paren{\sqrt{\frac{1}{\pi\paren{\sigx^2}^n\sigxi^2}}\frac{\Gamma\paren{\frac{n+1}{2}}}{\Gamma\paren{\frac{n}{2}}}\paren{\sigxi^2}^{\frac{n + 1}{2}}}}\nonumber\\
     &+ \frac{n + 1}{2}\sqrt{\frac{1}{\pi\paren{\sigx^2}^n\sigxi^2}}\frac{\Gamma\paren{\frac{n+1}{2}}}{\Gamma\paren{\frac{n}{2}}}\int\limits_{-\infty}^\infty \paren{\frac{\sigx^2\sigxi^2}{\sigx^2z^2 + \sigxi^2}}^{\frac{n + 1}{2}}\ln\bracket{z^2 + \frac{\sigxi^2}{\sigx^2}}dz.\label{eq:normalized}
\end{align}

The first term in Eq.~\eqref{eq:normalized} is due to $f_Z(z)$ being a probability density function.

Just solving for $\int\limits_{-\infty}^\infty \paren{\frac{\sigx^2\sigxi^2}{\sigx^2z^2 + \sigxi^2}}^{\frac{n + 1}{2}}\ln\bracket{z^2 + \frac{\sigxi^2}{\sigx^2}}dz$, let $\psi(x) \triangleq \frac{d}{dx}\ln{\Gamma(x)} = \frac{\Gamma^\prime(x)}{\Gamma(x)}$. Using the change of variables $\theta = \tan^{-1}(\frac{\sigx z}{\sigxi})$,

\begin{align}
    &\int\limits_{-\infty}^\infty \paren{\frac{\sigx^2\sigxi^2}{\sigx^2z^2 + \sigxi^2}}^{\frac{n + 1}{2}}\ln\bracket{z^2 + \frac{\sigxi^2}{\sigx^2}}dz = \sigx^n\sigxi\int\limits_{-\frac{\pi}{2}}^\frac{\pi}{2} \paren{\frac{1}{\sec^2{\theta}}}^{\frac{n - 1}{2}}\ln\bracket{\frac{\sigxi^2}{\sigx^2}\sec^2\theta}d\theta\\
    &=  \sigx^n\sigxi\bracket{\int\limits_{-\frac{\pi}{2}}^\frac{\pi}{2} \paren{\frac{1}{\sec^2{\theta}}}^{\frac{n-1}{2}}\ln\bracket{\sec^2\theta}d\theta + \ln\bracket{\frac{\sigxi^2}{\sigx^2}}\int\limits_{-\frac{\pi}{2}}^\frac{\pi}{2} \paren{\frac{1}{\sec^2{\theta}}}^{\frac{n-1}{2}}d\theta}\\
    &=  \sigx^n\sigxi\bracket{\frac{\sqrt{\pi}\Gamma\paren{\frac{n}{2}}\paren{\psi\paren{\frac{n+1}{2}} - \psi\paren{\frac{n}{2}}}}{\Gamma\paren{\frac{n+1}{2}}}  + \frac{\ln\bracket{\frac{\sigxi^2}{\sigx^2}}\sqrt{\pi}\Gamma\paren{\frac{n}{2}}  }{\Gamma\paren{\frac{n+1}{2}}}}\\
    &= \frac{ \sigx^n\sigxi\sqrt{\pi}\Gamma\paren{\frac{n}{2}}}{\Gamma\paren{\frac{n+1}{2}}}\bracket{\psi\paren{\frac{n+1}{2}} - \psi\paren{\frac{n}{2}} + \ln\bracket{\frac{\sigxi^2}{\sigx^2}}}.
\end{align}

Now returning to the computation of $h(Z)$, we have
\begin{align}
    \frac{1}{\log\bracket{e}}h(Z) &= - \ln{\paren{\sqrt{\frac{1}{\pi\paren{\sigx^2}^n\sigxi^2}}\frac{\Gamma\paren{\frac{n+1}{2}}}{\Gamma\paren{\frac{n}{2}}}\paren{\sigxi^2}^{\frac{n + 1}{2}}}}\nonumber\\ 
    &+ \frac{n + 1}{2}\bracket{\psi\paren{\frac{n+1}{2}} - \psi\paren{\frac{n}{2}} + \ln\bracket{\frac{\sigxi^2}{\sigx^2}}}\\
    &= - \ln{\paren{\frac{\sigxi^n}{\sqrt{\pi}\sigx^n}\frac{\Gamma\paren{\frac{n+1}{2}}}{\Gamma\paren{\frac{n}{2}}}}} + \frac{n + 1}{2}\bracket{\psi\paren{\frac{n+1}{2}} - \psi\paren{\frac{n}{2}} + \ln\bracket{\frac{\sigxi^2}{\sigx^2}}}\\
    &= \ln\bracket{\frac{\sigxi}{\sigx}} - \ln{\paren{\frac{\Gamma\paren{\frac{n+1}{2}}}{\sqrt{\pi}\Gamma\paren{\frac{n}{2}}}}} + \frac{n + 1}{2}\bracket{\psi\paren{\frac{n+1}{2}} - \psi\paren{\frac{n}{2}}}.
\end{align}
\end{proof}

\begin{figure}[htbp]
  \centering
  \begin{subfigure}{0.35\textwidth}
    \centering
    \includegraphics[width=\linewidth]{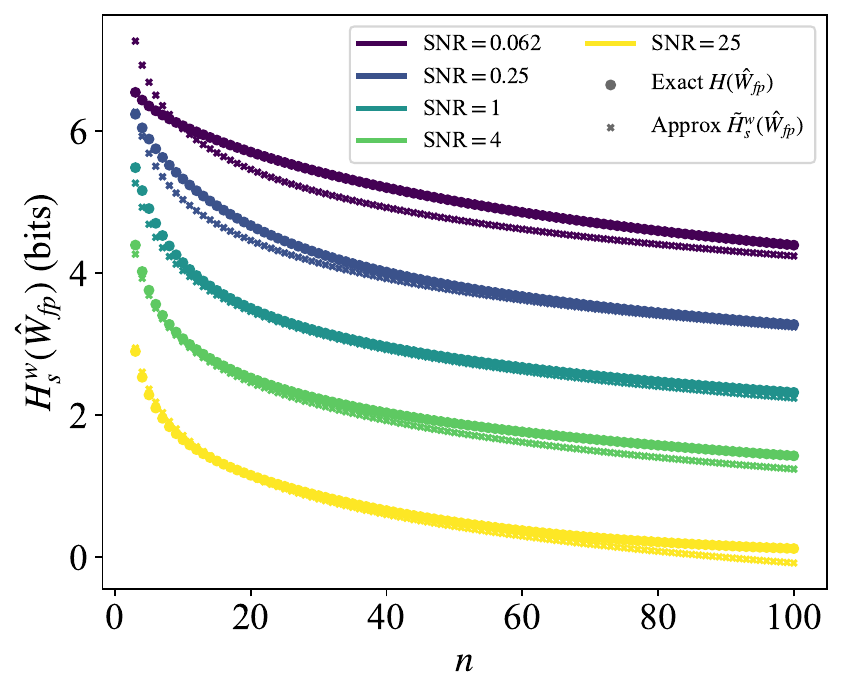}
    \caption{}
    \label{fig:H_fp_wh_exact}
  \end{subfigure}
  \begin{subfigure}{0.35\textwidth}
    \centering
    \includegraphics[width=\linewidth]{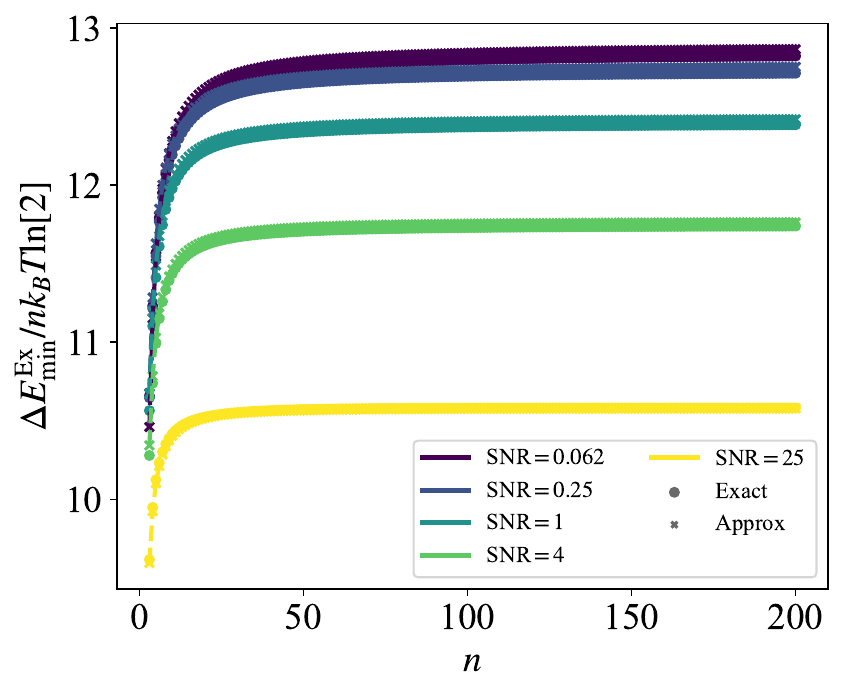}
    \caption{}
    \label{fig:Ex_dH_exact}
  \end{subfigure}\hfill
  \caption{\textit{The Landauer cost for exact zero-intercept simple linear regression.} Input and output states are floating-point numbers with $p = 4$ and $E=4$. The candidate values of the $SNR = \frac{w^2\sigx^2}{\sigxi^2}$ are $0.062$, $0.25$, $1$, $4$, and $25$. (\ref{fig:H_fp_wh_exact}) The output model entropy, and its dependence on the data size $n$ and ground truth $w$. (\ref{fig:Ex_dH_exact}) The entropy difference for various values of SNR and a range of $n$.}
  \label{fig:exact_lr_ent}
\end{figure}

\FloatBarrier
\pagebreak
\section{Mismatch Cost}
\label{app:mmc}

\subsubsection{Defining $\Delta S_{env}$ and $\Delta S_{tot}$ with a Continuous Time Markov Chain}

One general way to model such a time dependence is via a continuous time Markov chain (CTMC) based on transition rates between states of the joint computational system \cite{Wolpert2019, Manzano2024, KolchinskyWolpert2021, Owen_2019, esposito2010}. In this case, the dynamics of the computational system are fixed by a time varying matrix $W_{x_c,x_c'}(t)$, which defines the rate of transitions between states $x_c, x_c'\in \mathcal{X}_c$ at time $t$ \cite{Wolpert2019, KolchinskyWolpert2021}. 
While running an algorithm, the computational system evolves according to
\begin{equation}
    \label{eq:CTMC}
    \frac{d}{dt} p(x_c,t) = \sum_{x'_c \in \mathcal{X}_c} W_{x_c;x_c'}(t) p(x_c',t).
\end{equation}
For these dynamics, \cite{esposito2010} shows
\begin{equation}
    \Delta S_{sys} = -k_B \ln[2] \int^{t_f}_0 \sum_{x_c,x_c' \in \mathcal{X}_c} W_{x_c,x_c'}(t)p(x_c', t) \log \bigg[ \frac{p(x_c,t)}{p(x_c',t)} \bigg]dt,
\end{equation}
\begin{equation}
    \label{eq:s_env}
    \Delta S_{env} = - k_B \ln[2] \int^{t_f}_0\sum_{x_c,x_c'\in \mathcal{X}_c} W_{x_c,x_c'}(t) p(x_c', t) \log \bigg[ \frac{W_{x_c',x_c}(t)}{W_{x_c,x_c'}(t)} \bigg] dt,
\end{equation}
\begin{equation}
    \Delta S_{tot} = k_B \ln[2] \int^{t_f}_0\sum_{x_c,x_c'\in \mathcal{X}_c} W_{x_c,x_c'}(t) p(x_c', t) \log \bigg[ \frac{W_{x_c,x_c'}(t)p(x_c',t)}{W_{x_c',x_c}(t) p(x_c,t)} \bigg] dt.
\end{equation}
If the transition rates $W_{x_c,x_c'}$ satisfy $\ln[W_{x_c,x_c'}/W_{x_c',x_c}] = (E(x_c') - E(x_c))/(k_BT)$, where $E(x_c)$ is an energy function over computational states, then we can say the average heat is given by $Q = T \Delta S_{env}$ \cite[p.30]{Wolpert2019} \cite{Seifert_2012}. This implies that the average heat output, fixed by Eqs.~\eqref{eq:CTMC},~\eqref{eq:s_env} is given by $Q = T(-\Delta S_{sys} + \Delta S_{tot}) = \Delta E_{min} +T\Delta S_{tot}$.

Let $\compsimplex$ denote the probability simplex associated with the computational state space $\mathcal{X}_c$. We can think of the computation which takes us from $p_1 \in \compsimplex$ to $p_{F-1} \in \compsimplex$ as fixed by a conditional distribution $\pi(x_c|x_c')$, which determines $p_{F-1}$ from $p_1$ by
\begin{equation}
    p_{F-1}(x_c) = \sum_{x_c'\in \mathcal{X}_c} \pi(x_c|x_c') p_1(x_c').
\end{equation}
The conditional distribution $\pi(x_c|x_c')$ can be physically implemented to infinitesimally small error by the time-protocol of the rate matrix $W_{x_c,x_c'}(t)$ \cite[p. 23]{Wolpert2019}\cite{Owen_2019} (assuming $x_c$ [or $x_M$] is extended to include sufficient states, see \cite{Owen_2019}). The notation $p_{F-1} = \pi p_1$ emphasizes that $p_{F-1}$ is purely a function of the input state $p_1$ and the physical manipulations performed by the computation $W_{x_c,x_c'}(t)$, which determines $\pi(x_c|x_c')$.

\subsubsection{Variational MMC and the Island Decomposition}

The requirements to obtain the minimum entropy distribution $q_1$ that returns a unique MMC are specified in \cite{Wolpert_2020}. \cite{Wolpert_2020} points out that uniqueness of optimal distributions can only be verified for `islands' of the conditional distribution defined by the algorithm. 
We can describe an algorithm (possibly stochastic) as implementing a conditional distribution $\pi(x_c | x_c')$ that maps between computational states $x_c,x_c' \in \mathcal{X}_c$. The islands of the algorithm are specified by the following relation for $x_c, x'_c \in \mathcal{X}_c$:
\begin{equation}
    \label{eq:Islands}
    x_c \sim x_c' \Leftrightarrow \exists x_f \in \mathcal{X}_c : \pi(x_f|x_c') > 0, \pi(x_f|x_c) > 0,
\end{equation}
meaning there is a finite probability of $x_c$ and $x_c'$ transitioning to the same state $x_f$. An \textit{island} of $\pi(x_c | x_c')$ is a connected subset of $\mathcal{X}_c$ given by the transitive closure of Eq.~\eqref{eq:Islands}. The set of islands of $\pi$ partition $\mathcal{X}_c$, where the set of islands is denoted $L(\pi)$.

We can consider the distribution over an island $c \in L(\pi)$, given by $p^c(x_c) = \indic{x_c \in c}p(x_c)/p(c)$. If optimized over the probability simplex specific to the island $c$ is $\boldsymbol{\Delta}_c$, the minimizer $q^c(x)$ will be unique \cite{Wolpert_2020}. The total distribution optimized over $\compsimplex$ can be decomposed in terms of islands and the probability of their individual occupation $q_1(c)$, 
\begin{equation}
    \label{eq:island_dist}
    q_1(x_c) = \sum_{c \in L(\pi)} q_1(c) q_1^c(x_c).
\end{equation}
The factors $q(c)$ do not contribute to the KL divergences in Eq.~\eqref{eq:MMC}, meaning they do not contribute to the MMC and ensuring the MMC is unique for the optimal distribution $q_1$.

For a deterministic algorithm, $\pi$ will be valued 0 or 1 for all inputs. In this case, the islands of the algorithm will be specified by the set of inputs that correspond to a single output state. As we have considered them, both exact linear regression and SGD are deterministic algorithms. In the exact case, the final $\wh_Q$ follows deterministically from the data $\mathcal{D}_Q$, and given that for SGD $\wh_Q$ has a deterministic initialization, the final $\wh_Q$ also follows deterministically from the batches $\{\mathcal{B}_{Q,1},...,\mathcal{B}_{Q,\tau}\}$. Each island of the algorithms corresponds to all possible $\mathcal{D}_Q$ or $\{\mathcal{B}_{Q,1},...,\mathcal{B}_{Q,\tau}\}$ that have the same output $\wh_Q$. For both the exact regression and SGD, the islands of input states are complex but disjoint subsets of the input data, which correspond to specific values of $\Wh$, rendering the island decomposition difficult to consider in this case.

The variational approach breaks the island decomposition by preventing $q_1^c(x_c)$ from being individually optimized, and by preventing a description of the input state as specified by Eq.~\eqref{eq:island_dist}. With a variational $q_{1,v}$ that breaks the island decomposition, we cannot guarantee it to be unique. However, this lack of uniqueness does not prevent MMC$_v$ from providing a lower bound on the true MMC: one can perform a bounded optimization to find $q_{1,v}$, which may or may not be unique, the difference MMC$_v = \Delta S_{tot}(p_{1,v}) - \Delta S_{tot}(q_{1,v})$ will retain the property $0 \leq $ MMC$_v \leq $ MMC for variational distributions $p_{1,v} \in \mathcal{V}$. 

\subsubsection{Illustrative MMC for exact regression and SGD}

\FloatBarrier
\begin{figure}[htbp]
    \centering
    \begin{subfigure}{0.36\textwidth}
        \includegraphics[width=\linewidth]{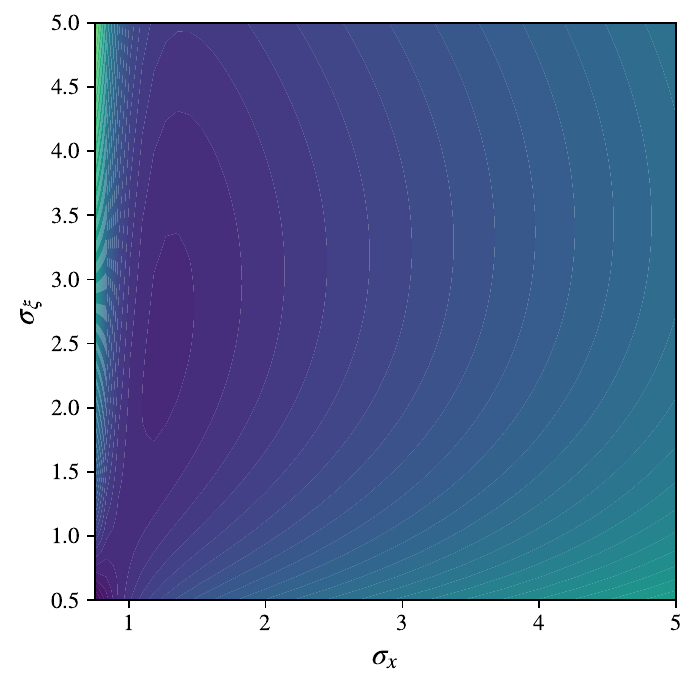}
        \caption{}
        \label{fig:MMC_Ex}
    \end{subfigure}
    \begin{subfigure}{0.42\textwidth}
        \includegraphics[width=\linewidth]{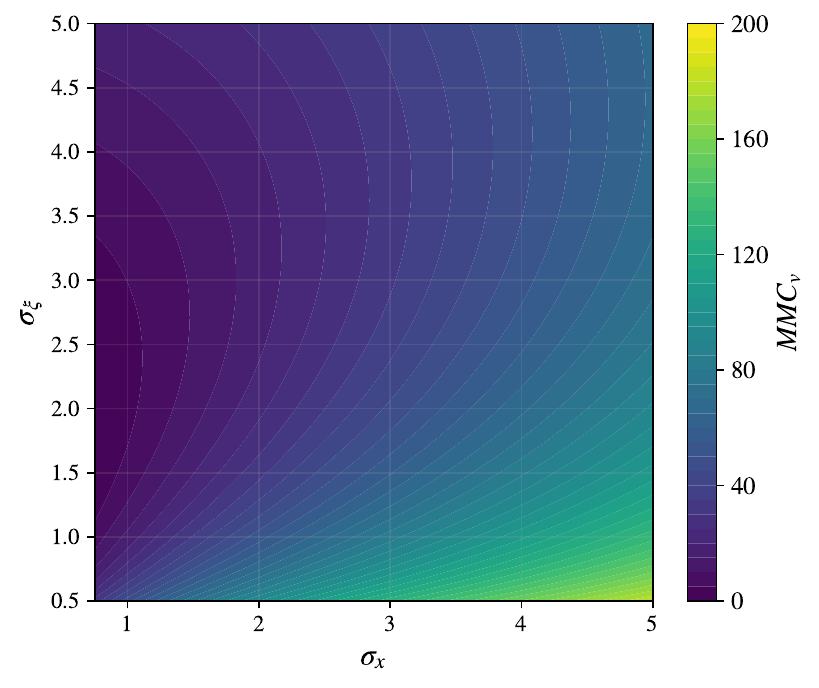}
        \caption{}
        \label{fig:MMC_SGD}
    \end{subfigure}
    \caption{\textit{A lower bound on the mismatch cost for continuously parameterized inputs}. MMC$_v(\sigma_x, \sigma_\xi)$ as a function of $\sigma_x$ and $\sigma_\xi$, with each plot assuming the entropy flow $\Delta S_{env,Ex} = \Delta S_{env, SGD} = C + \alpha(\sigma_x^2 + \sigma_\xi^2 )$, and $w = 1$. A bounded optimization is performed for $0.75 \leq \sigma_x \leq 5$, $0.5 \leq \sigma_\xi \leq 5$. \ref{fig:MMC_Ex} A sample MMC$_v$ landscape for exact linear regression with the illustrative $\Delta S_{env,Ex}=C + \alpha(\sigma_x^2 + \sigma_\xi^2 )$, and $n = 10$. \ref{fig:MMC_SGD} A sample MMC$_v$ landscape for linear regression via SGD with an illustrative $\Delta S_{env,SGD}=C + \alpha(\sigma_x^2 + \sigma_\xi^2 )$, and $B = 1$, $\tau = 50$.} 
    \label{fig:MMC}
\end{figure}

For exact regression, the variational family will be $\theta = (\sigma_x,\sigma_\xi)$, which only runs from inputs to outputs once, $\Delta S_{tot}$ will be
\begin{equation}
    \Delta S_{tot, Ex}(\sigma_x, \sigma_\xi) = -E^{Ex}_{min}(\sigma_x,\sigma_\xi)/T+\Delta S_{env,Ex}(\sigma_x,\sigma_\xi)
\end{equation}
from which $q_{1,v}(\sigma_x^\star,\sigma_\xi^\star)$, defined by the optimal $\sigma_x^\star$ and $\sigma_\xi^\star$, follows from Eq.~\eqref{eq:mmc_var_opt}. MMC$_v$ for exact regression is given by MMC$_v(\sigma_x,\sigma_\xi) = \Delta S_{tot, Ex}(\sigma_x,\sigma_\xi) - \Delta S_{tot,Ex}(\sigma_x^\star,\sigma_\xi^\star)$. Fixing $\Delta S_{env,Ex} = C+\alpha(\sigma_x^2 + \sigma_\xi^2)$, MMC$_v$ is plotted for a range of distributions in Fig.~\ref{fig:MMC_Ex}.

For the SGD device, each update cycle $k$ takes input $(\mathcal{B}_Q, \hat{W}_{Q,k})$ and produces
output $\hat{W}_{Q,k+1}$, constituting a separate physical process. The MMC should
therefore be evaluated at each step and summed,
\begin{equation}
    \mathrm{MMC}_{v,\mathrm{SGD}}(\sigma_x,\sigma_\xi) = \sum_{k=0}^{\tau-1} \Bigl(\Delta S_{tot, k}(\sigma_x,\sigma_\xi) - \Delta S_{tot,k}(\sigma_x^\star,\sigma_\xi^\star)\Bigr),
\end{equation}
where the per-step total entropy production is
\begin{equation}
    \Delta S_{\mathrm{tot},k}(\sigma_x, \sigma_\xi)
    = 
    - k_B \ln[2]\left(H(\hat{W}_{Q,k}) + B H(X_{fp}, Y_{fp})
      - H(\hat{W}_{Q,k+1})\right) + \Delta S_{env,SGD}.
\end{equation}
To simplify this expression, we can assume that the entropy flow at each step depends only
on the batch distribution, and not on the evolving model parameter distribution,
\begin{equation}
    \Delta S_{env,k}(\theta) = \Delta S_{env,Ex}(\theta)/\tau =  (C + \alpha(\sigma_x^2 + \sigma_\xi^2))/\tau \quad \forall\, k.
\end{equation}
This reduces the per-step sum to a telescoping series as in the preliminaries, giving the total entropy production
\begin{equation}
    \Delta S_{tot,\mathrm{SGD}}(\sigma_x,\sigma_\xi) = -E^{SGD}_{min}(\sigma_x,\sigma_\xi)/T + \Delta S_{env,SGD}(\sigma_x,\sigma_\xi),
\end{equation}
and thus $\mathrm{MMC}_{v,\mathrm{SGD}}(\sigma_x,\sigma_\xi) = \Delta S_{tot,\mathrm{SGD}}(\sigma_x,\sigma_\xi) - \Delta S_{tot,\mathrm{SGD}}(\sigma_x^\star,\sigma_\xi^\star)$.
Fixing $\Delta S_{env}$ beyond its reversible bound cannot be done at an algorithmic level. While we leave evaluation of the true entropy flow $\Delta S_{env}$ to future work, we can demonstrate an example of this approach using an illustrative choice of the entropy flow: $\Delta S_{env}(\sigma_x,\sigma_\xi) = \Delta S_{env,Ex}(\sigma_x,\sigma_\xi) = \Delta S_{env,SGD}(\sigma_x,\sigma_\xi) = C + \alpha(\sigma_x^2 + \sigma_\xi^2)$, as shown in Fig.~\ref{fig:MMC}.

\pagebreak
\FloatBarrier
\section{Derivations for Energy-cost Aware Scaling Laws}\label{sec:scaling_laws_app}

\subsubsection{Exact Formula}
Using the generalization error scaling for the exact linear regression formula in Eq.~\eqref{eq:exact_error}, we can solve a continuous relaxation of the optimization problem in Eq.~\eqref{eq:scaling_opt} where $\tilde n \in \R{}$. Let \begin{equation}\label{eq:u_ex}
u_{Ex}(\tilde n) \triangleq \rho_I\frac{\tilde n-2}{(\tilde n-1)\sigxi^2} - \rho_{J}\Delta E_{min}^{Ex}\paren{\tilde n}.
\end{equation}
 The optimization problem now becomes
\begin{align}
    \tilde n^\star &= \arg\max\limits_{\tilde n \geq 3} u_{Ex}(\tilde n) \text{, and } \nonumber\\
    n^\star &= \arg\max_{n\in\{\lfloor \tilde n^\star\rfloor,\lceil \tilde n^\star\rceil\}} u_{Ex}(n)\label{eq:nstar_ex}.
\end{align}
For simplicity, assume $|w| \gg \frac{\sigxi}{\sigx}$. Although the following method could be derived using the exact entropy calculations in Sections~\ref{sec:fpn}~and~\ref{sec:exact}, we will use $\tilde H_s(\Xfp, \Yfp)$ from Eq.~\eqref{eq:approx_joint_ent_SNR} and $\tilde H^w_s(\Wh_{fp}) = h(Z) + (p-1) -\log\bracket{\frac{|w|}{\sqrt{2}}}$ in this section for tractability. $h(Z)$ is computed in Lemma~\ref{lem:diff_ent_Ex}. Let $\psi^\prime(x) \triangleq \frac{d}{dx}\psi(x)$. Taking the derivative of $u$ with respect to $\tilde n$, we have
\begin{align}
    u^\prime_{Ex} (\tilde n) &= \frac{\rho_I}{\sigxi^2}\paren{\frac{1}{\tilde n-1}}^2\nonumber - \rho_J k_BT \ln[2] \paren{2\Ht^0_{s}(p) - \frac{1}{2}\log\bracket{ 1 + \frac{\sigma^2_x w^2}{\sigma_\xi^2}}- \frac{d}{d\tilde n}h(Z)}\label{eq:ddn_profit_ex},
\end{align}
where
\begin{align}
    \frac{d}{d\tilde n}h(Z) &= \log\bracket{e}\paren{\frac{\tilde n+1}{4} }\left(\psi^\prime\left(\frac{\tilde n+1}{2}\right)-\psi^\prime\left(\frac{\tilde n}{2}\right)\right).
\end{align}

It is not guaranteed that Eq.~\eqref{eq:ddn_profit_ex} will have a solution for all settings of parameters. It is recommended to solve the equation numerically using a root finder.\footnote{In our simulations, the profit function appears to be unimodal in $n$ which would justify the rounding procedure given in Eq.~\eqref{eq:nstar_ex}. We leave the proof of this unimodality for future work.} In Fig.~\ref{fig:ddn_profit_ex} we can see where $u^\prime_{Ex} (\tilde n)$ crosses zero for various setting of the price of energy $\rho_J$ and these zero crossings correspond to maximum values of the profit $u_{Ex}(n)$.

\subsubsection{Stochastic Gradient Descent}
\begin{figure}[bthp]
  \centering
  \begin{subfigure}{0.49\textwidth}
    \centering
    \includegraphics[width=\linewidth]{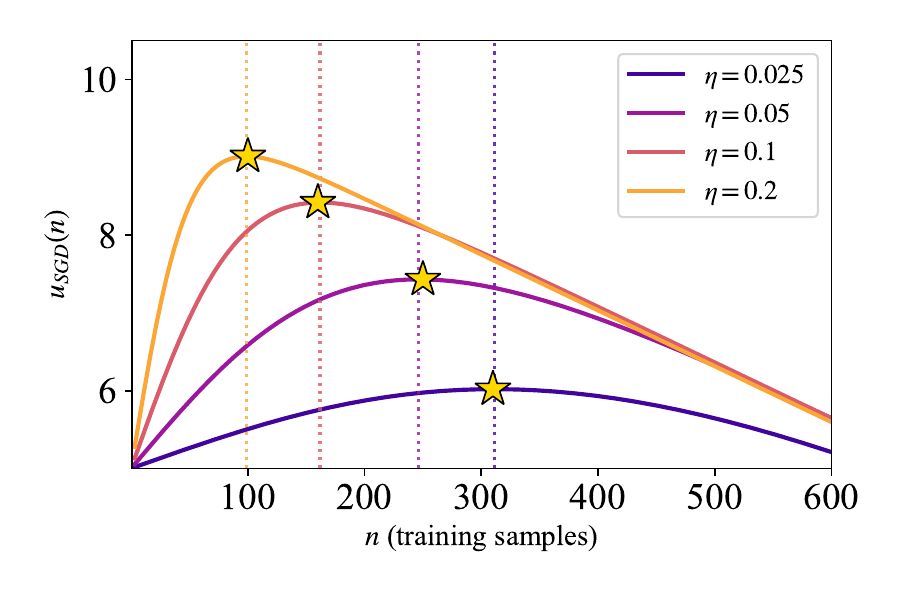}
    \caption{}
    \label{fig:profit_sgd_eta}
  \end{subfigure}
  \begin{subfigure}{0.49\textwidth}
    \centering
    \includegraphics[width=\linewidth]{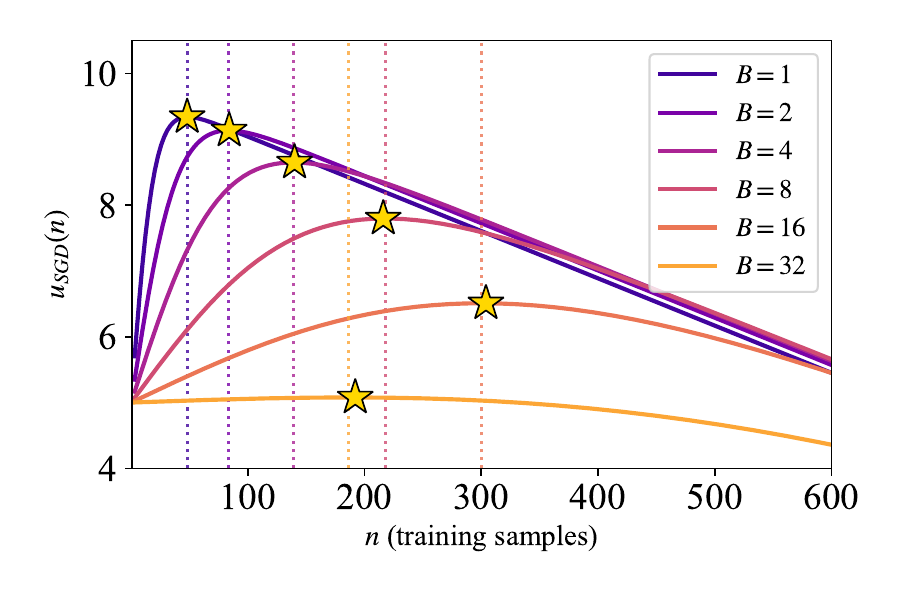}
    \caption{}
    \label{fig:profit_sgd_B}
  \end{subfigure}
  \begin{subfigure}{0.49\textwidth}
    \centering
    \hspace*{-0.35cm} 
    \includegraphics[width=\linewidth]{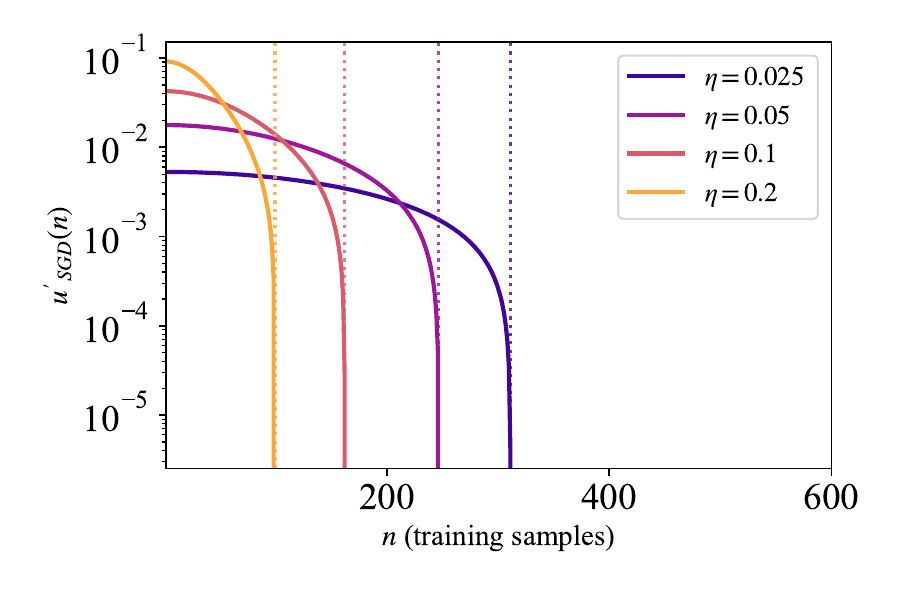}
    \caption{}
    \label{fig:ddn_profit_sgd_eta}
  \end{subfigure}
  \begin{subfigure}{0.49\textwidth}
    \centering
    \hspace*{-0.35cm} 
    \includegraphics[width=\linewidth]{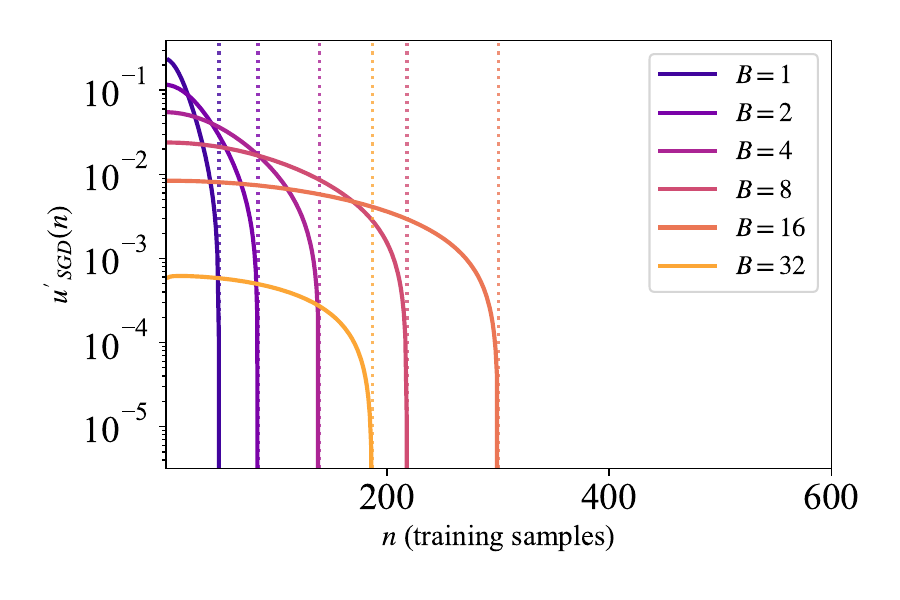}
    \caption{}
    \label{fig:ddn_profit_sgd_B}
  \end{subfigure}
  \caption{\textit{The effect of the learning rate and batch size on the optimal dataset size for stochastic gradient descent.} (\ref{fig:profit_sgd_eta}) shows $u_{SGD}(n)$ given in Eq.~\eqref{eq:u_sgd} for varying values of the learning rate $\eta$. Notice that smaller learning rates lead to larger optimal dataset sizes. (\ref{fig:ddn_profit_sgd_eta}) shows $u'_{SGD}(n)$ in Eq.~\eqref{eq:ddn_profit_sgd} for the different learning rates. (\ref{fig:profit_sgd_B}) shows $u_{SGD}(n)$ for various values of the batch size $B$. For the profit plots, the gold stars are the maximum values of the profit with respect to each algorithm's feasible set for each value of $\eta$ or $B$. (\ref{fig:ddn_profit_sgd_B}) shows $u'_{SGD}(n)$ for the different batch sizes. The vertical dotted lines show the point where each curve crosses zero. For all figures, $\sigxi^2 =\sigx^2 = 1$, $k_BT = 4\times 10^{-21}$, $w = 2$, $\wh_0 = 1$, $\rho_I = 10$, and $\rho_J = 5\times 10^{16}$. For the plots that vary $\eta$, $B=10$ and for the plots that vary $B$, $\eta = 0.05$.}
  \label{fig:scaling_laws_sgd_eta_B}
\end{figure}
For stochastic gradient descent, we can again take a continuous relaxation of the optimization problem in Eq.~\eqref{eq:scaling_opt}. Let 
\begin{equation}\label{eq:u_sgd}
u_{SGD}(\tilde n) \triangleq \rho_I\paren{\sigx^2\bracket{\frac{\eta\sigxi^2}{2B}\paren{1 -e^{-\frac{2\sigx^2 \eta \tilde n}{B}}} + \paren{\tilde \mu\paren{\frac{\tilde n}{B}} - w}^2} + \sigxi^2}^{-1} - \rho_J \Delta E_{min}^{SGD}\paren{\frac{\tilde n}{B}},
\end{equation} 
to obtain the optimization problem
\begin{align}
    \tilde n^\star &= \arg\max\limits_{\tilde n \geq B} u_{SGD}(\tilde n) \text{, and } \nonumber\\
    n^\star &= \arg\max_{n\in\{B\lfloor \tilde n^\star/B\rfloor,\, B\lceil \tilde n^\star/B\rceil\}} u_{SGD}(n)\label{eq:nstar_sgd},
\end{align}
where the rounding ensures $n^\star \in B\N{}$, since the total number of SGD training samples must satisfy $n = kB$ for integer $k$.
Under the assumptions that the learning rate $\eta$ is small enough that $|\tilde \mu(k)| \gg \sqrt{\frac{\eta\sigxi^2}{2B}}$ throughout the optimization domain, and that $k$ is large enough for the OU approximation Eq.~\eqref{eq:langevin} to track the true SGD dynamics (see Section~\ref{sec:nonstationarysgd}), we can use the approximate entropy $\Ht^{\tilde \mu(k)}_s(\Wh_{fp,k})$ given in Eq.~\eqref{eq:three_approxes}. Taking the derivative of $u_{SGD}$ with respect to $\tilde n$, we have
\begin{align}
    &u^\prime_{SGD} (\tilde n) = -\rho_I\frac{\eta  \sigx^4 e^{-\frac{2 \eta  \tilde n \sigx^2}{B}} \left(\eta  \sigxi ^2-2 B (w-\wh_0)^2\right)}{B^2\paren{\sigx^2\bracket{\frac{\eta\sigxi^2}{2B}\paren{1 -e^{-\frac{2\sigx^2 \eta \tilde n}{B}}} + \paren{\tilde \mu\paren{\frac{\tilde n}{B}} - w}^2} + \sigxi^2}^{2}} \nonumber\\
    &- \rho_J k_BT \ln[2] \Bigg(2\Ht^0_{s}(p) - \frac{1}{2}\log\bracket{ 1 + \frac{\sigma^2_x w^2}{\sigma_\xi^2}}+ \frac{\eta  \sigx^2  \left(w \left(1-e^{-\frac{\eta  \tilde n \sigx^2}{B}}\right)-\wh_0 \right)}{B \ln\bracket{2} \left(1 - e^{-\frac{2 \eta  \tilde n \sigx^2}{B}}\right) \left(w \left(e^{\frac{\eta  \tilde n \sigx^2}{B}}-1\right)+\wh_0\right)}\Bigg).\label{eq:ddn_profit_sgd}
\end{align}

Similar to the exact formula case, in Fig.~\ref{fig:ddn_profit_sgd} we plot $u^\prime_{SGD} (n)$ for various $\rho_J$ and show that its zero crossings correspond to maxima in  $u_{SGD} (n)$ in Fig.~\ref{fig:profit_sgd}. The optimal dataset size decreases as $\rho_J$ increases.

\FloatBarrier
\section{Exact Floating-point Entropy vs Approximated Floating-point Entropy Figures}\label{app:exact_FPN_extrafigs}

\begin{figure}[htbp]
  \centering
  \begin{subfigure}{0.25\textwidth}
    \centering
    \includegraphics[width=\linewidth]{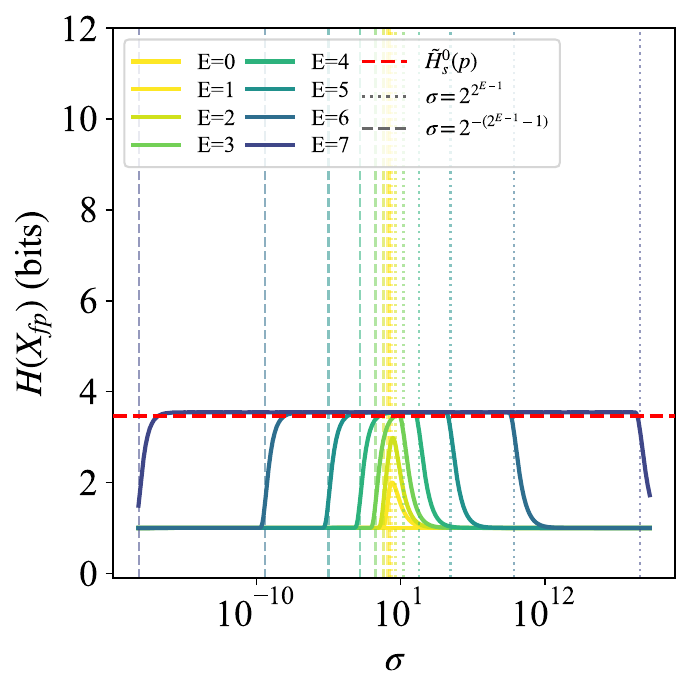}
    \caption{$p=1$.}
    \label{fig:app_sigma_p1}
  \end{subfigure}
  \begin{subfigure}{0.25\textwidth}
    \centering
    \includegraphics[width=\linewidth]{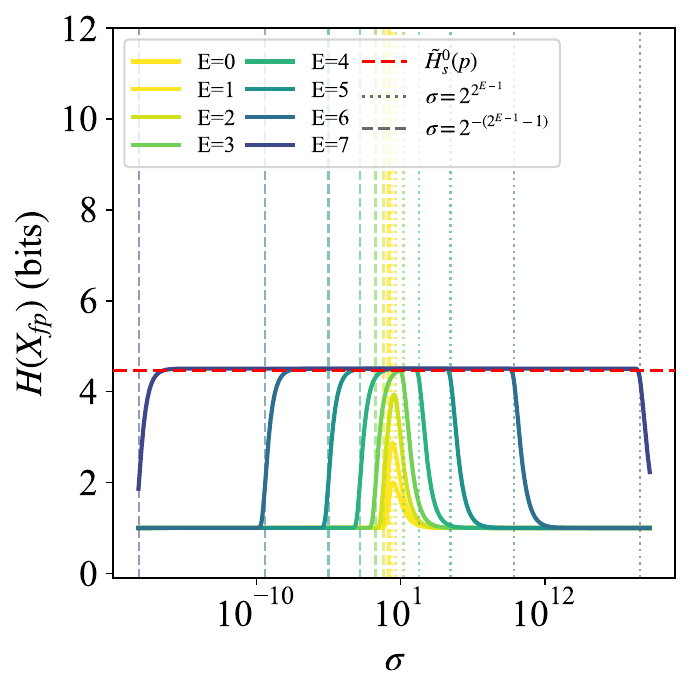}
    \caption{$p=2$.}
    \label{fig:app_sigma_p2}
  \end{subfigure}
  \begin{subfigure}{0.25\textwidth}
    \centering
    \includegraphics[width=\linewidth]{code/quantized_entropy_midpoint/sigma_p_E/entropy_sigma_midpoint_p3.pdf}
    \caption{$p=3$.}
    \label{fig:app_sigma_p3}
  \end{subfigure}
  \begin{subfigure}{0.25\textwidth}
    \centering
    \includegraphics[width=\linewidth]{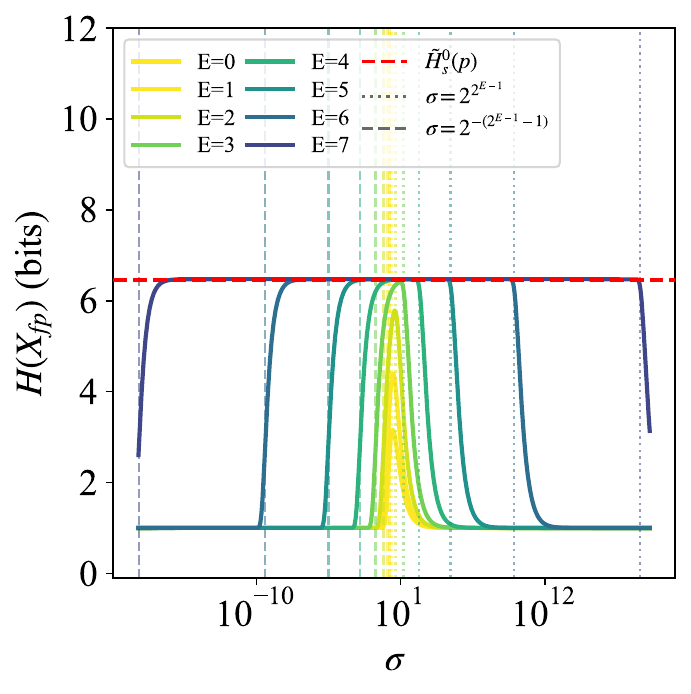}
    \caption{$p=4$.}
    \label{fig:app_sigma_p4}
  \end{subfigure}
  \begin{subfigure}{0.25\textwidth}
    \centering
    \includegraphics[width=\linewidth]{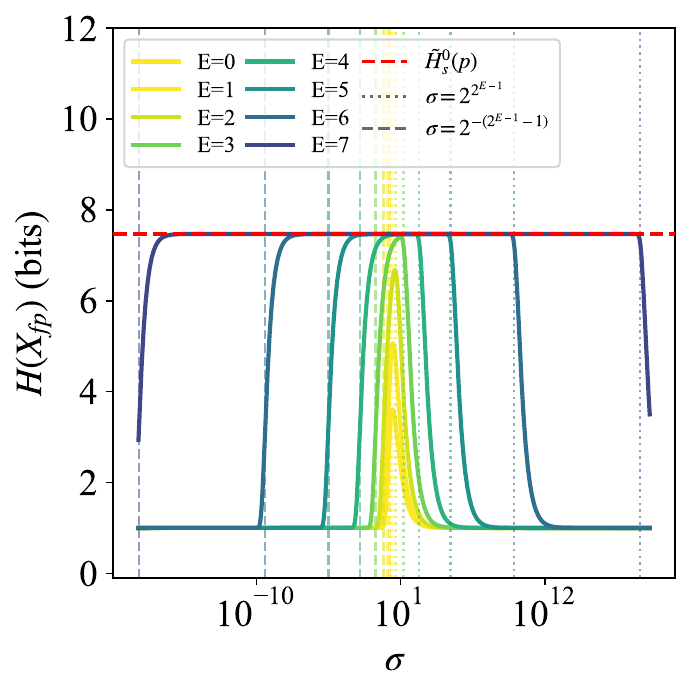}
    \caption{$p=5$.}
    \label{fig:app_sigma_p5}
  \end{subfigure}
  \begin{subfigure}{0.25\textwidth}
    \centering
    \includegraphics[width=\linewidth]{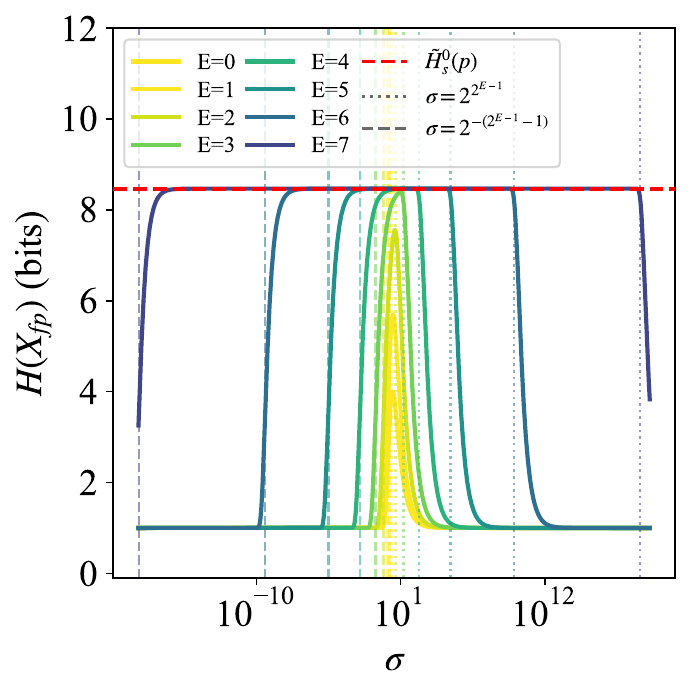}
    \caption{$p=6$.}
    \label{fig:app_sigma_p6}
  \end{subfigure}
  \begin{subfigure}{0.25\textwidth}
    \centering
    \includegraphics[width=\linewidth]{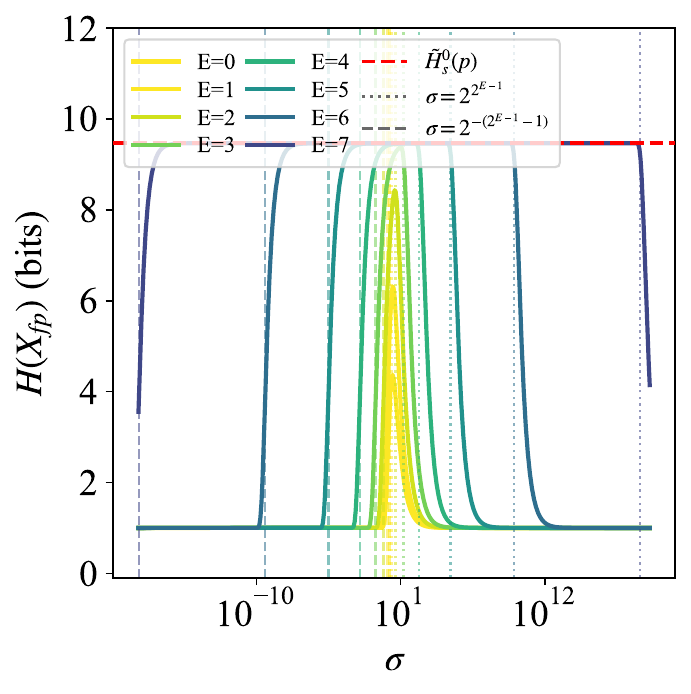}
    \caption{$p=7$.}
    \label{fig:app_sigma_p7}
  \end{subfigure}
  \begin{subfigure}{0.25\textwidth}
    \centering
    \includegraphics[width=\linewidth]{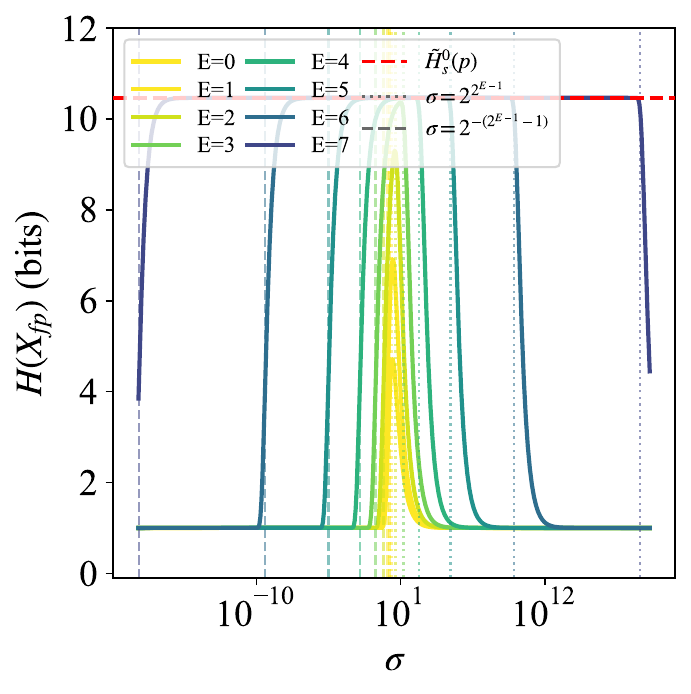}
    \caption{$p=8$.}
    \label{fig:app_sigma_p8}
  \end{subfigure}
  \caption{\textit{Exact midpoint-quantized entropy vs. standard deviation $\sigma$.} For each precision $p \in \{1,\ldots,8\}$, the exact discrete entropy $H(X_{fp})$ of $X \sim \gauss{0}{\sigma^2}$ (with $\mu=0$) is plotted as a function of $\sigma$ over a wide log-scale range. Each curve corresponds to a distinct value of exponent bits $E \in \{0,1,\ldots,7\}$. The vertical dashed lines mark $\sigma = 2^{e_{\min}}$ and the vertical dotted lines mark $\sigma = 2^{e_{\max}}$ for each $E$, and the horizontal red line shows the zero-mean approximation $\tilde{H}_s^0(p)$. These plots were generated by sweeping $\sigma$ over 500 log-spaced points and computing the exact entropy via Corollary~\ref{cor:gauss_fp_ent}.}
  \label{fig:app_entropy_vs_sigma_all}
\end{figure}

\begin{figure}[p]
  \centering
  \begin{subfigure}{0.32\textwidth}
    \centering
    \includegraphics[width=\linewidth]{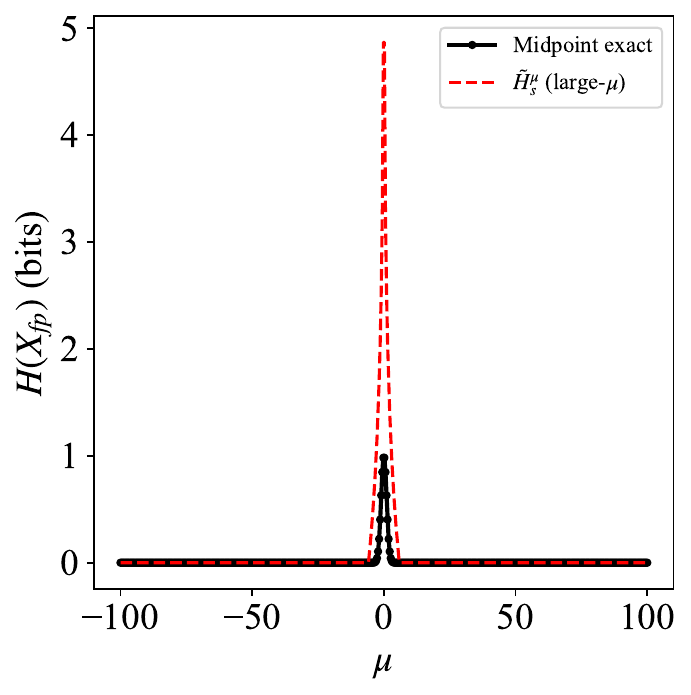}
    \caption{$p=1$, $E=0$.}
    \label{fig:app_mu_p1_E0}
  \end{subfigure}
  \begin{subfigure}{0.32\textwidth}
    \centering
    \includegraphics[width=\linewidth]{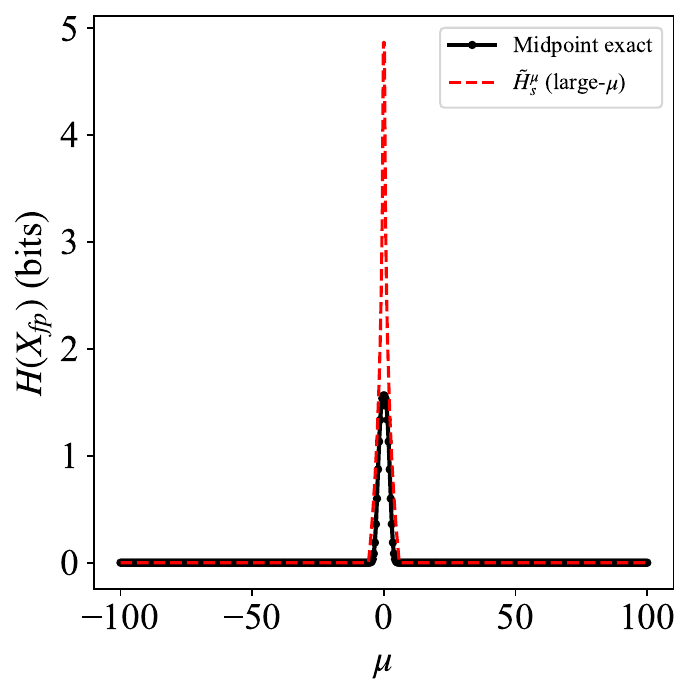}
    \caption{$p=1$, $E=1$.}
    \label{fig:app_mu_p1_E1}
  \end{subfigure}
  \begin{subfigure}{0.32\textwidth}
    \centering
    \includegraphics[width=\linewidth]{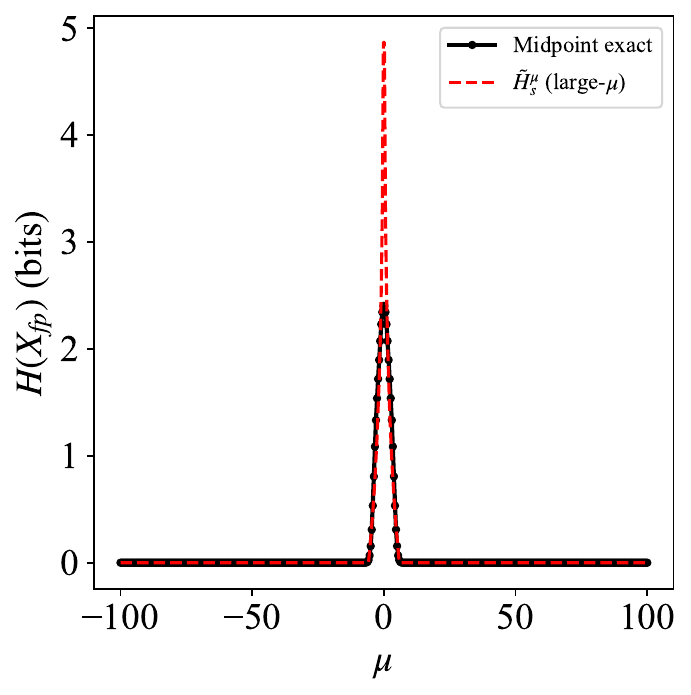}
    \caption{$p=1$, $E=2$.}
    \label{fig:app_mu_p1_E2}
  \end{subfigure}
  \begin{subfigure}{0.32\textwidth}
    \centering
    \includegraphics[width=\linewidth]{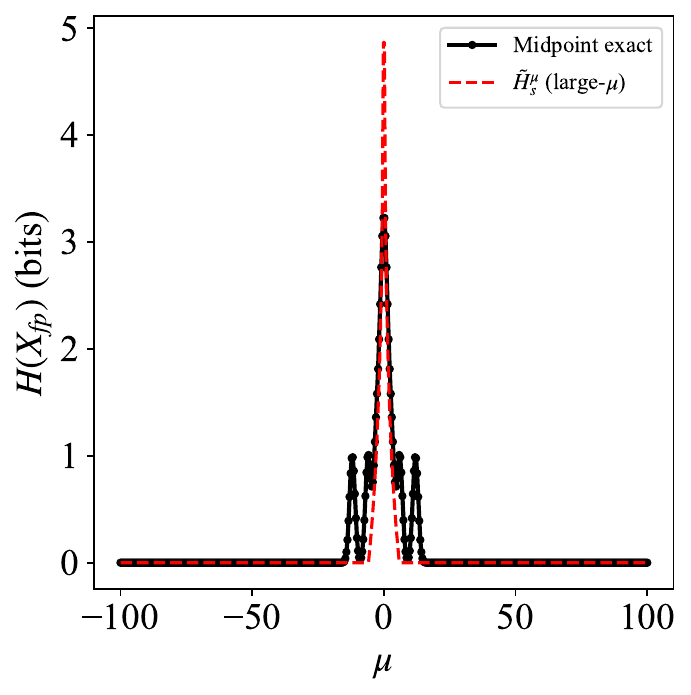}
    \caption{$p=1$, $E=3$.}
    \label{fig:app_mu_p1_E3}
  \end{subfigure}
  \begin{subfigure}{0.32\textwidth}
    \centering
    \includegraphics[width=\linewidth]{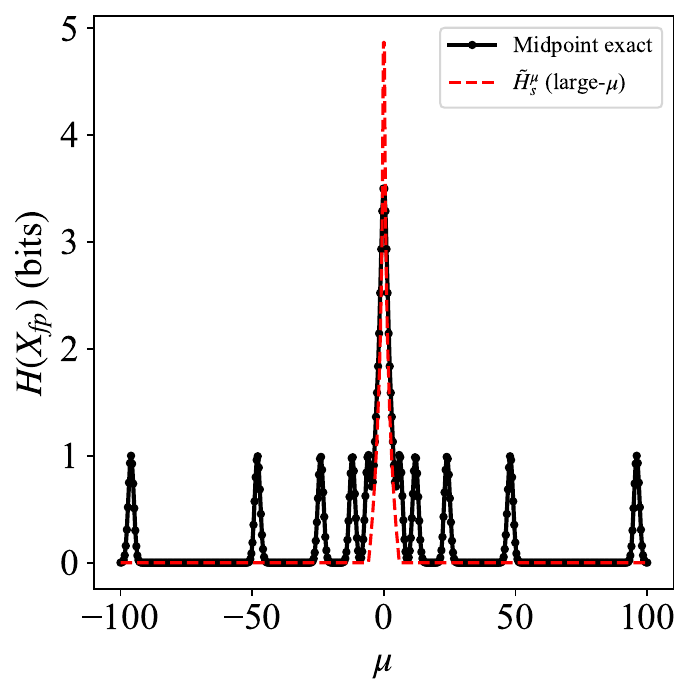}
    \caption{$p=1$, $E=4$.}
    \label{fig:app_mu_p1_E4}
  \end{subfigure}
  \begin{subfigure}{0.32\textwidth}
    \centering
    \includegraphics[width=\linewidth]{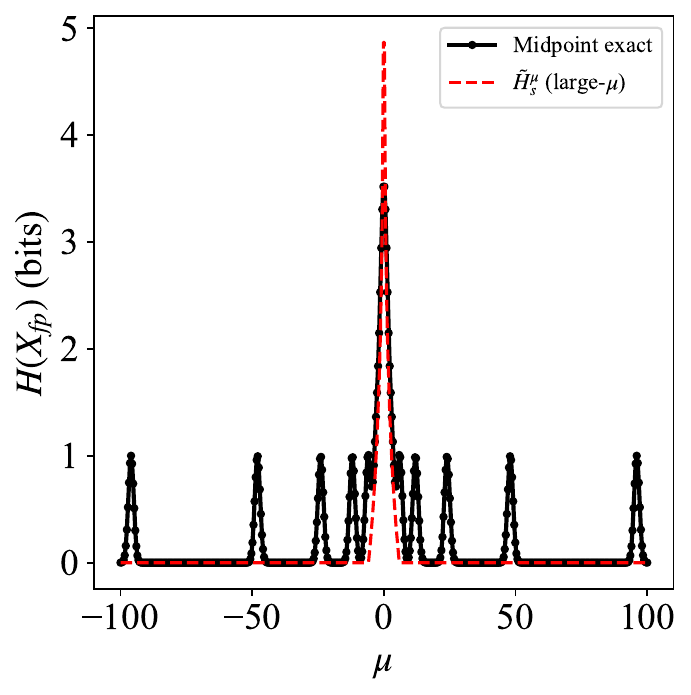}
    \caption{$p=1$, $E=5$.}
    \label{fig:app_mu_p1_E5}
  \end{subfigure}
  \begin{subfigure}{0.32\textwidth}
    \centering
    \includegraphics[width=\linewidth]{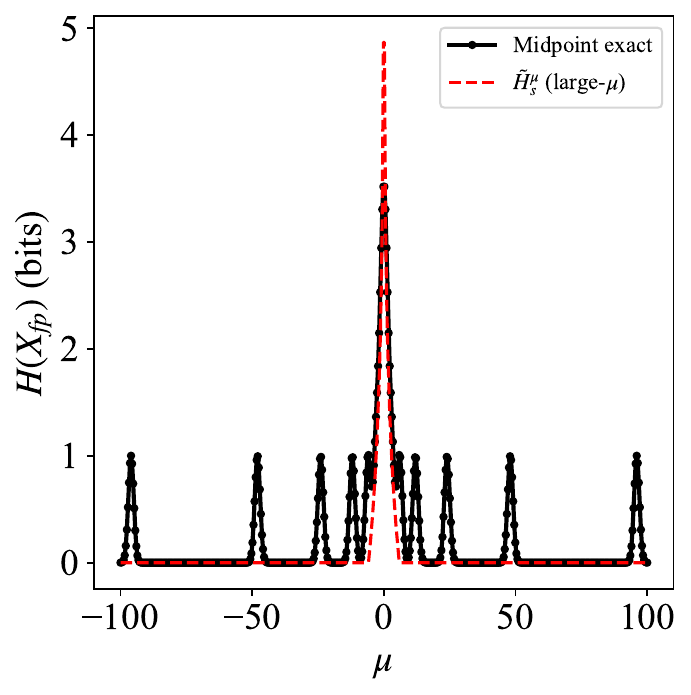}
    \caption{$p=1$, $E=6$.}
    \label{fig:app_mu_p1_E6}
  \end{subfigure}
  \begin{subfigure}{0.32\textwidth}
    \centering
    \includegraphics[width=\linewidth]{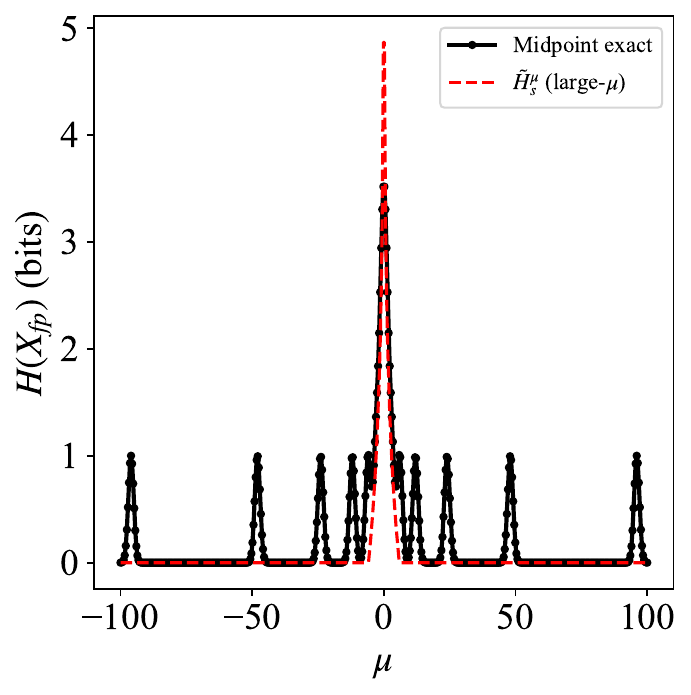}
    \caption{$p=1$, $E=7$.}
    \label{fig:app_mu_p1_E7}
  \end{subfigure}
  \caption{\textit{Exact midpoint-quantized entropy vs.\ mean $\mu$, $p=1$.} The exact entropy $H(X_{fp})$ of $X \sim \gauss{\mu}{1}$ (with $\sigma=1.0$ fixed and $p=1$) is plotted as a function of $\mu \in [-100, 100]$. Each panel shows a different number of exponent bits $E$. The solid curve is the exact entropy and the dashed curve is the large-$|\mu|$ approximation $\tilde{H}_s^\mu(\Xfp)$. These plots were generated by sweeping $\mu$ over 500 linearly-spaced points and evaluating the exact entropy via Corollary~\ref{cor:gauss_fp_ent}.}
  \label{fig:app_entropy_vs_mu_p1}
\end{figure}

\begin{figure}[p]
  \centering
  \begin{subfigure}{0.32\textwidth}
    \centering
    \includegraphics[width=\linewidth]{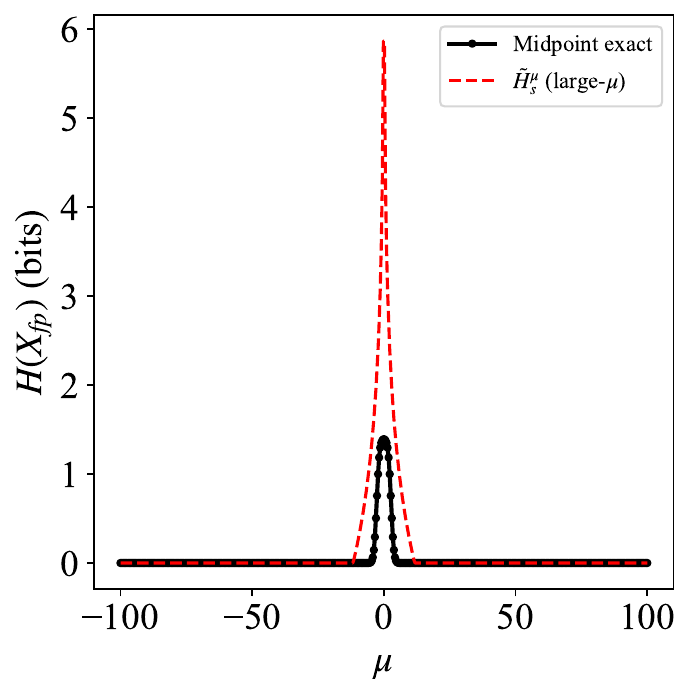}
    \caption{$p=2$, $E=0$.}
    \label{fig:app_mu_p2_E0}
  \end{subfigure}
  \begin{subfigure}{0.32\textwidth}
    \centering
    \includegraphics[width=\linewidth]{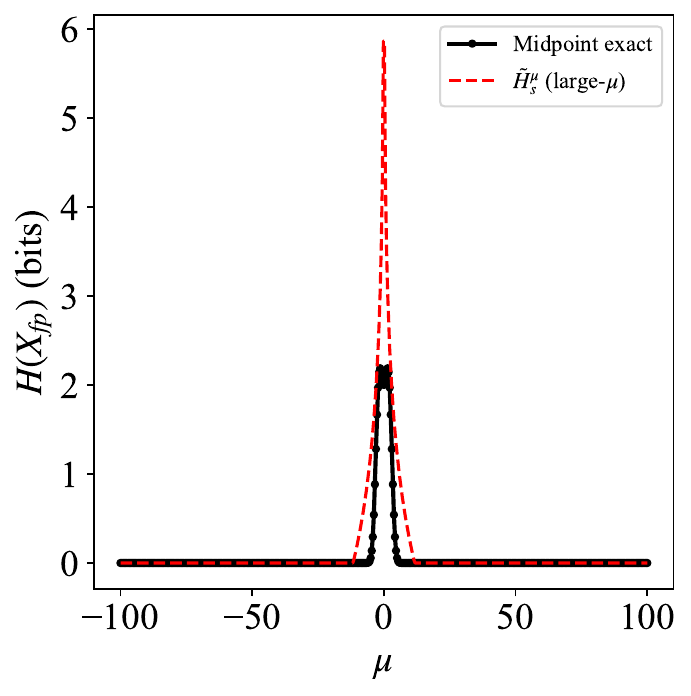}
    \caption{$p=2$, $E=1$.}
    \label{fig:app_mu_p2_E1}
  \end{subfigure}
  \begin{subfigure}{0.32\textwidth}
    \centering
    \includegraphics[width=\linewidth]{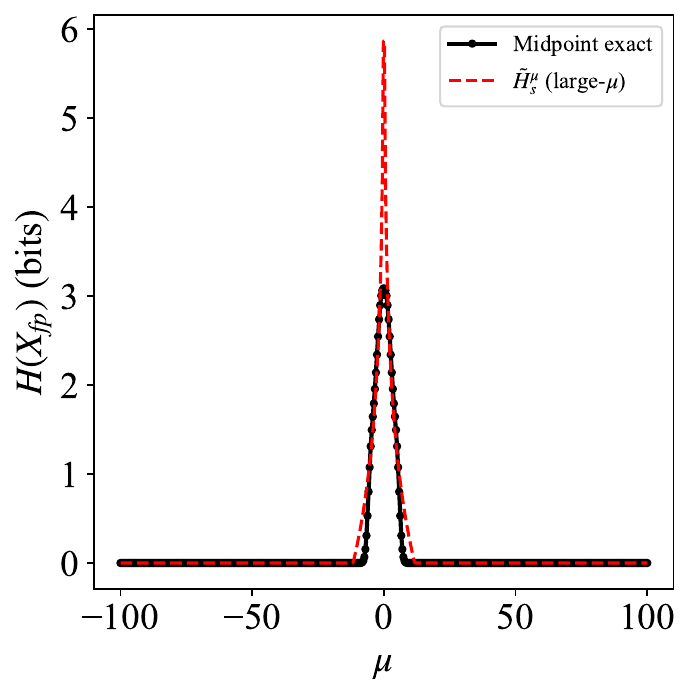}
    \caption{$p=2$, $E=2$.}
    \label{fig:app_mu_p2_E2}
  \end{subfigure}
  \begin{subfigure}{0.32\textwidth}
    \centering
    \includegraphics[width=\linewidth]{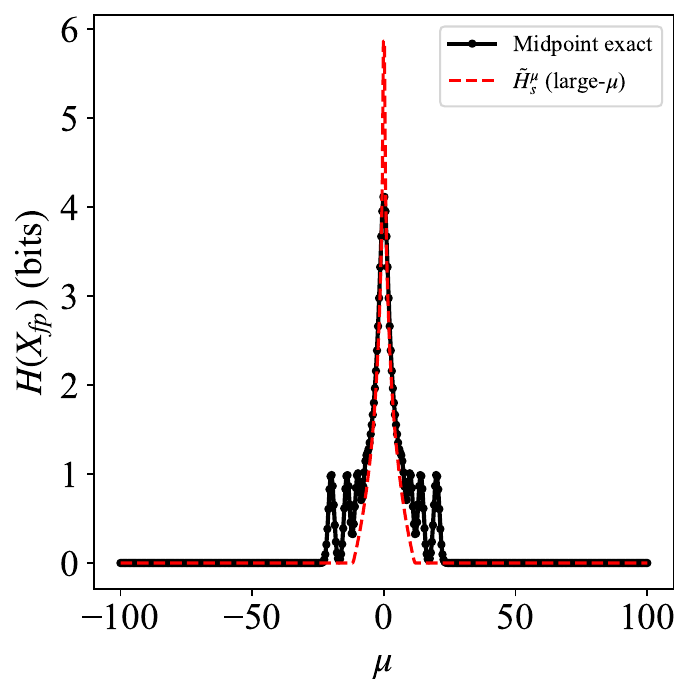}
    \caption{$p=2$, $E=3$.}
    \label{fig:app_mu_p2_E3}
  \end{subfigure}
  \begin{subfigure}{0.32\textwidth}
    \centering
    \includegraphics[width=\linewidth]{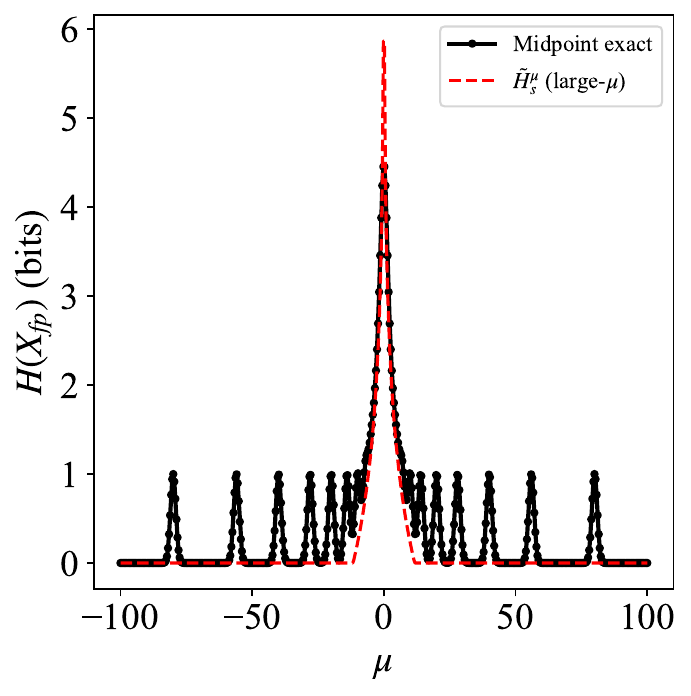}
    \caption{$p=2$, $E=4$.}
    \label{fig:app_mu_p2_E4}
  \end{subfigure}
  \begin{subfigure}{0.32\textwidth}
    \centering
    \includegraphics[width=\linewidth]{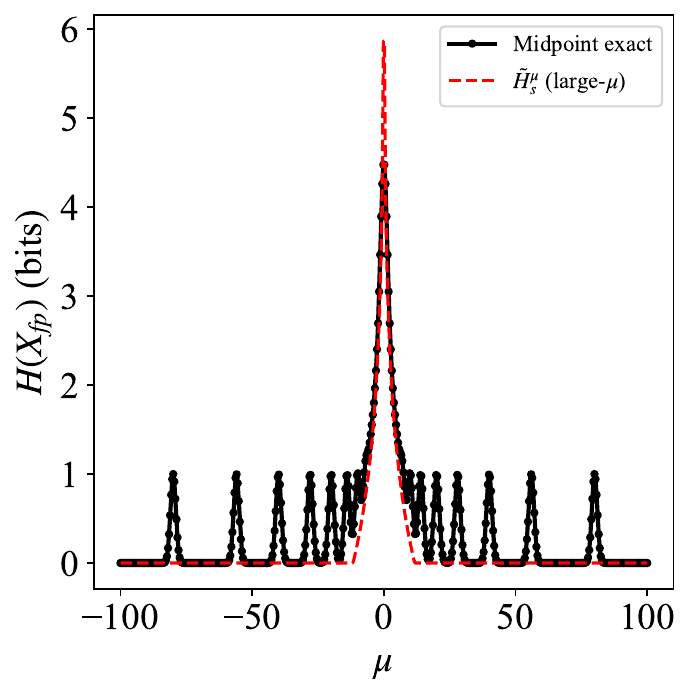}
    \caption{$p=2$, $E=5$.}
    \label{fig:app_mu_p2_E5}
  \end{subfigure}
  \begin{subfigure}{0.32\textwidth}
    \centering
    \includegraphics[width=\linewidth]{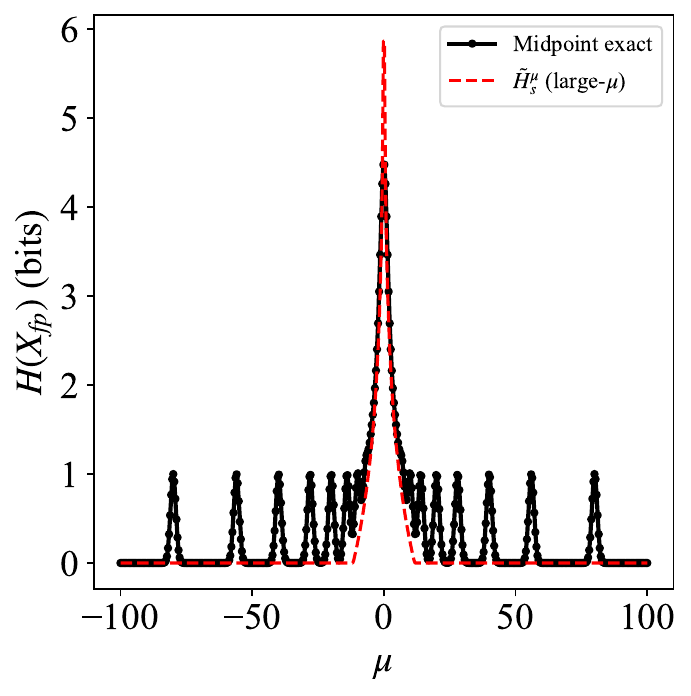}
    \caption{$p=2$, $E=6$.}
    \label{fig:app_mu_p2_E6}
  \end{subfigure}
  \begin{subfigure}{0.32\textwidth}
    \centering
    \includegraphics[width=\linewidth]{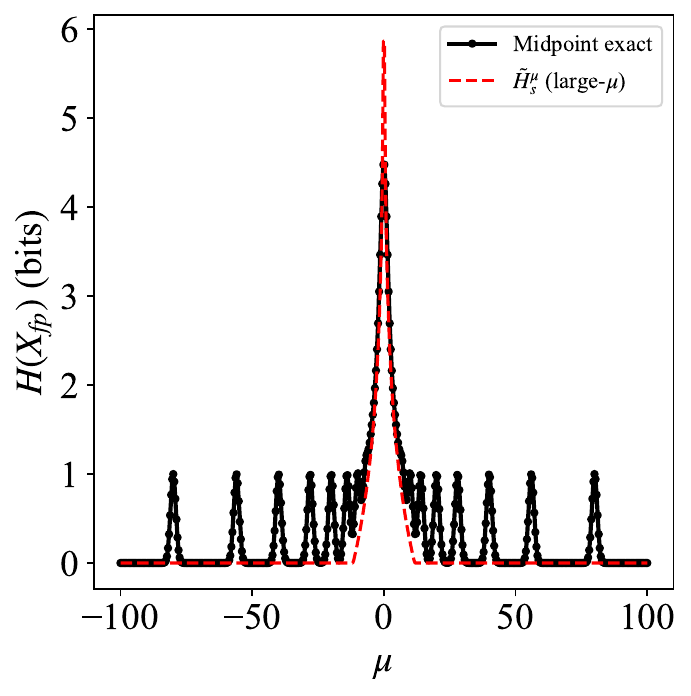}
    \caption{$p=2$, $E=7$.}
    \label{fig:app_mu_p2_E7}
  \end{subfigure}
  \caption{\textit{Exact midpoint-quantized entropy vs.\ mean $\mu$, $p=2$.} Same experiment as Fig.~\ref{fig:app_entropy_vs_mu_p1} with precision $p=2$.}
  \label{fig:app_entropy_vs_mu_p2}
\end{figure}

\begin{figure}[p]
  \centering
  \begin{subfigure}{0.32\textwidth}
    \centering
    \includegraphics[width=\linewidth]{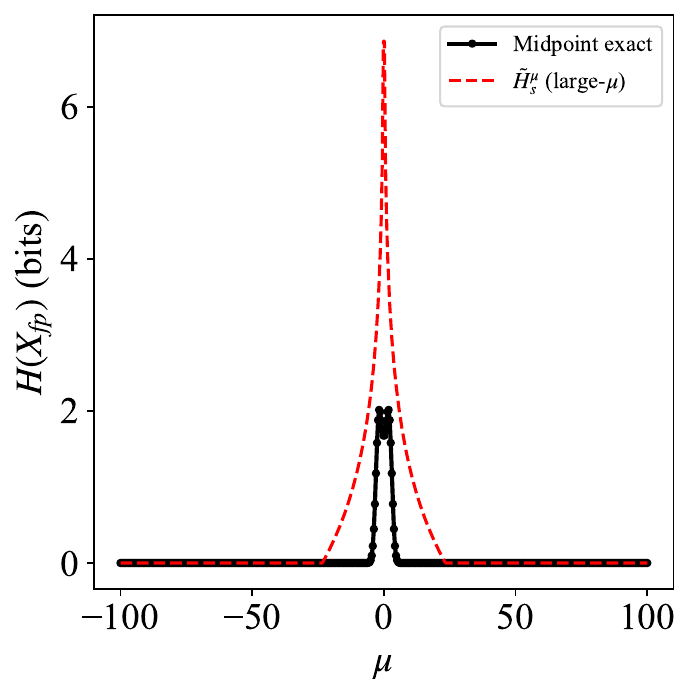}
    \caption{$p=3$, $E=0$.}
    \label{fig:app_mu_p3_E0}
  \end{subfigure}
  \begin{subfigure}{0.32\textwidth}
    \centering
    \includegraphics[width=\linewidth]{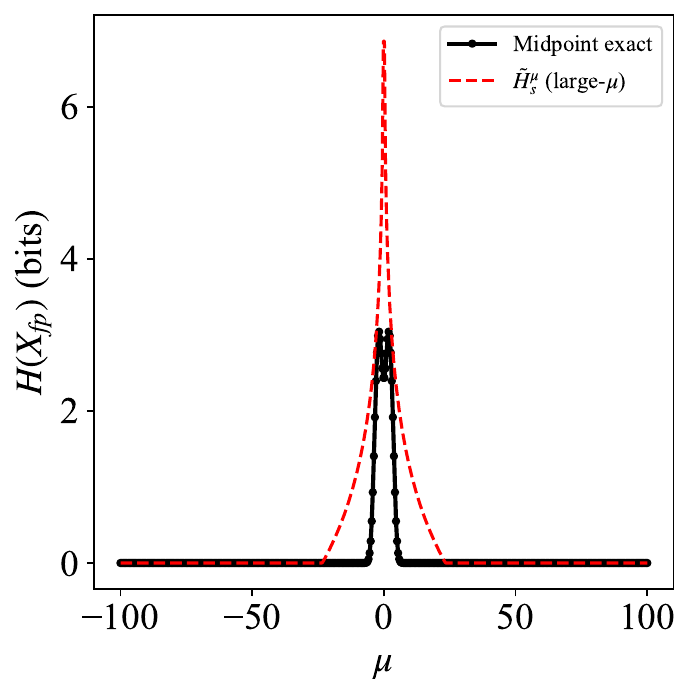}
    \caption{$p=3$, $E=1$.}
    \label{fig:app_mu_p3_E1}
  \end{subfigure}
  \begin{subfigure}{0.32\textwidth}
    \centering
    \includegraphics[width=\linewidth]{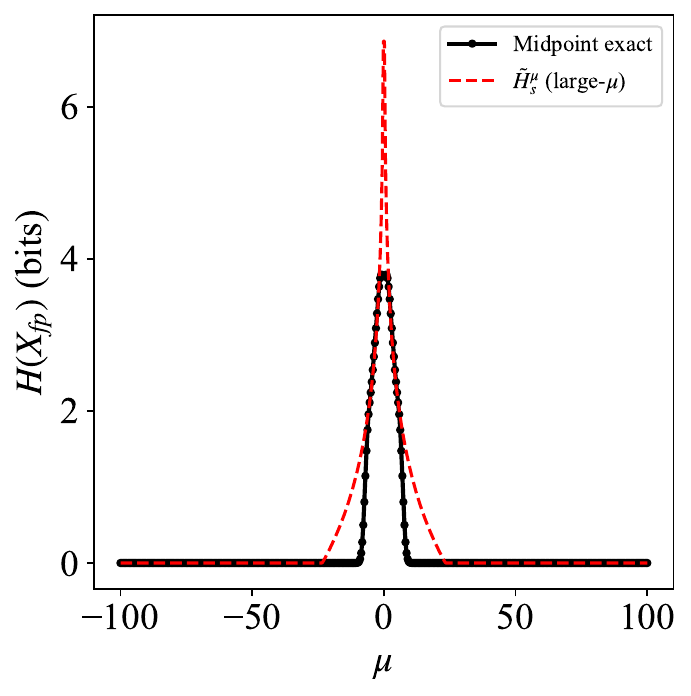}
    \caption{$p=3$, $E=2$.}
    \label{fig:app_mu_p3_E2}
  \end{subfigure}
  \begin{subfigure}{0.32\textwidth}
    \centering
    \includegraphics[width=\linewidth]{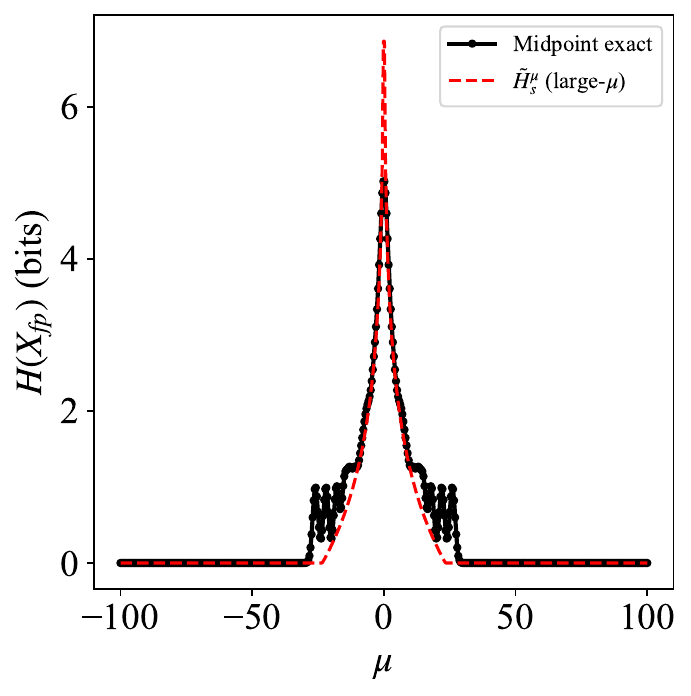}
    \caption{$p=3$, $E=3$.}
    \label{fig:app_mu_p3_E3}
  \end{subfigure}
  \begin{subfigure}{0.32\textwidth}
    \centering
    \includegraphics[width=\linewidth]{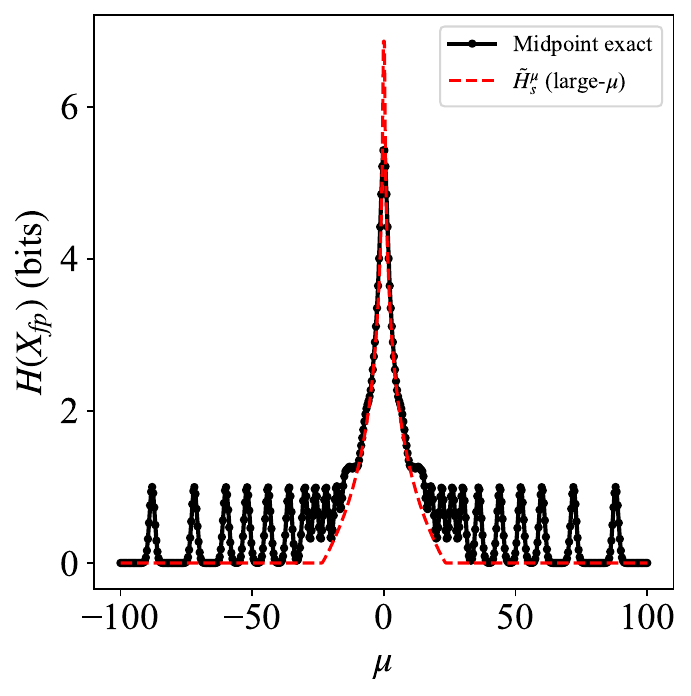}
    \caption{$p=3$, $E=4$.}
    \label{fig:app_mu_p3_E4}
  \end{subfigure}
  \begin{subfigure}{0.32\textwidth}
    \centering
    \includegraphics[width=\linewidth]{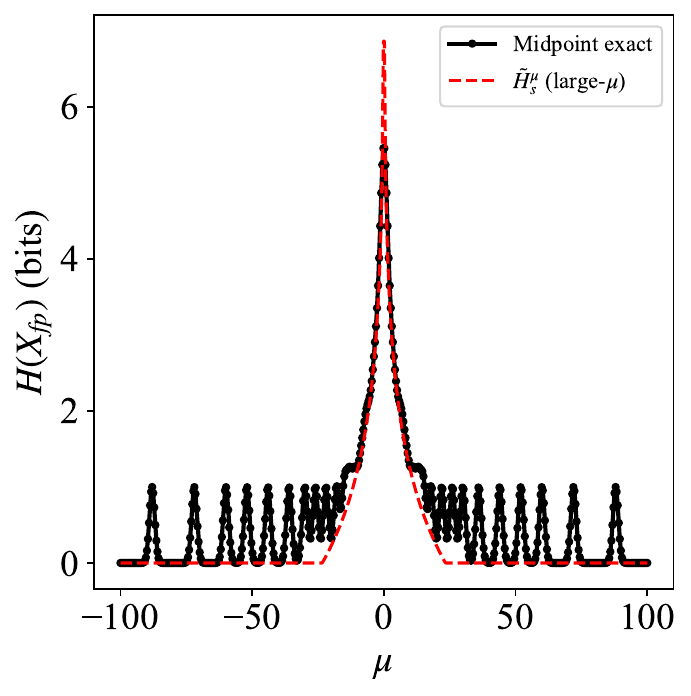}
    \caption{$p=3$, $E=5$.}
    \label{fig:app_mu_p3_E5}
  \end{subfigure}
  \begin{subfigure}{0.32\textwidth}
    \centering
    \includegraphics[width=\linewidth]{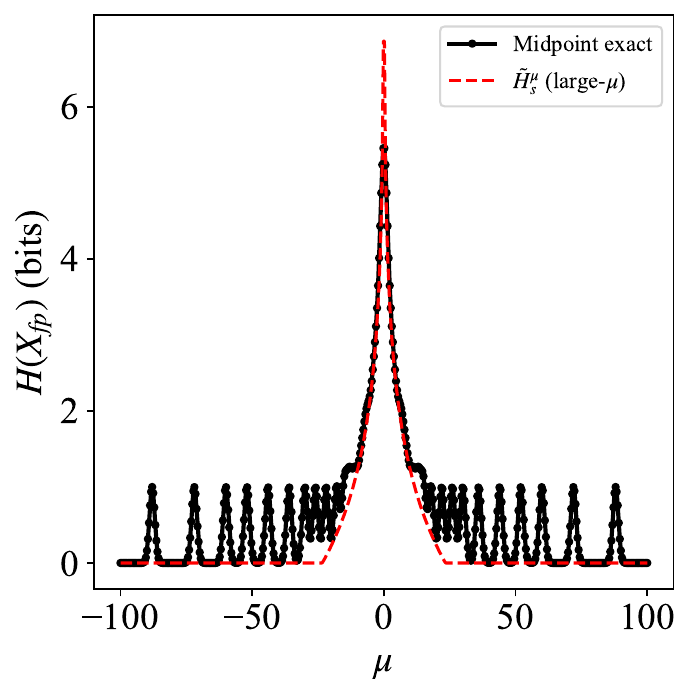}
    \caption{$p=3$, $E=6$.}
    \label{fig:app_mu_p3_E6}
  \end{subfigure}
  \begin{subfigure}{0.32\textwidth}
    \centering
    \includegraphics[width=\linewidth]{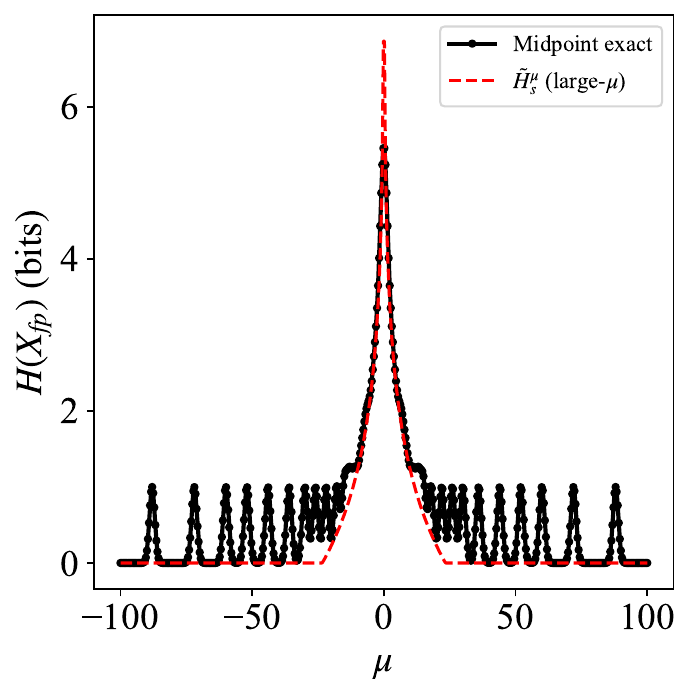}
    \caption{$p=3$, $E=7$.}
    \label{fig:app_mu_p3_E7}
  \end{subfigure}
  \caption{\textit{Exact midpoint-quantized entropy vs.\ mean $\mu$, $p=3$.} Same experiment as Fig.~\ref{fig:app_entropy_vs_mu_p1} with precision $p=3$.}
  \label{fig:app_entropy_vs_mu_p3}
\end{figure}

\begin{figure}[p]
  \centering
  \begin{subfigure}{0.32\textwidth}
    \centering
    \includegraphics[width=\linewidth]{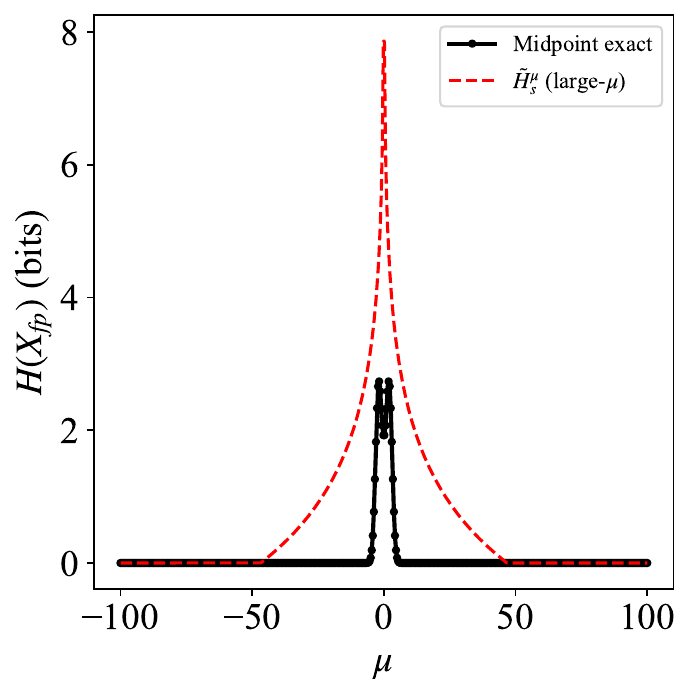}
    \caption{$p=4$, $E=0$.}
    \label{fig:app_mu_p4_E0}
  \end{subfigure}
  \begin{subfigure}{0.32\textwidth}
    \centering
    \includegraphics[width=\linewidth]{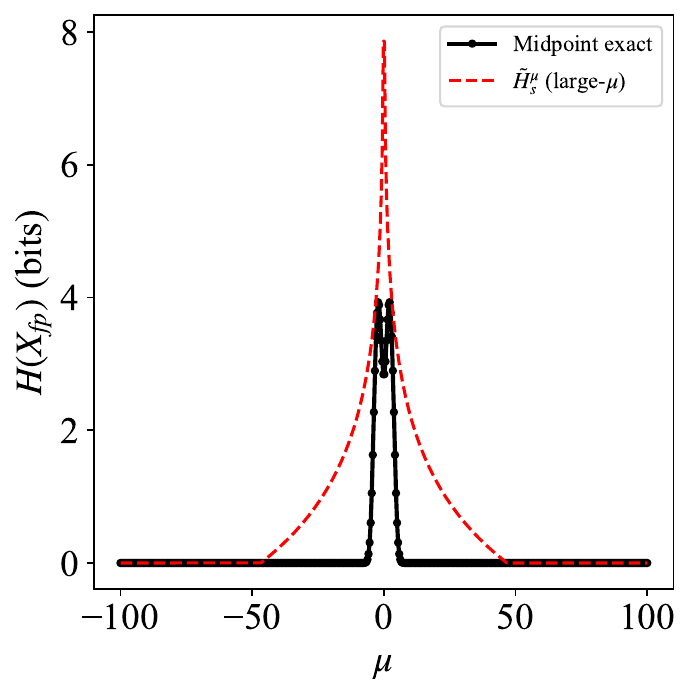}
    \caption{$p=4$, $E=1$.}
    \label{fig:app_mu_p4_E1}
  \end{subfigure}
  \begin{subfigure}{0.32\textwidth}
    \centering
    \includegraphics[width=\linewidth]{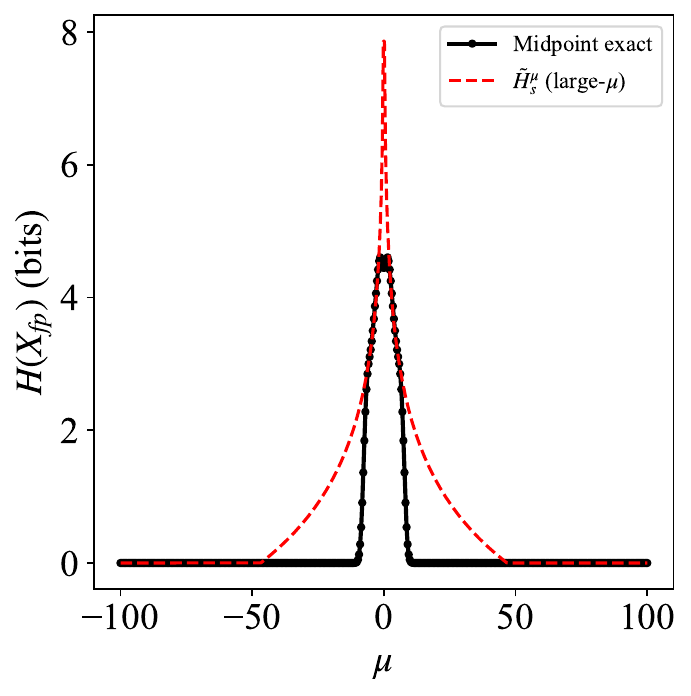}
    \caption{$p=4$, $E=2$.}
    \label{fig:app_mu_p4_E2}
  \end{subfigure}
  \begin{subfigure}{0.32\textwidth}
    \centering
    \includegraphics[width=\linewidth]{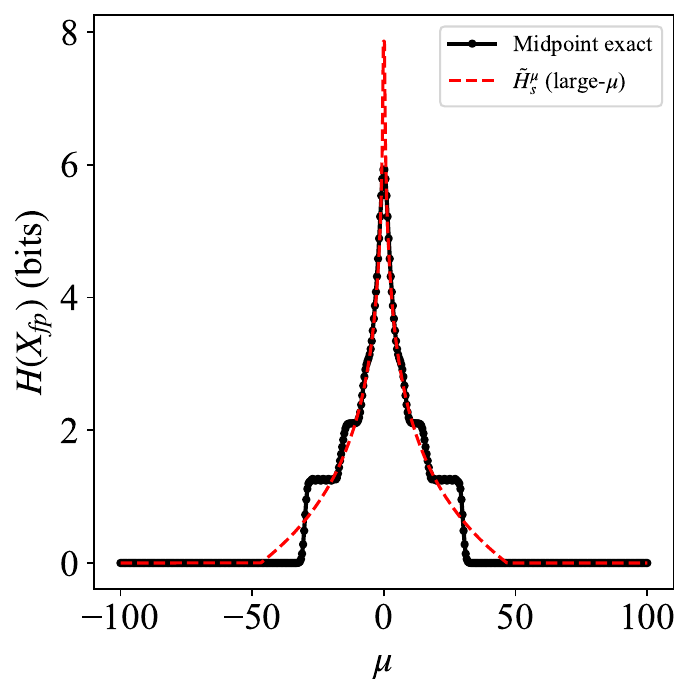}
    \caption{$p=4$, $E=3$.}
    \label{fig:app_mu_p4_E3}
  \end{subfigure}
  \begin{subfigure}{0.32\textwidth}
    \centering
    \includegraphics[width=\linewidth]{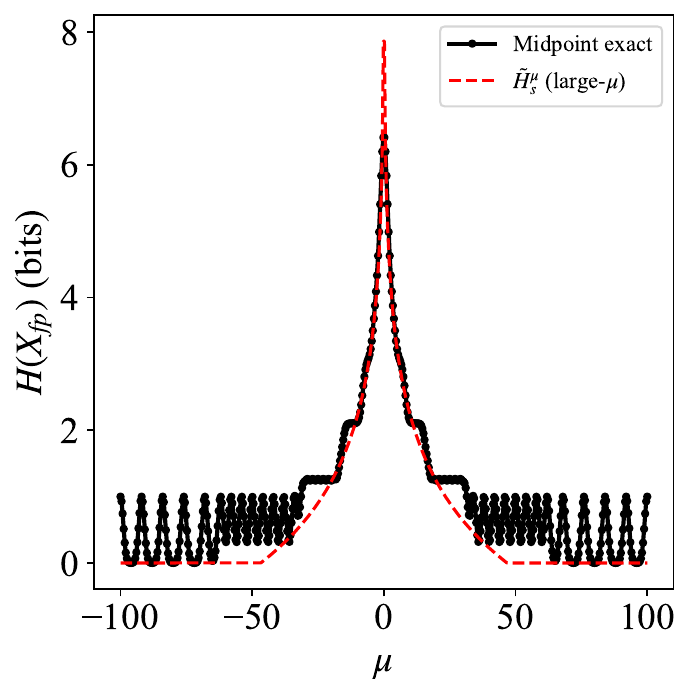}
    \caption{$p=4$, $E=4$.}
    \label{fig:app_mu_p4_E4}
  \end{subfigure}
  \begin{subfigure}{0.32\textwidth}
    \centering
    \includegraphics[width=\linewidth]{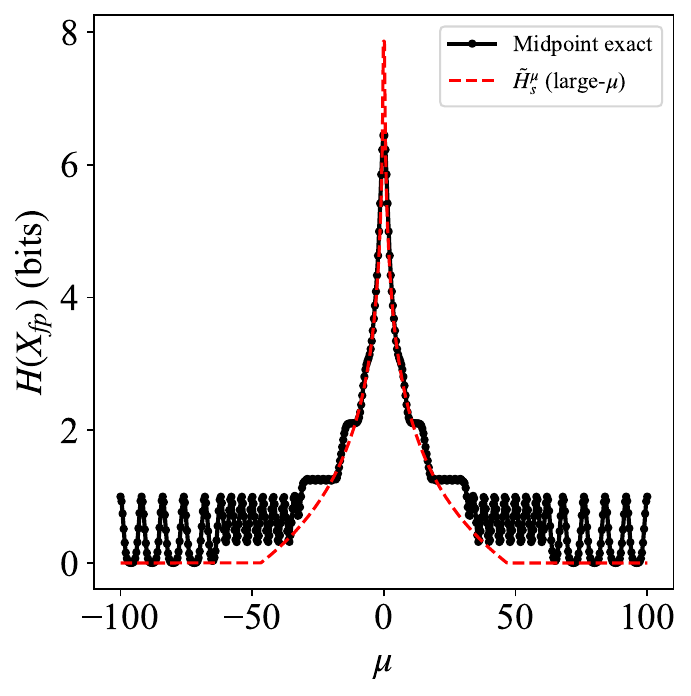}
    \caption{$p=4$, $E=5$.}
    \label{fig:app_mu_p4_E5}
  \end{subfigure}
  \begin{subfigure}{0.32\textwidth}
    \centering
    \includegraphics[width=\linewidth]{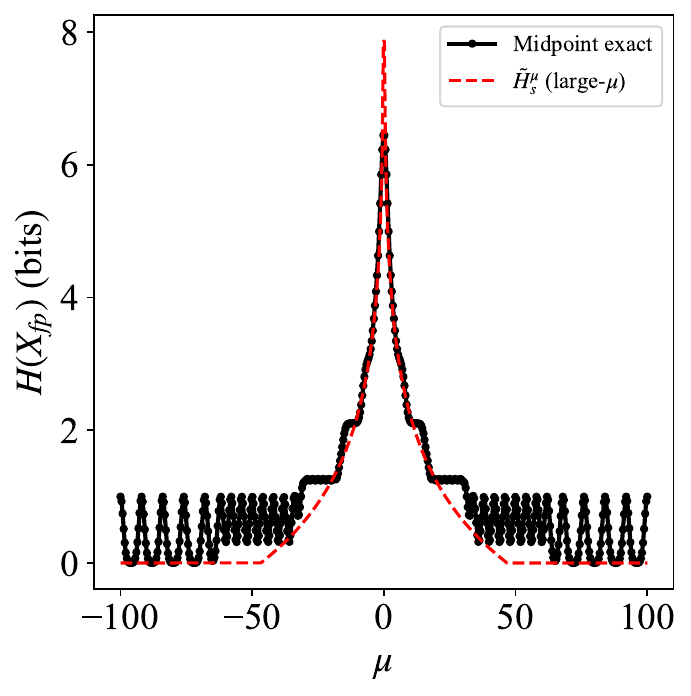}
    \caption{$p=4$, $E=6$.}
    \label{fig:app_mu_p4_E6}
  \end{subfigure}
  \begin{subfigure}{0.32\textwidth}
    \centering
    \includegraphics[width=\linewidth]{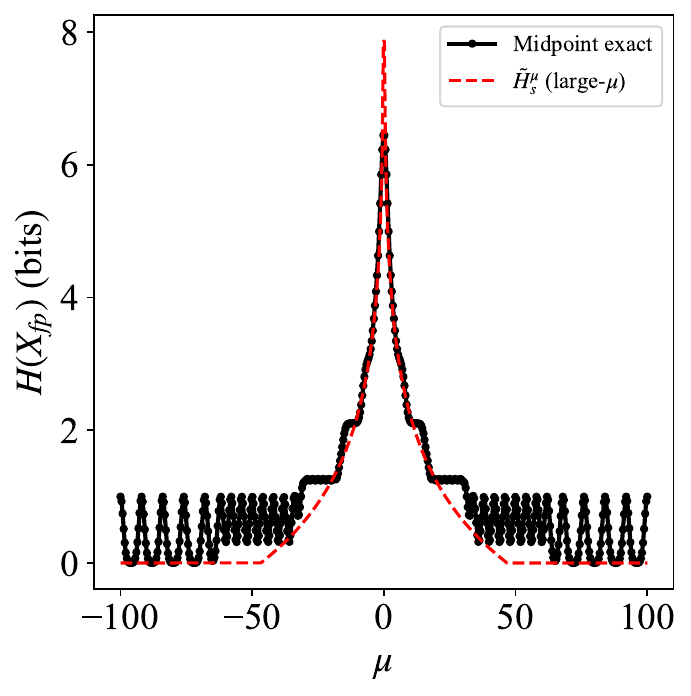}
    \caption{$p=4$, $E=7$.}
    \label{fig:app_mu_p4_E7}
  \end{subfigure}
  \caption{\textit{Exact midpoint-quantized entropy vs.\ mean $\mu$, $p=4$.} Same experiment as Fig.~\ref{fig:app_entropy_vs_mu_p1} with precision $p=4$.}
  \label{fig:app_entropy_vs_mu_p4}
\end{figure}

\begin{figure}[p]
  \centering
  \begin{subfigure}{0.32\textwidth}
    \centering
    \includegraphics[width=\linewidth]{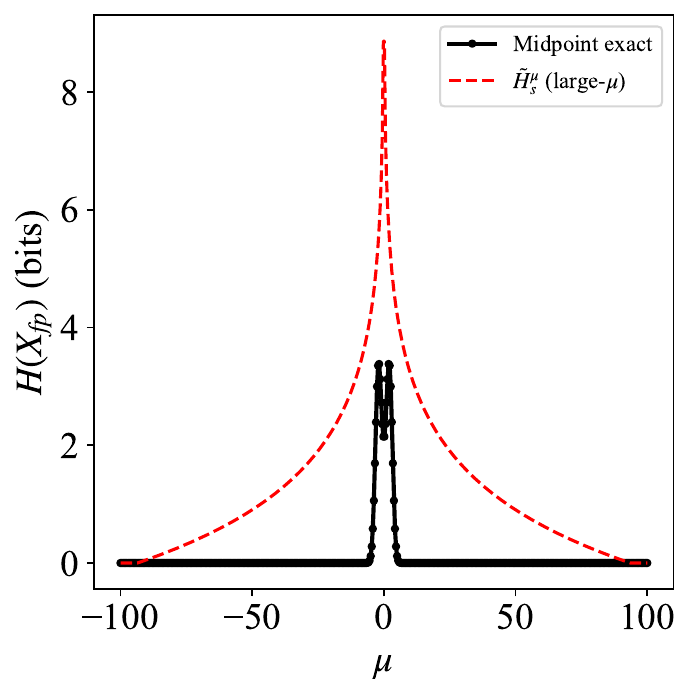}
    \caption{$p=5$, $E=0$.}
    \label{fig:app_mu_p5_E0}
  \end{subfigure}
  \begin{subfigure}{0.32\textwidth}
    \centering
    \includegraphics[width=\linewidth]{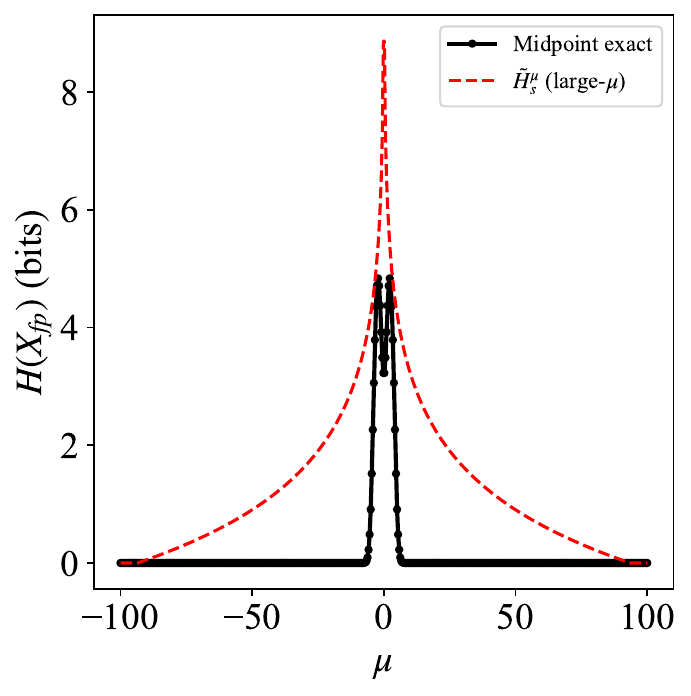}
    \caption{$p=5$, $E=1$.}
    \label{fig:app_mu_p5_E1}
  \end{subfigure}
  \begin{subfigure}{0.32\textwidth}
    \centering
    \includegraphics[width=\linewidth]{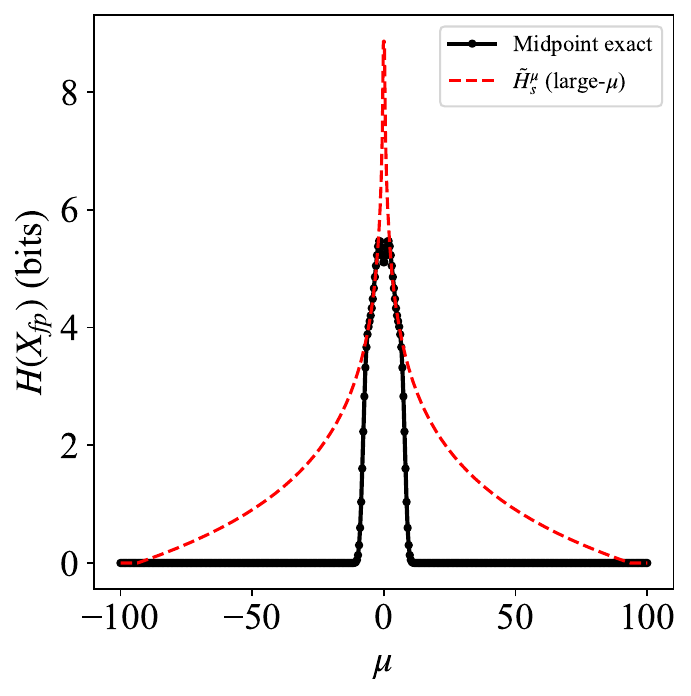}
    \caption{$p=5$, $E=2$.}
    \label{fig:app_mu_p5_E2}
  \end{subfigure}
  \begin{subfigure}{0.32\textwidth}
    \centering
    \includegraphics[width=\linewidth]{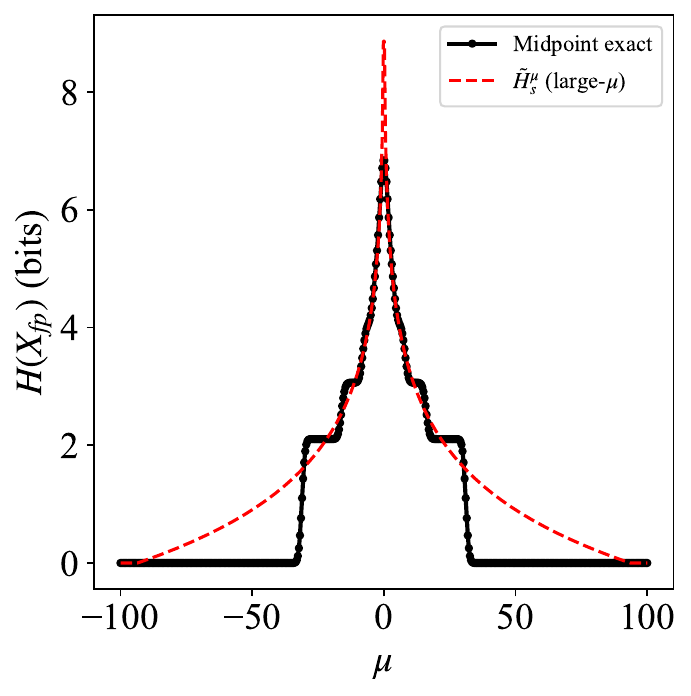}
    \caption{$p=5$, $E=3$.}
    \label{fig:app_mu_p5_E3}
  \end{subfigure}
  \begin{subfigure}{0.32\textwidth}
    \centering
    \includegraphics[width=\linewidth]{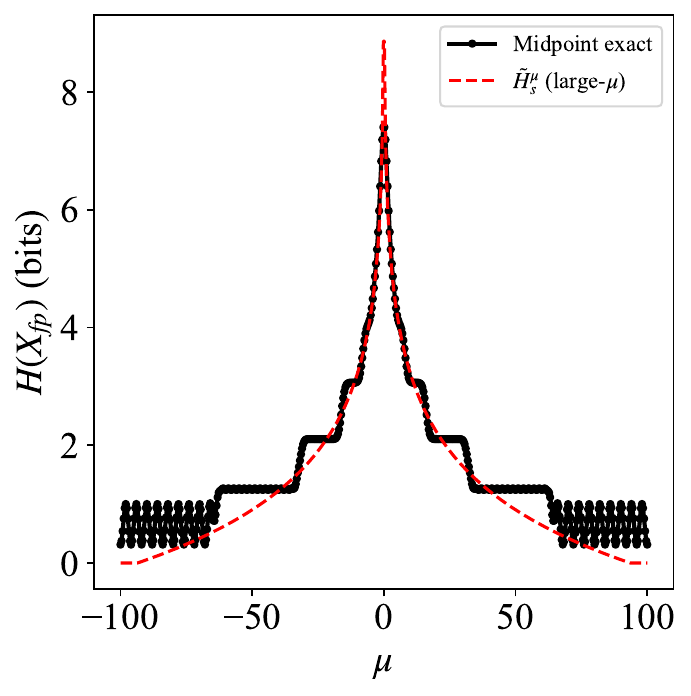}
    \caption{$p=5$, $E=4$.}
    \label{fig:app_mu_p5_E4}
  \end{subfigure}
  \begin{subfigure}{0.32\textwidth}
    \centering
    \includegraphics[width=\linewidth]{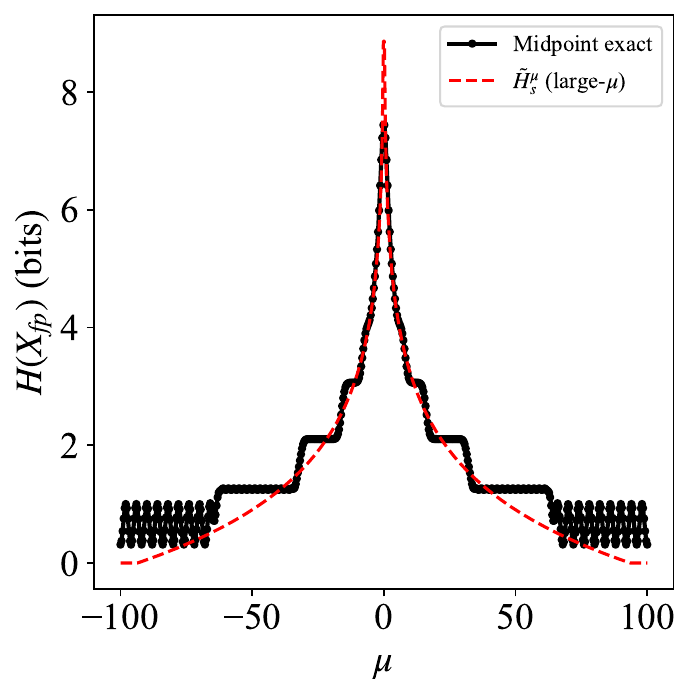}
    \caption{$p=5$, $E=5$.}
    \label{fig:app_mu_p5_E5}
  \end{subfigure}
  \begin{subfigure}{0.32\textwidth}
    \centering
    \includegraphics[width=\linewidth]{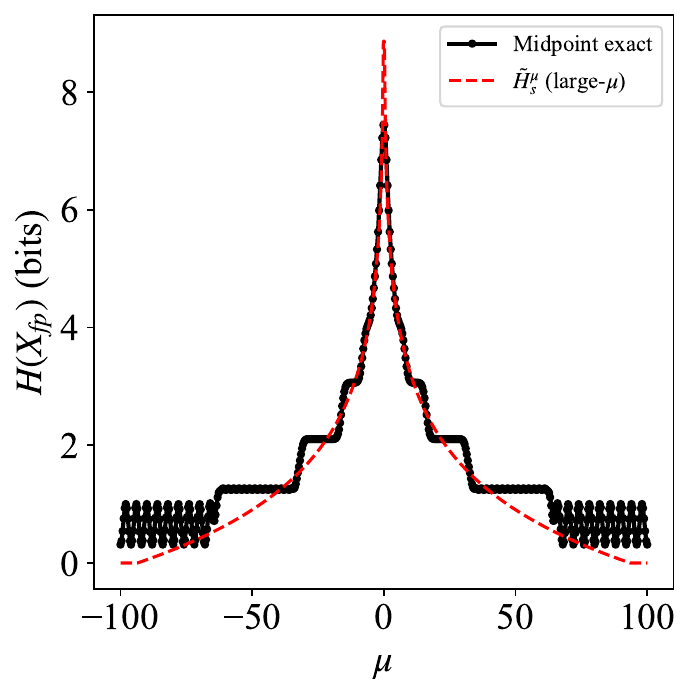}
    \caption{$p=5$, $E=6$.}
    \label{fig:app_mu_p5_E6}
  \end{subfigure}
  \begin{subfigure}{0.32\textwidth}
    \centering
    \includegraphics[width=\linewidth]{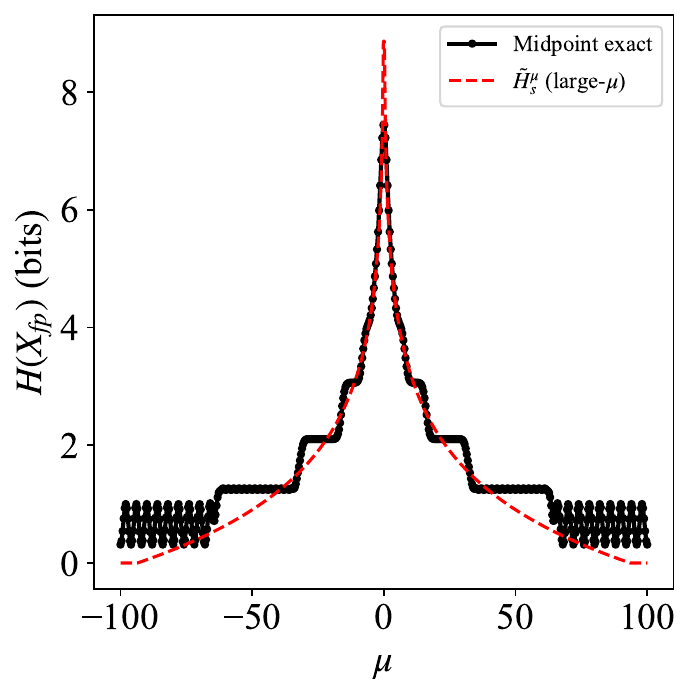}
    \caption{$p=5$, $E=7$.}
    \label{fig:app_mu_p5_E7}
  \end{subfigure}
  \caption{\textit{Exact midpoint-quantized entropy vs.\ mean $\mu$, $p=5$.} Same experiment as Fig.~\ref{fig:app_entropy_vs_mu_p1} with precision $p=5$.}
  \label{fig:app_entropy_vs_mu_p5}
\end{figure}

\begin{figure}[p]
  \centering
  \begin{subfigure}{0.32\textwidth}
    \centering
    \includegraphics[width=\linewidth]{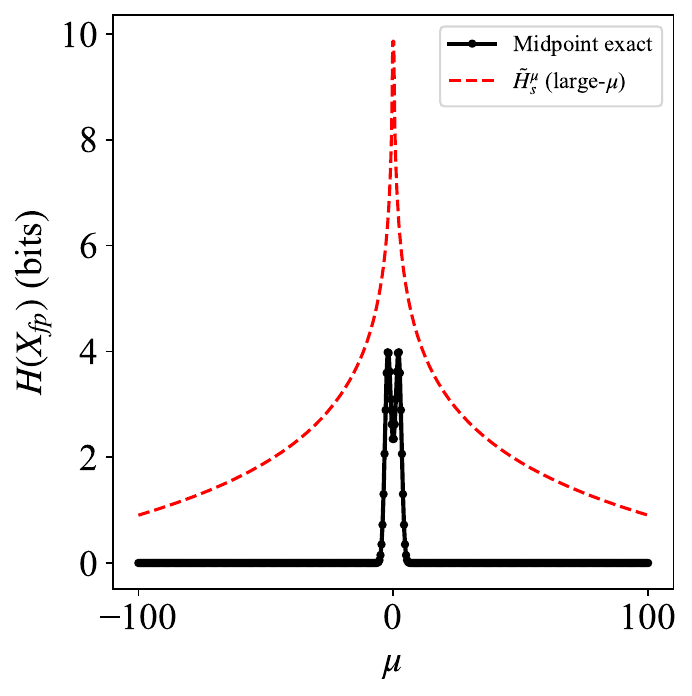}
    \caption{$p=6$, $E=0$.}
    \label{fig:app_mu_p6_E0}
  \end{subfigure}
  \begin{subfigure}{0.32\textwidth}
    \centering
    \includegraphics[width=\linewidth]{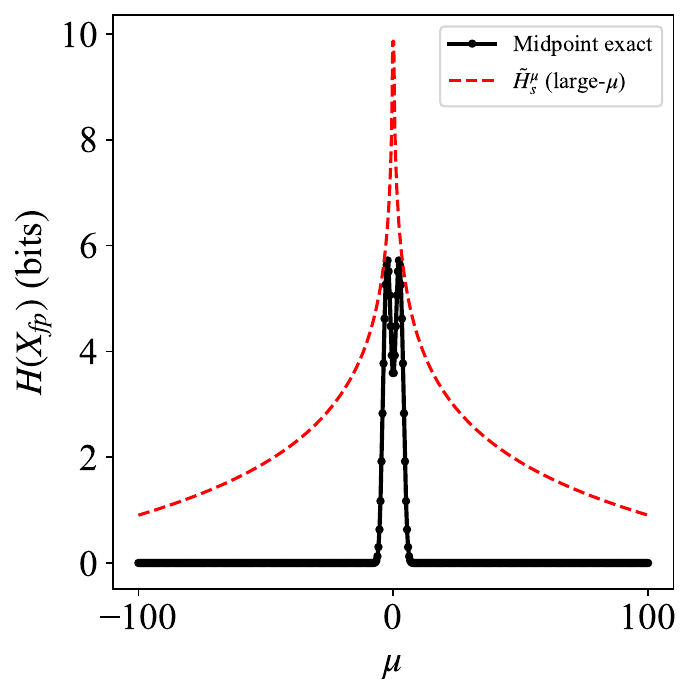}
    \caption{$p=6$, $E=1$.}
    \label{fig:app_mu_p6_E1}
  \end{subfigure}
  \begin{subfigure}{0.32\textwidth}
    \centering
    \includegraphics[width=\linewidth]{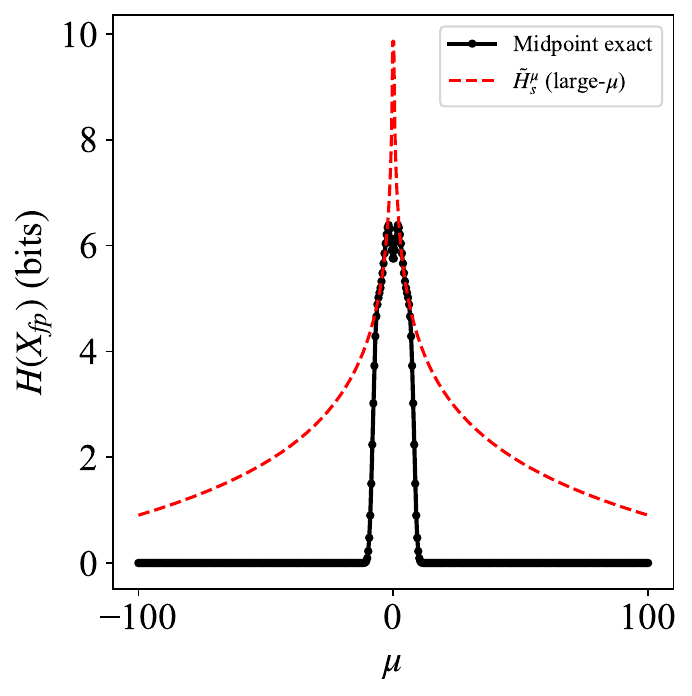}
    \caption{$p=6$, $E=2$.}
    \label{fig:app_mu_p6_E2}
  \end{subfigure}
  \begin{subfigure}{0.32\textwidth}
    \centering
    \includegraphics[width=\linewidth]{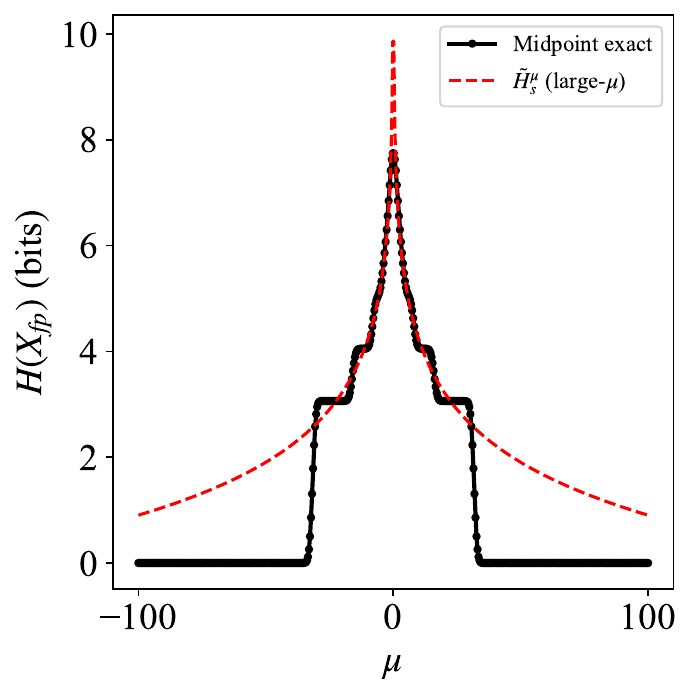}
    \caption{$p=6$, $E=3$.}
    \label{fig:app_mu_p6_E3}
  \end{subfigure}
  \begin{subfigure}{0.32\textwidth}
    \centering
    \includegraphics[width=\linewidth]{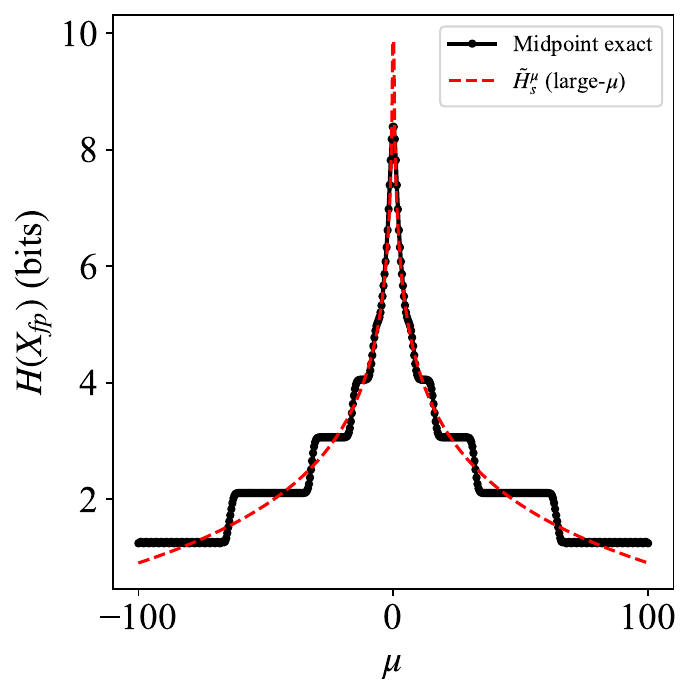}
    \caption{$p=6$, $E=4$.}
    \label{fig:app_mu_p6_E4}
  \end{subfigure}
  \begin{subfigure}{0.32\textwidth}
    \centering
    \includegraphics[width=\linewidth]{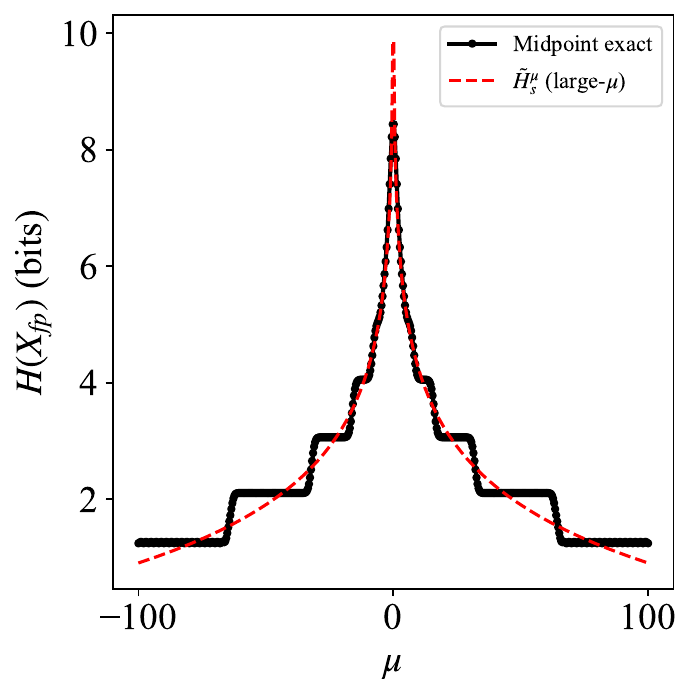}
    \caption{$p=6$, $E=5$.}
    \label{fig:app_mu_p6_E5}
  \end{subfigure}
  \begin{subfigure}{0.32\textwidth}
    \centering
    \includegraphics[width=\linewidth]{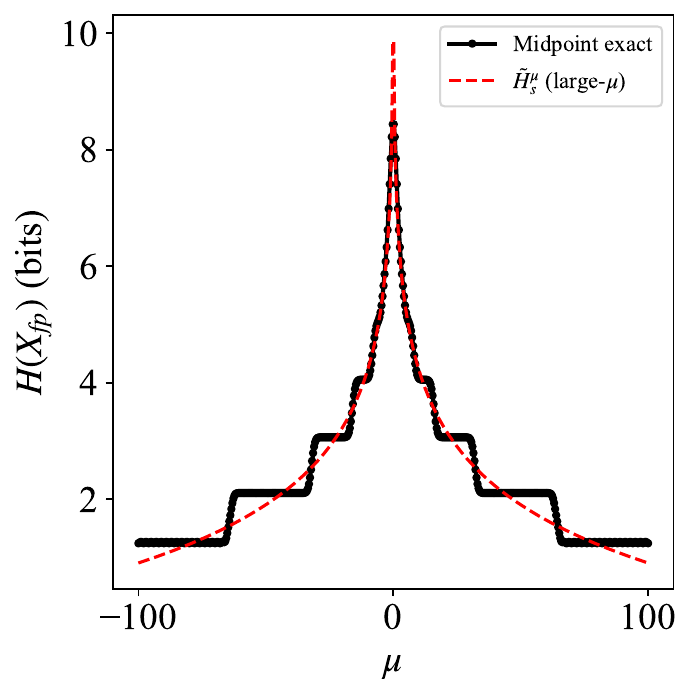}
    \caption{$p=6$, $E=6$.}
    \label{fig:app_mu_p6_E6}
  \end{subfigure}
  \begin{subfigure}{0.32\textwidth}
    \centering
    \includegraphics[width=\linewidth]{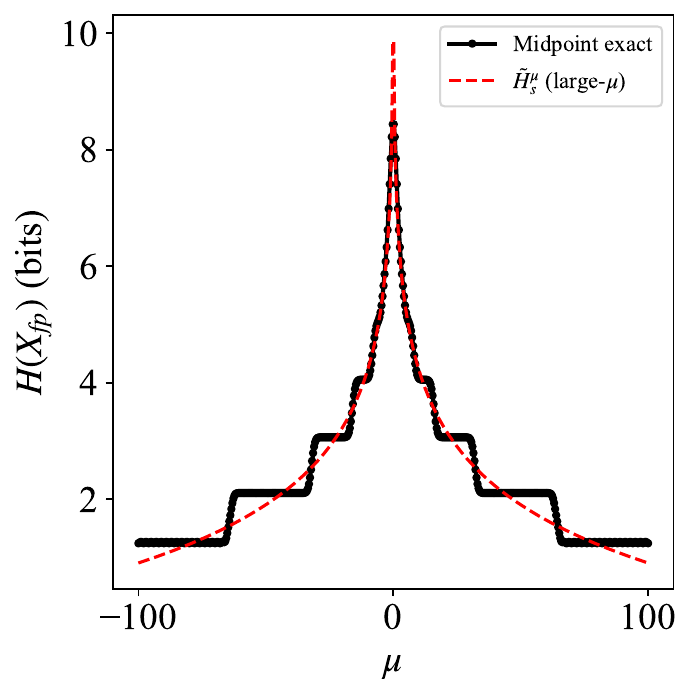}
    \caption{$p=6$, $E=7$.}
    \label{fig:app_mu_p6_E7}
  \end{subfigure}
  \caption{\textit{Exact midpoint-quantized entropy vs.\ mean $\mu$, $p=6$.} Same experiment as Fig.~\ref{fig:app_entropy_vs_mu_p1} with precision $p=6$.}
  \label{fig:app_entropy_vs_mu_p6}
\end{figure}

\begin{figure}[p]
  \centering
  \begin{subfigure}{0.32\textwidth}
    \centering
    \includegraphics[width=\linewidth]{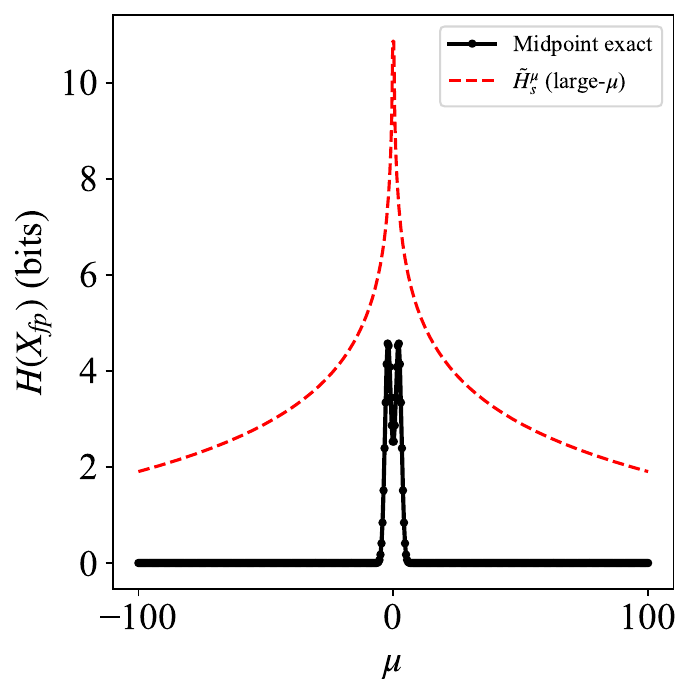}
    \caption{$p=7$, $E=0$.}
    \label{fig:app_mu_p7_E0}
  \end{subfigure}
  \begin{subfigure}{0.32\textwidth}
    \centering
    \includegraphics[width=\linewidth]{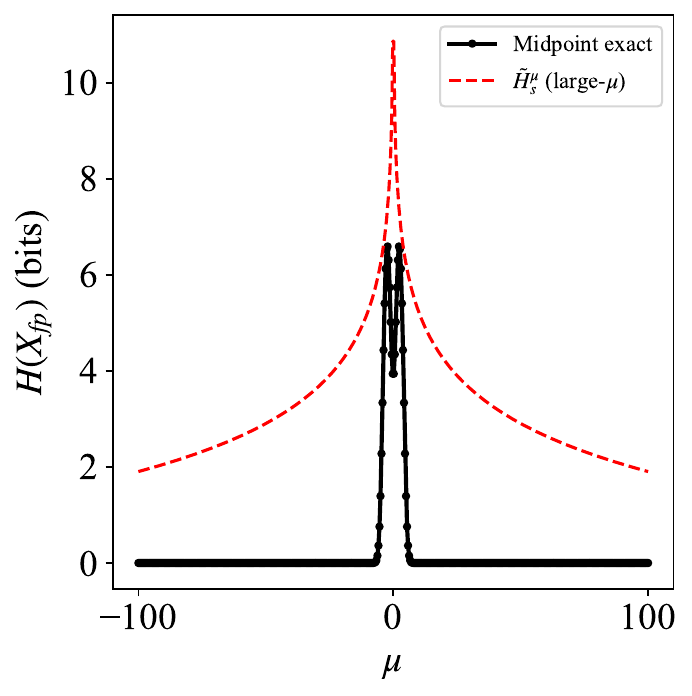}
    \caption{$p=7$, $E=1$.}
    \label{fig:app_mu_p7_E1}
  \end{subfigure}
  \begin{subfigure}{0.32\textwidth}
    \centering
    \includegraphics[width=\linewidth]{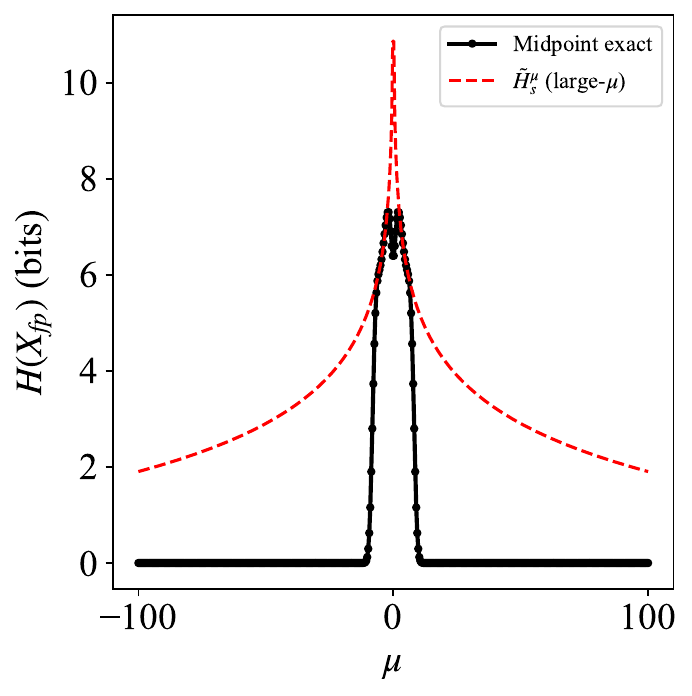}
    \caption{$p=7$, $E=2$.}
    \label{fig:app_mu_p7_E2}
  \end{subfigure}
  \begin{subfigure}{0.32\textwidth}
    \centering
    \includegraphics[width=\linewidth]{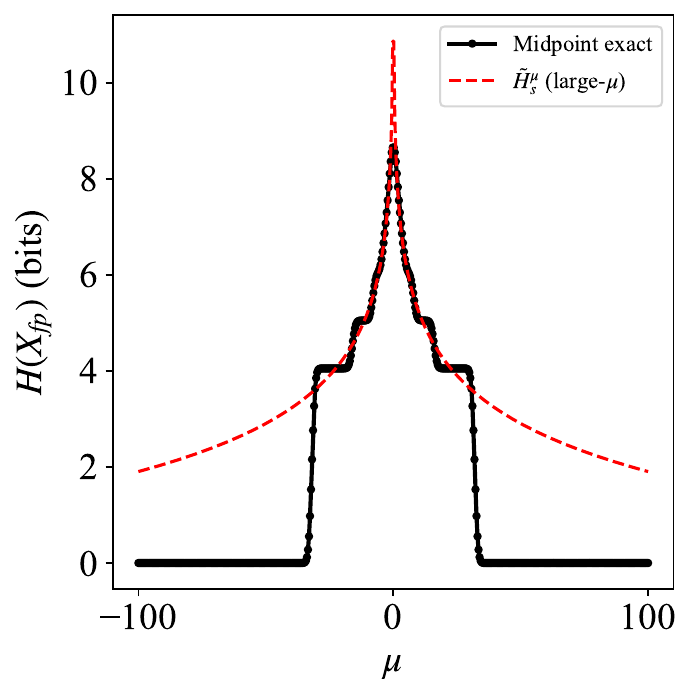}
    \caption{$p=7$, $E=3$.}
    \label{fig:app_mu_p7_E3}
  \end{subfigure}
  \begin{subfigure}{0.32\textwidth}
    \centering
    \includegraphics[width=\linewidth]{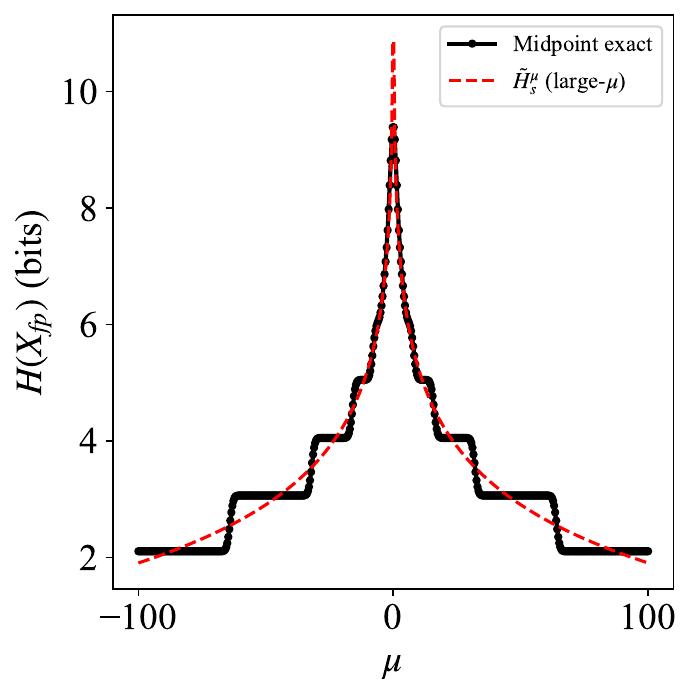}
    \caption{$p=7$, $E=4$.}
    \label{fig:app_mu_p7_E4}
  \end{subfigure}
  \begin{subfigure}{0.32\textwidth}
    \centering
    \includegraphics[width=\linewidth]{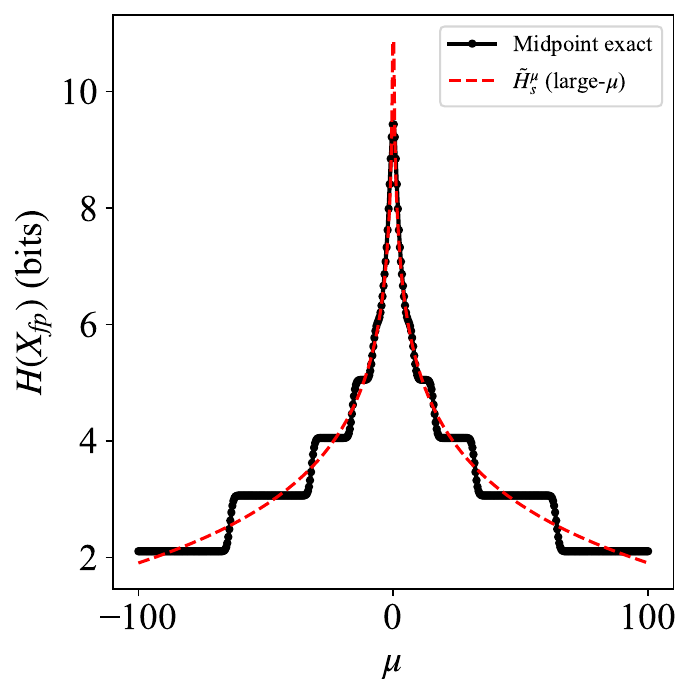}
    \caption{$p=7$, $E=5$.}
    \label{fig:app_mu_p7_E5}
  \end{subfigure}
  \begin{subfigure}{0.32\textwidth}
    \centering
    \includegraphics[width=\linewidth]{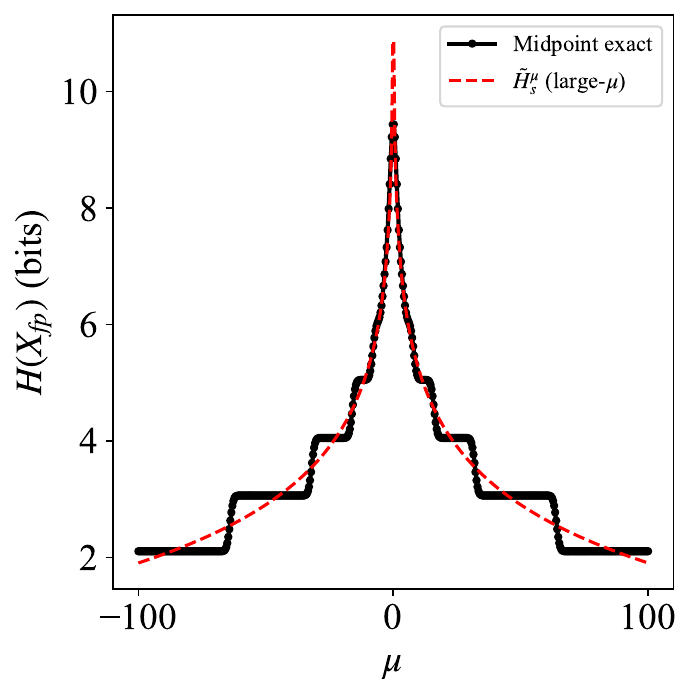}
    \caption{$p=7$, $E=6$.}
    \label{fig:app_mu_p7_E6}
  \end{subfigure}
  \begin{subfigure}{0.32\textwidth}
    \centering
    \includegraphics[width=\linewidth]{code/quantized_entropy_midpoint/mu_p_E/entropy_mu_midpoint_sigma_1.0_p_7_E_7.pdf}
    \caption{$p=7$, $E=7$.}
    \label{fig:app_mu_p7_E7}
  \end{subfigure}
  \caption{\textit{Exact midpoint-quantized entropy vs.\ mean $\mu$, $p=7$.} Same experiment as Fig.~\ref{fig:app_entropy_vs_mu_p1} with precision $p=7$.}
  \label{fig:app_entropy_vs_mu_p7}
\end{figure}

\end{document}